%% file: arxiv.tex
\documentclass[11pt]{article}

\usepackage{latexsym,graphicx,amsmath,amssymb,enumerate}
\usepackage[margin=1in]{geometry}
\usepackage{url}
\usepackage{cite}
\usepackage{amsthm}
\usepackage{thmtools}
\usepackage{thm-restate}

\usepackage[dvipsnames,usenames]{color}
\usepackage[colorlinks=true,urlcolor=Blue,citecolor=Green,linkcolor=BrickRed]
{hyperref}
\usepackage[usenames,dvipsnames]{xcolor}
\usepackage{subcaption}

\declaretheorem[name=Theorem,numberwithin=section]{theorem}
\declaretheorem[name=Lemma,numberlike=theorem]{lemma}
\declaretheorem[name=Corollary,numberlike=theorem]{corollary}

\def\bd{{\partial}}

\def\VD{{\rm VD}}
\def\Vor{{ \rm Vor}}
\def\wt{{ \omega}}

\def\pre{{\rm pre}}
\def\Dbisr{{{\bf \Delta}_\beta^r}}
\def\H{{h}}
\def\NH{{t}}

\def\fullver{}
\ifdefined \fullver 
   \def\fullqed{}
\else
   \def\fullqed{\qed}
\fi

\begin{document}

\title{Voronoi diagrams on planar graphs, \\ and computing the diameter in deterministic $\tilde{O}(n^{5/3})$ time\thanks{A preliminary version of this paper appeared in SODA 2018.}}

\author{
Pawe{\l} Gawrychowski\thanks{University of Wroc{\l}aw, Poland, \texttt{gawry@cs.uni.wroc.pl}. }
\and Haim Kaplan\thanks{Tel Aviv University, Israel, \texttt{\{haimk,michas\}@tauex.tau.ac.il}. 
Work by Haim Kaplan partially supported by the Israel Science Foundation (grants No. 1841/14 and 1595/19), and by GIF (grants No. 1161 and 1367). 
Work by Micha Sharir partially supported by the Israel Science Foundation (grants No. 892/13 and 260/18), by GIF (grant No. 1367), by Len Blavatnik and the Blavatnik Research Fund in Computer Science at Tel Aviv University, by the Israeli Centers of Research Excellence (I-CORE) program (center  No. 4/11), and by the Hermann Minkowski-MINERVA Center for Geometry at Tel Aviv University.}
\and Shay Mozes\thanks{Interdisciplinary Center Herzliya, Israel, \texttt{smozes@idc.ac.il}. Partially supported by the Israel Science Foundation (grants No. 794/13 and 592/17).}
        \and Micha Sharir$^\ddagger$
        \and Oren Weimann\thanks{University of Haifa, Israel, \texttt{oren@cs.haifa.ac.il}. Partially supported by the Israel Science Foundation (grants No. 794/13 and 592/17).}
        }

        \date{}

    \maketitle

    \begin{abstract}
    
    We present an explicit and efficient construction of additively weighted
Voronoi diagrams on planar graphs. Let $G$ be a planar graph with $n$ vertices and $b$ sites that lie on a constant number of faces. We show how to preprocess $G$ in $\tilde O(nb^2)$ time\footnote{The $\tilde O$ notation hides polylogarithmic factors.} so that one can compute any additively weighted Voronoi diagram for these sites in $\tilde O(b)$ time. 

We use this construction to compute the diameter of a directed planar graph
with real arc lengths in $\tilde{O}(n^{5/3})$ time. This improves the recent breakthrough result of Cabello (SODA'17), both by improving the running time (from $\tilde{O}(n^{11/6})$), and by providing a deterministic algorithm. It is in fact the first truly subquadratic {\em deterministic} algorithm for this problem. 
Our use of Voronoi diagrams to compute the diameter follows that of Cabello, but he used abstract Voronoi diagrams, which makes his diameter algorithm more involved, more expensive, and randomized. 

As in Cabello's work, our algorithm can compute, for every vertex $v$, both the farthest vertex from $v$ (i.e., the eccentricity of $v$), and the sum of distances from $v$ to all other vertices. 
Hence, our algorithm can also compute the radius, median, and Wiener index (sum of all pairwise distances) of a planar graph within the same time bounds.
Our construction of Voronoi diagrams for planar graphs is of independent interest. 

    \end{abstract}

\section{Introduction}

Computing the diameter of a (directed, weighted) graph (the largest distance between a pair of vertices) is a fundamental problem in graph algorithms, with numerous research papers studying its complexity.

For general graphs, the current fastest way to compute the diameter is by computing all pairs of shortest paths (APSP) between its vertices. Computing APSP can be done in cubic $O(n^3)$ time, using  the classical Floyd-Warshall algorithm~\cite{Flo62,War62}. Following a long line of improvements by polylogarithmic factors~\cite{Fre76,Dob90,Tak91,Han04,Zwick2004,Takaoka05,Chan08,Han2006,Chan10,HanT12}, the current fastest APSP algorithm is the  $O(n^3/2^{\Omega(\log n)^{1/2}})$-time algorithm by Williams~\cite{RyanWilliams}. However, no truly subcubic, i.e.\ $O(n^{3-\varepsilon})$-time algorithm, for any fixed $\varepsilon>0$, is known for either APSP or the problem of
computing the diameter, referred to as {\sc diameter}. We also do not know if {\sc diameter} is as hard as APSP (i.e., if a truly subcubic algorithm for {\sc diameter} implies a truly subcubic algorithm for APSP). This is in contrast to many other related problems that were recently shown to be as hard as APSP~\cite{TED2018,BackursTzamos,BackursDikkalaTzamos,AbboudVassilevskaFOCS14,AbboudLewi2013,CountingWeightedSubgraphs,AbboudPlanar,RodittyZwick,AbboudGrandoniVassilevska,VW10,AmirVassilevskaYu} (including the problem of computing the  radius of a graph~\cite{AbboudGrandoniVassilevska}).

For sparse graphs (with $\tilde O(n)$ edges), the problem is better understood. We can solve APSP (and thus {\sc diameter}) in $\tilde O(n^2)$ time by running a single-source shortest path algorithm from every vertex. There is clearly no truly subquadratic $O(n^{2-\varepsilon})$ algorithm for APSP (as the output size is $\Omega (n^2)$). However,  Chan~\cite{Chan06} showed that, for undirected unweighted sparse graphs, APSP can be represented and computed in $O(n^2/\log n)$ time.  Interestingly, assuming the strong exponential time hypothesis (SETH) \cite{IPZ01}, there is also no truly subquadratic algorithm for {\sc diameter}. In fact,  even distinguishing between diameter 2 and 3  requires quadratic time assuming SETH~\cite{RodittyW13}. This holds even for  undirected, unweighted graphs with treewidth $O(\log n)$. For graphs of bounded treewidth, the diameter can be computed in near-linear time~\cite{AWW16} (see also~\cite{Husfeldt16,Eppstein} for algorithms with time bounds that depend on the value of the diameter itself). Near-linear time algorithms were developed for many other restricted graph families, see e.g.~\cite{Chordal,PlaneTriangulations,CactusGraphs,OuterplanarGraphs,IntervalGraphs,HyperbolicGeodesic,Euclidian,SMAWK,Geodesic,DistanceHereditaryGraphs}.

For planar graphs, we know that the diameter can in fact be computed faster than APSP. This was illustrated by
Wulff-Nilsen who gave an algorithm for computing
the diameter of unweighted, undirected planar graphs in
$O(n^2 \log\log n/ \log n)$ time~\cite{WN08} and of weighted directed planar graphs in
$O(n^2 (\log\log n)^4/ \log n)$ time~\cite{WN10}. The question of whether a truly subquadratic ($O(n^{2-\varepsilon})$-time) algorithm exists (even for undirected, unweighted planar graphs) was a major open problem in planar graph algorithms. In SODA'17, a breakthrough result of Cabello~\cite{Cabello} showed that the diameter of weighted directed planar graphs can be computed in truly subquadratic $\tilde O(n^{11/6})$ expected time. Whether the exponent $11/6$ could be substantially reduced became a main open problem.

\medskip
\noindent
{\bf Voronoi diagrams on planar graphs.}
The breakthrough in Cabello's work is his novel use of Voronoi diagrams in planar graphs.\footnote{Using Voronoi diagrams for algorithms in planar graphs was done before (see e.g.~\cite{MarxP15,Verdiere10}), but never for the purpose of obtaining polynomial factor speedups.}
An $r$-{\em division}~\cite{F87} of a planar graph $G$ is a decomposition
of $G$ into $\Theta(n/r)$ pieces,
each of them with $O(r)$ vertices and $O(\sqrt{r})$ boundary vertices (vertices shared with other pieces). A useful property in $r$-divisions is when the boundary vertices of each piece lie on a constant number of faces of the piece (called {\em holes}). An $O(n\log n)$ (deterministic) algorithm for computing an $r$-division with this property was given in~\cite{FR06}  and improved to $O(n)$ in~\cite{KMS13}.
Unsurprisingly, it turns out that the difficult case for computing the diameter is when the farthest pair of vertices lie in different pieces.

Consider some source vertex $v_0$ outside of some piece $P$. For every boundary vertex $u$ of $P$, let $\wt(u)$ denote the $v_0$-to-$u$ distance in $G$. The {\em additively weighted Voronoi diagram} of $P$ with respect to $\wt(\cdot)$ is a partition of the vertices of $P$ into pairwise disjoint sets (Voronoi cells), each associated with a unique boundary vertex (site) $u$. The vertices in the cell
$\Vor(u)$ of $u$ are all the vertices $v$ of $P$ such that $u$ minimizes the quantity $\wt(u)$ plus the length of the $u$-to-$v$ shortest path inside $P$ (i.e., $u$ is the last boundary vertex of $P$ on the $v_0$-to-$v$ shortest path in $G$). The {\em boundary} of a cell $\Vor(u)$ consists of all edges of $P$ that have exactly one endpoint in $\Vor(u)$. For example, in a Voronoi diagram of just two sites $u$ and $v$, the boundary of the cell $\Vor(u)$ is a $uv$-cut and is therefore a cycle in the dual graph. We call this cycle the $uv$-{\em bisector}. The {\em complexity} of a Voronoi diagram is defined as the number of faces of $P$ that contain vertices belonging to three or more Voronoi cells. The complexity of a Voronoi {\em cell} $\Vor(u)$ is the number of faces of $P$ that contain vertices of $\Vor(u)$ and of at least two more Voronoi cells. 

For every $v_0$, computing the farthest vertex from $v_0$ in $P$ thus boils down to computing, for each site $u$, the farthest vertex from $u$ in $\Vor(u)$, and returning the maximum of these quantities.
The main challenges with this approach are: (1) to efficiently compute the Voronoi diagram of a piece $P$ (that is, to identify the partition of the vertices of $P$ into Voronoi cells), with respect to site weights equal to the $v_0$-to-site distances, and (2) to find the maximum site-to-vertex distance in each cell. To appreciate these issues, note that we need to perform tasks (1) and (2) much faster than $O(|P|)=O(r)$ time! Otherwise, since there are $n$ possible sources $v_0$, and $\Theta(n/r)$ pieces, the overall running time would be $\Omega(n^2)$. Thus, even task (2), which is trivial to implement  by simply inspecting all vertices of $P$, requires a much more sophisticated approach to make the overall algorithm subquadratic.

Voronoi diagrams of geometric objects of many kinds, and with respect to different metrics, have been extensively studied (see, e.g., \cite{Aurenhammer}). One of the basic approaches for computing Voronoi diagrams is divide and conquer. Using this approach, we split the objects (sites) into two groups, compute the Voronoi diagram of each group recursively, and then somehow merge the diagrams. Using this approach,
Shamos amd Hoey \cite{Shamos} gave the first $O(n\log n)$ deterministic algorithm for computing the Voronoi diagram of $n$ points in the plane with respect to the Euclidean metric. Since then other approaches, such as randomized incremental construction, have also been developed.

\medskip
\noindent
{\bf Cabello's approach.}
For every source $v_0$ and every piece $P$, Cabello uses the ``heavy hammer'' of {\em abstract Voronoi diagrams} in order to compute the Voronoi diagram of $P$ with respect to the $v_0$-to-boundary distances as additive weights.
Klein \cite{Klein89}  introduced  the abstract Voronoi diagram framework, trying to abstract the properties of objects and  metric that make the computation of Voronoi diagrams efficient.
Cabello uses (as a black box) the randomized construction of abstract Voronoi diagrams by Klein, Mehlhorn and Meiser~\cite{KleinMM93} (see also~\cite{KleinLN09}). This construction, when applied to our special setup, is highly efficient, requiring only
 $\tilde O(\sqrt{r})$ time. However, in order to use it, certain requirements must be met.

 First, it requires knowing in advance the Voronoi diagrams induced by every subset of four sites (boundary vertices). Cabello addresses this by first showing that the planar case only requires the Voronoi diagrams for any triple of sites to be known in advance.
The Voronoi diagram of each triple is composed of segments of bisectors.
The situation is further complicated because each pair of sites can have many bisectors, depending on (the difference between) the weights of the sites.
Cabello shows that for each pair of sites, there are only $O(r)$ different bisectors over all possible weight differences for this pair.
His algorithm computes each of these bisectors separately in $O(r)$ time, so computing all the bisectors for a single pair of sites takes $O(r^2)$ time. This yields a procedure that precomputes all 3-site weighted Voronoi diagrams from the bisectors of pairs of sites,
in $\tilde{O}(r^{7/2})$ time. 

 The second requirement is that all the sites must lie on the boundary of a single hole. However, from the properties of $r$-divisions, the boundary vertices of $P$ may lie on several (albeit a constant number of) holes. Overcoming this is a significant technical difficulty in Cabello's paper. He achieves this by allowing the bisectors to venture through the subgraphs enclosed by the holes (i.e., through other pieces).  Informally, Cabello's algorithm ``fills up'' the holes in order to apply the mechanism of abstract Voronoi diagrams.

\medskip
\noindent
{\bf Our result.}
We present a deterministic $\tilde{O}(n^{5/3})$-time algorithm for computing the
diameter of a directed planar graph with no negative-length cycles. This improves
Cabello's bound by a factor of $\tilde O(n^{1/6})$, and is the
first deterministic truly subquadratic algorithm for this problem.

Our algorithm follows the general high-level approach of Cabello, but differs in the way it is implemented.
Our main technical contribution is an efficient and deterministic construction of additively weighted
Voronoi diagrams on planar graphs (under the assumption that  the sites lie on the boundaries of a constant number of faces, which holds in the context considered here).
In contrast to Cabello, we only need to precompute the bisectors between {\em pairs} of sites, and then use a technique, developed in this paper, for intersecting three bisectors in $\tilde O(1)$ time. Moreover, we compute the bisectors faster,
so that all possible bisectors between a given pair of sites can be computed
in $\tilde O(r)$ time, compared to $O(r^2)$ time in~\cite{Cabello}. Another difference is that we use a deterministic divide-and-conquer approach that exploits the planar structure, rather than the randomized incremental construction of abstract Voronoi diagrams used by Cabello.
Formally, we prove the following result.
\begin{theorem}\label{thm:vor}
Let $P$ be a directed planar graph with real arc lengths, $r$ vertices, and no negative-length cycles. Let $S$ be a set of $b$ sites that lie on the boundaries of a constant number of faces (holes) of $P$. One can preprocess $P$ in $\tilde O(r b^2)$ time, so that, given any subset $S' \subseteq S$ of the sites, and a weight $\wt(u)$ for each $u \in S'$, one can construct a representation of the additively weighted Voronoi diagram of $S'$ with respect to the weights $\wt(\cdot)$, in $\tilde O(|S'|)$  time.
With this representation of the Voronoi diagram we can, for any site $u \in S'$, (i)
report the boundary of the Voronoi cell of $u$ in $\tilde O(1)$ time per edge, and (ii) query the maximum distance from $u$ to a vertex in the Voronoi cell of $u$ in $\tilde O(|\bd \Vor(u)|) $ time, where $|\bd \Vor(u)|$ denotes the complexity of the Voronoi cell of $u$.
\end{theorem}

 The above $\tilde O(|S'|)$ construction time is significant because $|S'| \leq b = O(\sqrt{r}) \ll |P|=O(r)$.
This fast construction comes with the price
of the cost of the
preprocessing stage. In other words, a one-time construction of the diagram
is less efficient than a brute force $\tilde{O}(br)$ algorithm (that simply computes a shortest path tree from each site), because we need to account for the preprocessing cost too. Nevertheless, the preprocessing
is \emph{independent} of the weights of the sites. Hence, the overall approach, which consists of
a one-time preprocessing stage (per piece), followed by many calls to the diagram construction procedure,
makes the overall algorithm efficient. 
This general approach is borrowed from Cabello, except that
our improvements lead to a simpler, more efficient and deterministic algorithm for {\sc diameter}. 

As in Cabello's work, our algorithm can compute, for every vertex $v$, both the farthest vertex from $v$ (i.e., the eccentricity of $v$), and the sum of distances from $v$ to all other vertices. 
Hence, our algorithm can also compute the radius, median (the vertex minimizing the sum of distances to all other vertices), and Wiener index (sum of all pairwise distances) of a planar graph within the same time bounds. To the best of our knowledge, these are currently the fastest algorithms for all of these problems in planar graphs. 

Our detailed study of Voronoi diagrams for planar graphs, and the resulting algorithm of Theorem~\ref{thm:vor} are of independent interest, and we believe it will find additional uses. For example, in a recent work, Cohen-Addad, Dahlgaard, and Wulff-Nilsen~\cite{SubquadraticDistanceOracle} show that the Voronoi diagram approach can be used to construct 
exact {\em distance oracles} for planar graphs with subquadratic space and polylogarithmic query time. Specifically, they show how to implement the following point location queries in Voronoi diagrams in planar graphs. Given a vertex $v$, return the site $s$ such that $v$ lies in the Voronoi cell of $s$. They show how to use such queries to obtain 
an oracle of size $\tilde O(n^{5/3})$, query time $O(\log n)$, and preprocessing time $O(n^2)$. They also improve the state-of-the-art tradeoff between space and query time. Using Theorem~\ref{thm:vor} 
to construct the Voronoi diagrams yields an improved and  nearly optimal preprocessing time 
for the oracles in~\cite{SubquadraticDistanceOracle}, 
 matching their space requirements up to polylogarithmic factors.\footnote{ It should be noted that these space-to-query-time tradeoffs have since been significantly improved in~\cite{AlmostOptimal} (see also~\cite{ourSODA2018}). The oracles in~\cite{AlmostOptimal} have almost linear size, and polylogarithmic (or subpolynomial) query time. Their preprocessing time is roughly $O(n^{3/2}$), using a different method for computing a suitable representation of a single additively weighted Voronoi diagram rather than the buy-at-bulk approach of Theorem~\ref{thm:vor}. Whether any of these construction methods can be improved to obtain faster preprocessing time is an interesting open question.}
 
\subsection{The diameter algorithm}\label{sec:framework}
Let $G$ be a directed planar graph with non-negative arc lengths.\footnote{Negative arc length can be handled so long as there are no negative-length cycles. This is done using the standard technique of feasible price functions, which make all the weights non-negative. See, e.g., \cite{Cabello,FR06}.} The algorithm computes an $r$-division of $G$ with few holes per piece.
For a vertex $v$ that is not a boundary vertex, let $P_v$ be the unique piece of the $r$-division that contains $v$ and let $\partial P_v$ denote the set of boundary vertices of $P_v$.
The diameter of $G$ is the length of a shortest $u$-to-$v$ path for some pair of vertices $u, v$. There are three cases: ($i$) at least one of $u,  v$ is a boundary vertex, ($ii$) none of $u, v$ is a boundary vertex and $P_{u} = P_{v}$, and ($iii$) none of $u, v$ is a boundary vertex and $P_{u} \neq P_{v}$.

To take care of case ($i$), we reverse all arcs in $G$ and compute single-source shortest paths trees rooted at each boundary vertex. This takes $O(n \cdot n/\sqrt r)$ time using the linear time algorithm of Henzinger et al.~\cite{HKRS},\footnote{In fact, for our purposes it suffices to use Dijkstra's algorithm that takes $O(n\log n)$ time.} and computes the distance from every vertex in $G$ to every boundary vertex.  

To take care of case ($ii$), the algorithm next computes the distances in $G$ from $v$ to all vertices of $P_v$. This is done by a single-source shortest-path computation in $P_v$, with the distance labels of vertices of $\partial P_v$ initialized to their distance from $v$ in the entire graph $G$ (these distances were already computed in case ($i$)). Note that in doing so we  take care of the possibility that the shortest path from $u$ to $v$ traverses other pieces of $G$, even though $u$ and $v$  lie in the same piece. This takes $O(r)$ time per vertex, for a total of $O(nr)$ time.

It remains to take care of case ($iii$), that is, to compute the maximum distance between vertices that are not in the same piece. To this end we first invoke Theorem~\ref{thm:vor}, on every piece $P$ in the $r$-division, preprocessing $P$ in $\tilde O(r(\sqrt r)^2) = \tilde O(r^2)$ time.
Then, for every non-boundary source vertex $v_0 \in G$, and every piece $P$ in the $r$-division (except for the piece of $v_0$) we compute, in $\tilde O(\sqrt r)$ time, the additively weighted Voronoi diagram of $P$, where the weight $\wt(v)$, for each $v \in \partial P$, is the $v_0$-to-$v$ distance in $G$ (these distances have been computed in case ($i$)). We then use the efficient maximum query in item $(ii)$ of Theorem~\ref{thm:vor} to compute the farthest vertex from $v_0$ in each cell of the Voronoi diagram. Each query takes time proportional to the complexity of the Voronoi cell, which together sums up to the complexity of the Voronoi diagram, which is $\tilde O(\sqrt r)$. Thus, we find the farthest vertex from $v_0$ in $P$ in $\tilde O(\sqrt r)$ time. 
The total preprocessing time over all pieces is $\tilde O(\frac{n}{r} \cdot r^2) = \tilde O(nr)$, and the total query time over all sources $v_0$ and all pieces $P$ is $\tilde O(n \cdot \frac{n}{r} \cdot \sqrt r) = \tilde O(n^2 / \sqrt r)$. Summing the preprocessing and query bounds, we get an overall cost of $\tilde O(nr + n^2/\sqrt r)$.

Concerning the running time of the entire algorithm, we note that computing the $r$-division takes $O(n)$ time \cite{KMS13}. 
The total running time is thus dominated by the $\tilde O(nr + n^2/\sqrt r)$ of case~($iii$). Setting $r=n^{2/3}$ yields the claimed $\tilde O(n^{5/3})$ bound.

\subsection{A high-level description of the Voronoi diagram algorithm}\label{sec:highlevel}
We now outline the proof of Theorem~\ref{thm:vor}. The description is given in the context of the diameter algorithm, so the input planar graph is called a piece $P$, and the faces to which the sites are incident are called holes. The set of sites is denoted by $S$. To avoid clutter, and without loss of generality, we assume that the subset $S'$ of sites of the desired Voronoi diagram is the entire set $S$. 
Hence, $|S'| = |S| = b = O(\sqrt r)$.
Given a weight assignment $\wt(\cdot)$ to the sites in $S$, the algorithm  computes (an implicit representation of)
the \emph{additively weighted Voronoi diagram} of $S$ in $P$, denoted as $\VD(S,\wt)$, or
just $\VD(S)$ for short.
Recall that our goal is to construct many instances of $\VD(S)$ in
$\tilde{O}(b) = \tilde{O}(\sqrt{r})$ time each.
What enables us to achieve this is the fact that
these instances differ only in the weight assignment
$\wt(\cdot)$. We exploit this by carrying out a preprocessing stage, which is weight-independent.
We use the information collected in the preprocessing stage to make the construction of the weighted diagrams efficient.

\medskip
\noindent
{\bf The structure of \boldmath$\VD(S)$.}
Each Voronoi cell $\Vor(v)$ is in fact a tree rooted at $v$, which is a subtree of the shortest-path tree from $v$ (within $P$). We also consider the dual graph $P^*$ of $P$, and use the following dual representation of $\VD(S)$ within $P^*$, which we denote as $\VD^*(S)$. For each pair of distinct sites $u,v$, we define the \emph{bisector} $\beta^*(u,v)$ of $u$ and $v$ to be the collection of all edges $(pq)^*$ of $P^*$ that are dual to edges $pq$ of $P$ for which $p$ is nearer to $u$ than to $v$ and $q$ is nearer to $v$ than to $u$. Each bisector  is a simple cycle in $P^*$.\footnote{In this high level description we treat bisectors as undirected objects. In the precise definition, given in Section~\ref{sec:prelims}, they are  directed.}
In general, in the presence of other sites, only part of  $\beta^*(u,v)$ appears in $\VD^*(S)$.
 The set of maximal connected segments of $\beta^*(u,v)$ that appear in $\VD^*(S)$ is a collection of one or several \emph{Voronoi edges} that separate between the cells $\Vor(u)$ and $\Vor(v)$.\footnote{Even in the case of standard additively weighted Euclidean Voronoi diagrams, two cells can have a disconnected common boundary.} 
 Each Voronoi edge terminates at  a pair of \emph{Voronoi vertices}, where such a vertex $f^*$ is dual to a face $f$ of $P$ whose vertices belong to three or more distinct Voronoi cells. To simplify matters, we triangulate each face of $P$ (except for the holes), so all Voronoi vertices (except those dual to the holes, if any) are of degree $3$. If we contract the edges of each maximal connected segment comprising a Voronoi edge into a single abstract edge connecting its endpoints, we get a planar map of size $O(b)$ (by Euler's formula). This planar map is the dual Voronoi diagram $\VD^*(S)$.
A useful property is that, in the presence of just three sites, the diagram $\VD^*$ has at most two vertices. In fact, even if there are more than three sites, but assuming that the sites lie on the boundaries of only three holes, the diagram has at most two ``trichromatic'' vertices, that is, Voronoi vertices whose dual faces have vertices in the Voronoi cells of sites from all three holes. See Sections~\ref{sec:prelims} and~\ref{sec:structure}.

\medskip
\noindent
{\bf Computing all bisectors.}
The heart of the preprocessing step is the computation of bisectors of every pair of sites $u,v$, for every possible pair of weights $\wt(u),\wt(v)$ that can be assigned to them. Note that the  bisector $\beta^*(u,v)$ only depends on the \emph{difference} $\delta = \wt(v)-\wt(u)$ between the weights of these sites, and that, due to the discrete nature of the setup, the bisector changes only at a discrete set of differences, and one can show that  there are only $O(r)$ different bisectors for each pair $(u,v)$.
Computing a representation of the $O(r)$ possible bisectors for a fixed pair of sites $(u,v)$ is done in $\tilde O(r)$ time by varying $\delta$ from $+\infty$ to $-\infty$. As we do this, the bisector ``sweeps'' over the piece, moving farther from $u$ and closer to $v$. This is reminiscent of the multiple source shortest paths (MSSP) algorithm of~\cite{Klein05,CabelloCE13}, except that in our case $u$ and $v$ do not necessarily lie on the same face. We represent all the possible bisectors for $(u,v)$ using a persistent binary search tree~\cite{persistence} that requires $\tilde O(r)$ space. Summing over all $O((\sqrt r)^2) = O(r)$ pairs of sites, the total preprocessing time is $\tilde O(r^2)$.
See Section~\ref{sec:bisectors}.

\medskip
\noindent
{\bf Computing \boldmath$\VD(S)$.}
We compute the diagram in four steps:
\begin{enumerate} 
\item	
%($i$)
 We first compute the diagram for the sites on the boundary of a single hole (the ``monochromatic'' case). We do this using a divide-and-conquer technique, whose main ingredient is merging two sub-diagrams into a larger diagram. We compute one diagram for each of the $\NH$ holes. 
 \item 
 %($ii$) 
 We then compute the diagram of the sites that lie on the boundaries of a fixed pair of holes (the ``bichromatic case''), for all $\binom{\NH}{2}=O(1)$ pairs of holes. Each such diagram is obtained by merging the two corresponding monochromatic diagrams. 
 \item 
 %($iii$) 
 We then compute all ``trichromatic'' diagrams, each of which involves the sites that lie on the boundaries of three specific holes. Here too we merge the three (already computed) corresponding bichromatic diagrams. 
 \item 
 %($iv$) 
 Since each Voronoi edge (resp., vertex) is defined by two (resp., three) sites (even if a Voronoi vertex is defined by more than three sites, it is uniquely determined by any three of them), we now have a superset of the vertices of $\VD(S)$. Each Voronoi edge of $\VD(S)$, say separating the cells of sites $u,v$, is contained in the appropriate Voronoi edges of all trichromatic diagrams involving the holes of $u$ and $v$. A simple merging step along each bisector constructs the true Voronoi edges and vertices. See section~\ref{sec:voronoi}.
\end{enumerate}

\medskip
\noindent
{\bf Intersecting bisectors, finding trichromatic vertices, and merging diagrams.}
The main technical part of the algorithm is an efficient implementation of steps ($i$)--($iii$) above. A crucial procedure that the algorithm repeatedly uses is an efficient construction of the (at most two) Voronoi vertices of three-sites diagrams, where the cost of this step is only $\tilde{O}(1)$. In fact, we generalize this procedure so that it can compute the (at most two) trichromatic Voronoi vertices of a diagram of the form $\VD(r,g,B)$, where $r$ and $g$ are two sites and $B$ is a sequence of sites on a fixed hole. This extended version also takes only $\tilde{O}(1)$ time. See Section~\ref{sec:trichrom}.

Having such a procedure at hand, each merging of two subdiagrams $\VD(S_1)$, $\VD(S_2)$,
into the joint diagram $\VD(S_1\cup S_2)$, is performed by tracing the {\em $S_1S_2$-bisector}, which is the collection of all Voronoi edges of the merged diagram that separate between a cell of a site in $S_1$ and a cell of a site in $S_2$. In all the scenarios where such a merge is performed, the $S_1S_2$-bisector is also 
a cycle, which we can trace segment  by segment. In each step of the trace, we are on a bisector of the form $\beta^*(u_1,u_2)$, for $u_1\in S_1$ and $u_2\in S_2$. We then take the cell of $u_1$ in $\VD(S_1)$, and compute the trichromatic vertices for $u_1$, $u_2$, and $S_1\setminus\{u_1\}$. We do the same for the cell of $u_2$ in $\VD(S_2)$, and the two steps together determine the first exit point of $\beta^*(u_1,u_2)$ from one of these two cells, and thereby obtain the terminal vertex of the portion of $\beta^*(u_1,u_2)$ that we trace, in the combined diagram. Assume that we have left the cell of $u_1$ and have entered the cell of $u'_1$ in $\VD(S_1)$. We then apply the same procedure to this new bisector $\beta^*(u'_1,u_2)$, and keep doing so until all of the $S_1S_2$-bisector is traced. See Section~\ref{sec:voronoi}.

The details concerning the identification of trichromatic vertices are rather involved, and we only
give a few hints in this overview. Let $r,g,b$ be our three sites. We keep the bisector $\beta^*(g,b)$ fixed, meaning that we keep the weights $\wt(g)$, $\wt(b)$ fixed, and vary the weight $\wt(r)$ from $+\infty$ to $-\infty$. As already reviewed, this causes the cell $\Vor(r)$, which is initially empty, to gradually expand, ``sweeping'' through $P$.
This expansion occurs at discrete critical values of $\wt(x)$.
We show that, as $\Vor(r)$ expands, it annexes a contiguous portion of $\beta^*(g,b)$, which keeps growing as $\wt(x)$ decreases. The terminal vertices of the annexed portion of $\beta^*(g,b)$ are the desired trichromatic vertices (for the present value of $\wt(x)$). By a rather involved variant of binary search through $\beta^*(g,b)$ (with the given $\wt(x)$ as a key), we find those endpoints, in $\tilde{O}(1)$ time. See Section~\ref{sec:trichrom}.

To support this binary search we need to store, at the preprocessing step, for each vertex $p$ of $P$ and for each pair of sites $r,g$, the ``time'' (i.e, difference $\wt(r) - \wt(g)$) at which $\beta^*(r,g)$ sweeps through $p$. Note that this can be done within the $O(r^2)$ time and space that it takes to precompute the bisectors.

Putting all the above ingredients together (where the many details that we skimmed through in this
overview are spelled out  later on), we get an overall algorithm that constructs the Voronoi diagram,
within a single piece and under a given weight assignment, in $\tilde{O}(b) = \tilde{O}(\sqrt{r})$ time.

\medskip
\noindent
{\bf Finding the farthest vertex in each Voronoi cell.}
To be useful for the diameter algorithm, our representation of Voronoi diagrams must be augmented to report the farthest vertex in each Voronoi cell from the site of the cell in nearly linear time in the number of Voronoi vertices of the cell. Such a mechanism has been developed by Cabello~\cite[Section 3]{Cabello}, but only for pieces with a single hole, where a Voronoi cell is bounded by a single cycle. We extend this procedure to pieces with a constant number of holes by exploiting the structure of Voronoi diagrams, and the interaction between a shortest path tree, its cotree, and the Voronoi diagram. See Section~\ref{sec:prep_max}.

\subsection{Discussion of the relation to Cabello's work}
We have already mentioned the similarities and differences of our work and Cabello's.  To summarize and further clarify the relation of the two papers, we distinguish three aspects in which our construction of Voronoi diagrams differs from his.
The first main difference is the faster computation of bisectors, the representation of bisectors, and the new capability to compute trichromatic vertices on the fly, which
has no analogue in~\cite{Cabello}. One could try to plug in just these components into Cabello's algorithm (i.e., still using the randomized incremental construction of abstact Voronoi diagrams) in order to obtain an $\tilde O(n^{5/3})$-time randomized algorithm for {\sc diameter}. Doing so, however, is not trivial since there are difficulties in modifying Cabello's technique for ``filling-up" the holes to work with  our persistent representation of the bisectors. This seems doable, but seems to require quite a bit of technical work.

The second main difference is that we develop a deterministic divide-and-conquer construction of Voronoi diagrams, while Cabello uses the randomized incremental construction for abstract Voronoi diagrams~\cite{KleinMM93}. This makes the algorithm more explicit and deterministic.

The third main difference is that our construction of Voronoi diagrams works when the sites lie on multiple holes,\footnote{We assume throughout the paper that the number of holes is constant but, in fact, the dependency of the construction time in our algorithm on the number of holes is polynomial, so we could tolerate a non-constant number of holes.}  
whereas Cabello's use of abstract Voronoi diagrams requires the sites to lie on a single hole.
Assuming that the sites lie on a single hole would significantly simplify multiple components of our construction of Voronoi diagrams, but has its drawbacks. Most concretely, it leads to a more complicated and less elegant algorithm for diameter due to the need to ``fill-up" holes. 
%Indeed, in~\cite{Cabello}, Cabello defers the entire description of this technical issue to the full version of his paper. 
More generally, allowing multiple holes in the construction of Voronoi diagrams leads to a stronger interface that can be more suitable, easier to use, and perhaps even crucial for other applications of Voronoi diagrams on planar graphs beyond {\sc diameter}. As an anonymous reviewer pointed out, computing the diameter of a graph embedded on a surface of genus $g$ seems to reduce to the planar case with $O(g)$ holes.

\ifdefined\fullver
We next list the parts of our Voronoi construction algorithm that would be simplified by assuming that all sites lie on a single hole. We hope this can assist the reader in navigating through some of the more technically challenging parts of the paper.
(i) Parts of the procedure for finding trichromatic vertices in Section~\ref{sec:trichrom} can be made simpler when considering just a single hole. See, e.g., Section~\ref{sec:overview}. (ii) In the algorithm for merging Voronoi diagrams in Section~\ref{sec:voronoi}, only the single hole case (Section~\ref{sec:single}) is required. Sections~\ref{sec:double}--~\ref{sec:final} are not required.
(iii) In the mechanism for reporting the furthest vertex from a site in Section~\ref{sec:prep_max}, the extension for dealing with multiple holes (Section~\ref{sec:max_holes}) is not required.
\fi

 % prelims
\section{Preliminaries}\label{sec:prelims}
\medskip
\noindent
{\bf Planar embedded graphs.}
We assume basic familiarity with planar embedded graphs and planar duality. We treat the graph $G=(V,E)$ as a directed object, where $V$ is the set of vertices, and $E$ is the set of arcs. We assume that for every arc $e=uv$ there is an antiparallel arc $rev(e)=vu$ that is embedded on the same curve in the plane as $e$. Each arc $uv$ has a {\em length} $\ell(uv)$ associated with it. The length of an arc and its reverse need not be equal. We call $u$ the {\em tail} of $e$ and $v$ the {\em head} of $e$. We use the term {\em  edge} when the direction plays no role (i.e., when we wish to refer to the undirected object, not distinguishing
between the two antiparallel arcs).
The dual of a planar embedded graph $G$ is a planar embedded graph $G^* = (V^*, E^*)$, where the  nodes in $V^*$ represent faces in $G$, and the dual arcs in
$E^*$ stand in 1-1 correspondence with the arcs in $E$, in the sense that the arc $e^*$ dual to an arc $e$ connects
the face to the left of $e$ to the face to the right of $e$.
We use some well known properties of planar graphs, see e.g., \cite{planarbook}. If $T$ is a spanning tree of $G$ then the edges not in $T$ form a spanning tree $T^*$ of the dual $G^*$. The tree $T^*$ is called the {\em cotree} of $T$. If $(X,\bar X)$ is a partition of $V$, and the subgraphs induced by $X$ and by $\bar X$ are both connected, then the set of duals of the arcs whose tail is in $X$ and whose head is in $\bar X$ forms a simple cycle in $G^*$.

\medskip
\noindent
{\bf Assumptions.}
By a \emph{piece} $P=(V,E)$ we mean an embedded directed planar graph 
with $t=O(1)$ distinguished faces $h_1,\ldots, h_t$, to which we refer as \emph{holes}.
(In our context, the pieces are the subgraphs of the input graph $G$ produced by an $r$-division.)
We assume, for the diameter algorithm, that arc lengths are non-negative.
This assumption can be enforced using the standard technique of price functions and reduced lengths (see, e.g.,\cite{Cabello,FR06}). 
We assume that all faces of $G$, except possibly for the holes, are triangulated.
This assumption can be enforced by  triangulating each non-triangular face that is not a hole by infinite length diagonal edges.
Except for the dual nodes representing the holes, all other nodes of $P^*$ therefore have degree $3$.

We assume that shortest paths are unique. Our assumption is identical to the one made in~\cite{CabelloCE13}. More specifically, let $P$ be a piece with a set $S$ of boundary sites and additive weights $\wt(\cdot)$. Consider the process in which we vary the weight $\wt(u)$ of exactly one site $u \in S$. We assume that vertex-to-vertex distances in $P$ are distinct, and that shortest paths are unique, except at a discrete set of  critical values of $\wt(u)$ where there is a unique {\em tense} arc (see Section~\ref{sec:bisectors} for the definition). This assumption can be achieved deterministically with $\tilde O(1)$ time overhead using a deterministic lexicographic perturbation~\cite{CabelloCE13,HartvigsenMardon}. See also~\cite{EFL18} for a different approach. 

\medskip
\noindent
{\bf Additively weighted Voronoi diagram in a piece.}
We are given a piece $P=(V,E)$ with  $|V| = O(r)$ (and so $|E| = O(r)$). We are also given a set $S\subseteq V$ of $b=O(\sqrt{r})$ vertices, each of which is a vertex of one of the $O(1)$ holes $h_1,\ldots, h_t$; $S$ is a subset of the boundary vertices of the piece $P$, and its elements are the {\em sites} of the Voronoi diagram that we are going to define. Each site $v\in S$ has a real \emph{weight} $\wt(v)$ 
associated with it. We think of $\wt$ as a weight function on $S$.
For every pair of vertices $u$ and $v$, we denote by $\pi(u,v)$  the shortest path from $u$ to $v$ in $P$.
We denote the length of $\pi(u,v)$ by $d(u,v)$. The (additively weighted)
 {\em distance} between a site $u \in S$ and a vertex $v\in V$, is $\wt(u) + d(u,v)$.

The \emph{additively weighted Voronoi diagram} of $(S,\wt)$ within $P$, denoted by $\VD(S,\wt)$ (we will often
drop $\wt$ from this notation, although the diagram does depend on $\wt$), is a partition of $V$ into pairwise disjoint
sets, one set, denoted by $\Vor_{S,\wt}(u)$, for each site $u\in S$. We omit the subscript when it is clear from the context. The set $\Vor(u)$ contains all vertices closer (by additively weighted distance) to $u$ than to any other site $u' \not= u \in S$.
We call $\Vor(u)$ the (primal) {\em Voronoi cell of $u$}. Note that a Voronoi cell might be empty.
In what follows, we denote the shortest-path tree rooted at a site $u\in S$ as $T_u$. See Figure~\ref{fig:VD} for an illustration of some of the definitions in this section.
The Voronoi diagram has the following basic properties.

\begin{figure}[h]
\begin{center}
\includegraphics[width=0.7\textwidth, clip=true, trim = 0mm 0mm 150mm 0mm]{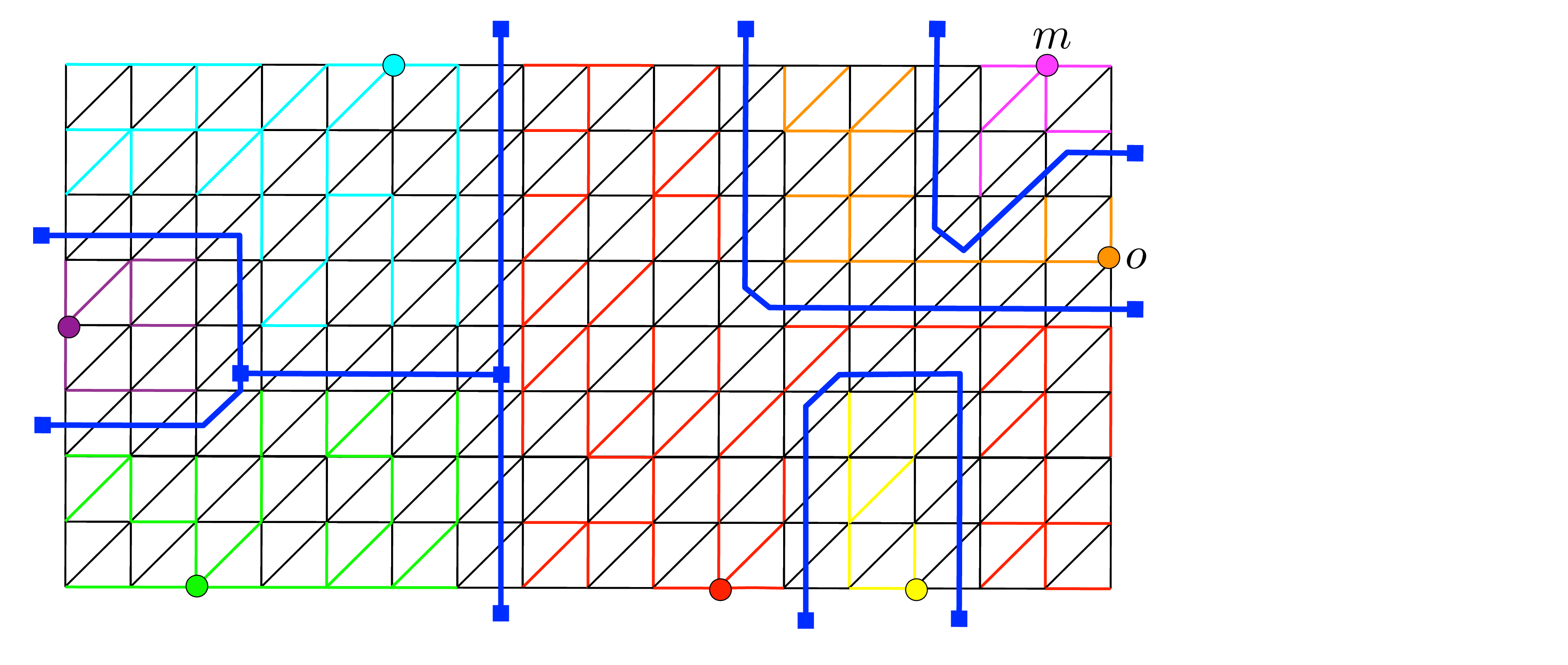}
\caption{\small A graph (triangulated grid) with 7 sites on the infinite face (colored circles). Edge lengths and additive weights are not indicated. For each site $u$, the subtree of $T_u$ spanning $\Vor(u)$ is indicated by edges with the same color as $u$. The dual Voronoi diagram $\VD^*(S)$ is shown in blue. To avoid clutter, edges of $\VD^*(S)$ incident to the infinite face are not shown to have a common endpoint. In this example $\VD^*(S)$ has 3 Voronoi vertices which are indicated by solid blue squares (all the blue squares in the infinite face are in fact the same vertex). In this example, the boundary between the cells of the magenta site $m$ and the orange site $o$ happens to be the entire bisector $\beta^*(m,o)$, which is a simple dual cycle that passes through the infinite face.
\label{fig:VD}}
\end{center}
\end{figure}

\begin{restatable}{lemma}{vorsubtree}
\label{lem:vorsubtree}
For each $u\in S$, the vertices in $\Vor(u)$ form a connected subtree (rooted at $u$) of $T_u$.
\end{restatable}
\begin{proof}
This is an immediate consequence of the property that shortest paths
from a single source cannot cross one another
(under our non-degeneracy assumption); in fact, they cannot even meet one another except at a common prefix.
\hfill \fullqed \end{proof}

Let $B\subseteq E$ be  the set of edges $vw\in E$ such that the sites $u_1$, $u_2$,
for which $v\in \Vor(u_1)$ and $w\in \Vor(u_2)$, are distinct.
Let $B^*$ denote the set of the edges that are dual to the edges of $B$,
and let $\VD^*(S,\wt)$ (or $\VD^*(S)$ for short) denote the subgraph of $P^*$ with $B^*$ as a set of edges.
The following consequence of Lemma~\ref{lem:vorsubtree}
gives some structural properties of $\VD^*(S)$.

\begin{restatable}{lemma}{dualvor}
\label{lem:dualvor}
The graph $\VD^*(S)$ consists of at most $|S|$ faces, so that each of its faces corresponds to a
site $u\in S$ and is the union of all faces of $P^*$ that are dual to the vertices of $\Vor(u)$.
\end{restatable}
\begin{proof}
For each $u\in S$, the union of all faces of $P^*$ that are dual to the vertices of $\Vor(u)$
is connected, since $\Vor(u)$ is a tree.
Moreover, by construction, a dual edge that separates two adjacent such faces is not an edge of $\VD^*(S)$. %cannot be part of any bisector.
 Each face of $P^*$ belongs to exactly one face of $\VD^*(S)$,
because the trees $\Vor(u)$ are pairwise disjoint, and form a partition of $V$. Hence the faces of $\VD^*(S)$
stand in 1-1 correspondence with the trees $\Vor(u)$, for $u\in S$, in the sense asserted in the lemma, and the claim follows.
\hfill \fullqed \end{proof}

We refer to the face of $\VD^*(S)$ corresponding to a site $u\in S$ as the \emph{dual Voronoi cell} of $u$, and denote it as $\Vor^*(u)$. ($\Vor^*(u)$ is empty when $\Vor(u)$ is empty.) By Lemma \ref{lem:dualvor},  $\Vor^*(u)$ is the union of the faces dual to the vertices of $\Vor(u)$.
Once we fix a concrete way in which we draw the dual edges in $B^*$ in the plane, we can regard
each $\Vor^*(u)$ as a concrete embedded  planar region. Since the sets $\Vor(u)$ form a partition of $V$, it follows that the sets $\Vor^*(u)$ induce a partition of the sets of dual faces of $P^*$. 

We define a vertex $f^*\in \VD^*(S)$ to be a {\em Voronoi vertex} if its degree in $\VD^*(S)$ is at least $3$.
This means that there exist at least three distinct sites whose primal Voronoi cells contain vertices
incident to the primal face $f$ dual to $f^*$. The next corollary follows directly from Euler's formula for planar graphs.

\begin{corollary}
The graph  $\VD^*(S)$ consists of $O(|S|)$ vertices of degree $\ge 3$ (which are the Voronoi vertices); all other vertices are of degree $2$.
The only vertices of degree strictly larger than $3$ are those corresponding to the non-triangular
holes among $h_1,\ldots, h_t$.
\end{corollary}

\begin{restatable}{lemma}{twovvert}
\label{lem:2vvert}
For any three distinct sites $u$, $v$, $w$ in $S$
there are at most two faces $f$ of $P$
such that each of the cells $\Vor(u)$, $\Vor(v)$, $\Vor(w)$ contains a vertex of $f$.
\end{restatable}
\begin{proof}
Assume to the contrary that there are three such faces $f_1$, $f_2$, $f_3$.
Let $p_i$, $q_i$, $r_i$, for $i=1,2,3$, denote the vertices of $f_i$ satisfying
$p_i\in\Vor(u)$, $q_i\in\Vor(v)$, and $r_i\in\Vor(w)$. Let $f_i^*$ be an arbitrary
point inside $f_i$, for $i=1,2,3$. We construct the following embedding of $K_{3,3}$
in the plane. One set of vertices is $\{f_1^*,f_2^*,f_3^*\}$ and the other is $\{u,v,w\}$.
To connect $f_i^*$ with $u$, say, we connect $u$ to $p_i$ via the shortest path
$\pi(u,p_i)$, concatenated with the segment $p_if_i^*$. The connections with
$v$, $w$ are drawn analogously. We calim that  the edges in this
drawing do not cross one another.
To see this note that, under our non-degeneracy assumption, shortest paths from a single site do not cross, and shortest paths within distinct Voronoi cells are vertex disjoint since distinct Voronoi cells are vertex disjoint.
We thus get an impossible planar embedding of $K_{3,3}$, a contradiction that implies the claim.
See Figure~\ref{2verts}.
\hfill \fullqed \end{proof}

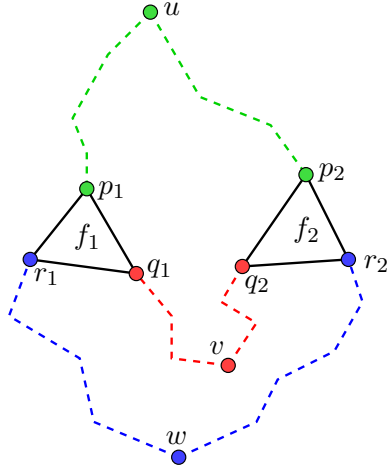
\begin{figure}[htb]
\begin{center}
\centerline{ \input{2verts.pstex_t}}
\caption{\small Illustration of the proof of Lemma~\ref{lem:2vvert}.
\label{2verts}}
\end{center}
\end{figure}

Let $u$ and $v$ be two distinct sites in $S$. 
Let $D$ denote the doubleton $\{u,v\}$.
The \emph{bisector} between $u$ and $v$ (with respect
to their assigned weights), denoted as $\beta^*(u,v)$, is defined to be the set of arcs of $P^*$ whose corresponding primal arcs have their tail 
in $\Vor_{D,\wt}(u)$ and their head in $\Vor_{D,\wt}(v)$. 
In other words, $\beta^*(u,v)$ is the set of duals of arcs of $P$ whose tail is closer (with respect to $\wt$) to $u$ than to $v$, and whose head is closer to $v$ than to $u$.
Note that, unless we explicitly say otherwise, the bisector $\beta^*(u,v)$ is a directed object, and  $\beta^*(v,u)$ consists of the reverses of the arcs of $\beta^*(u,v)$. Bisectors satisfy the following crucial property.
 
\begin{restatable}{lemma}{bicycle}
\label{lem:bicycle}
$\beta^*(u,v)$ is a simple cycle of arcs of $P^*$. If $u$ and $v$ are incident to the same hole $h$ and $\beta^*(u,v)$ is nonempty then $\beta^*(u,v)$ is incident to $h^*$.
\end{restatable}
\begin{proof}
In the diagram of only the two sites $u$, $v$,
$\Vor(u)$ and $\Vor(v)$ form a partition of the vertices of $P$ into two connected sets. Therefore, the set $B$ of arcs with tail in $\Vor(u)$ and head in $\Vor(v)$ is a simple cut in $P$. By the duality of simple cuts and simple cycles (cf.~\cite{planarbook}), $B^* = \beta^*(u,v)$ is a simple cycle.
If $u$ and $v$ are incident to the same hole and $\beta^*(u,v)$ is nonempty,  the simple cut defined by the partition $(\Vor(u),\Vor(v))$ must contain an arc $e$ on the boundary of $h$. Therefore, $e^*$ is an arc of $\beta^*(u,v)$ that is incident to $h^*$.
\hfill \fullqed \end{proof}

By our conventions about the relation between primal and dual arcs, the bisector $\beta^*(u,v)$ is a directed clockwise cycle around $u$, and $\beta^*(v,u)$ is a directed clockwise cycle around $v$. 
Viewed as an undirected object, $\beta^*(u,v)$ corresponds to the dual Voronoi diagram $\VD^*(\{u,v\},\wt)$. See Figure~\ref{fig:VD}.

\section{Computing bisectors during preprocessing} \label{sec:bisectors}

To facilitate an efficient implementation of the algorithm, we carry out
a preprocessing stage, in which  we compute the bisectors
$\beta^*(u,v)$ for every pair of sites $u,v\in S$ and for every pair of weights that can
be assigned to $u$ and $v$. To clarify this statement, note that $\beta^*(u,v)$
only depends on the \emph{difference} $\delta = \wt(v)-\wt(u)$ between the weights
of $u$ and $v$, and that, due to the discrete nature of the setup, $\beta^*(u,v)$
changes only at a discrete set of differences. The preprocessing stage computes,
for each pair $u,v\in S$, all possible bisectors, by varying $\delta$ from $+\infty$
to $-\infty$. As we do this, $\beta^*(u,v)$ ``sweeps'' over $P^*$, moving farther
from $u$ and closer to $v$, in a sense that will be made more precise shortly.
We find all the critical values of $\delta$ at which $\beta^*(u,v)$ changes,
and store all versions of $\beta^*(u,v)$ in a (partially) persistent binary search tree~\cite{persistence}.
Each version of the bisector is represented as a binary search tree on the (cyclic) list
of its dual vertices and edges (which we cut open at some arbitrary point, to make the list linear). Hence, we can find the
$k$-th edge on any bisector $\beta^*(u,v)$ in $O(\log |\beta^*(u,v)|) = O(\log r)$ time,
for any $1\le k \le |\beta^*(u,v)|$
(where $|\beta^*(u,v)|$ denotes the number of edges on the bisector $\beta^*(u,v)$).

 Consider a critical value of $\delta=\wt(v)-\wt(u)$
at which $\beta^*(u,v)$ changes.
We assume (see Section \ref{sec:prelims}) that
there is a unique arc
$yz\in T_v$ such that $y\in \Vor(v)$, $z\in \Vor(u)$, and
$\delta = d(u,z) - d(v,z)$. We say that $yz$ is {\em tense} at $\delta$.
For $\delta' > \delta$, $z$ was a node in the shortest-path (sub)tree $\Vor(u)$,
and for $\delta' < \delta$ it becomes a node of $\Vor(v)$, as a child of $y$. If $z$ was not a leaf of $\Vor(u)$
(at the time of the switch), then the entire subtree rooted at $z$ moves with $z$ from
$\Vor(u)$ to $\Vor(v)$ (this is an easy consequence of the property that, under our unique shortest paths assumption, shortest paths from a single source do not
meet or cross one another). See Figure~\ref{fig:pivot}. These dynamics imply the following crucial property of bisectors (the proof is deferred to Section~\ref{sec:structure}).

\begin{restatable}{lemma}{contigbisector}
\label{lem:contig-bisector}
Consider some critical value $\delta$ of $\wt(v)-\wt(u)$ where $\beta^*(u,v)$ changes.
The dual edges that newly join $\beta^*(u,v)$ at $\delta$ form a contiguous portion of the new bisector, and
the dual edges that leave $\beta^*(u,v)$ at $\delta$ form a contiguous portion of the old bisector.
\end{restatable}

\begin{figure}[h]
\begin{center}
\includegraphics[width=0.4\textwidth, clip=true, trim = 0mm 0mm 150mm 0mm]{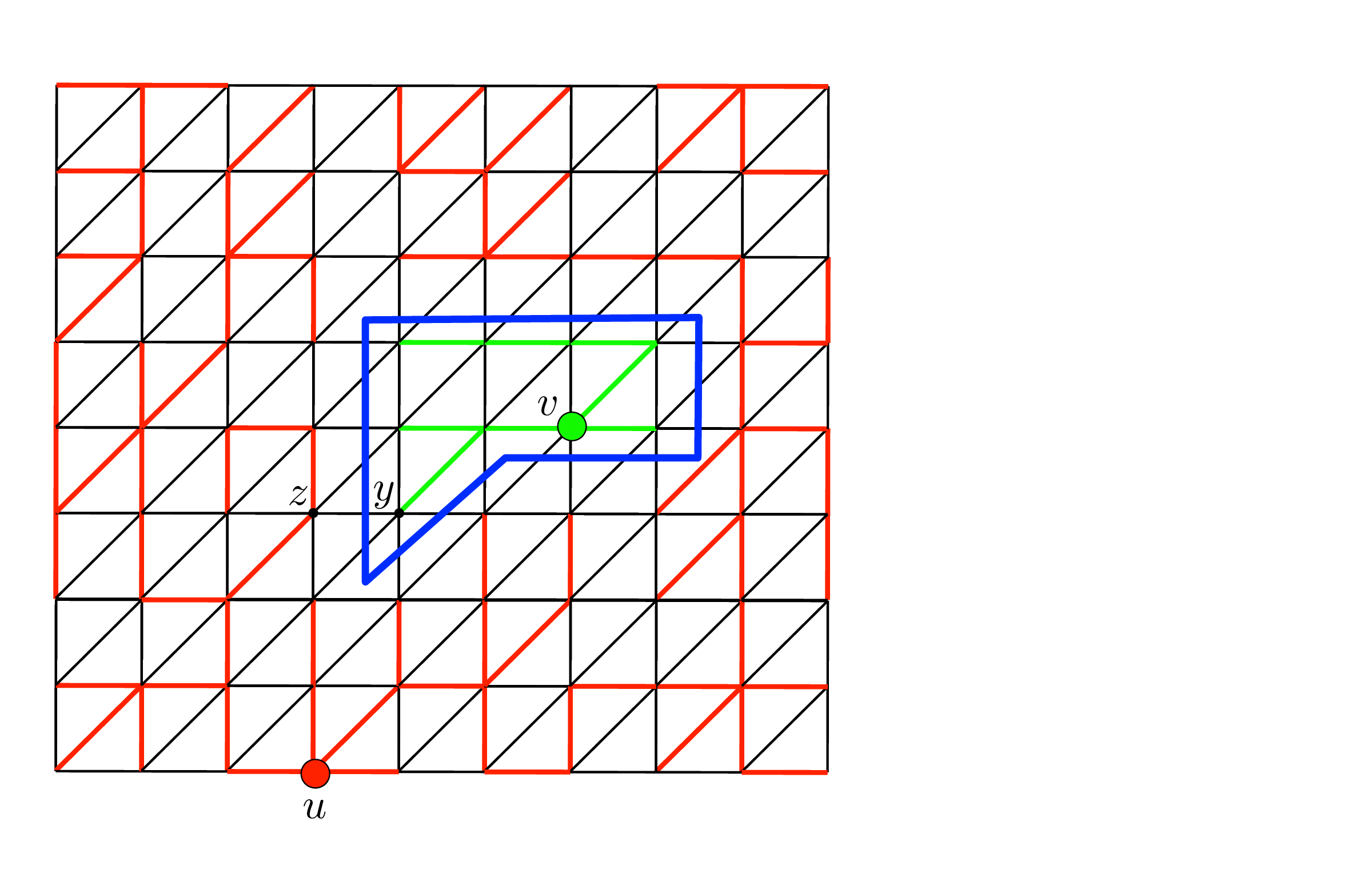}
\hspace{0.1\textwidth}
\includegraphics[width=0.4\textwidth,, clip=true, trim = 0mm 0mm 150mm 0mm]{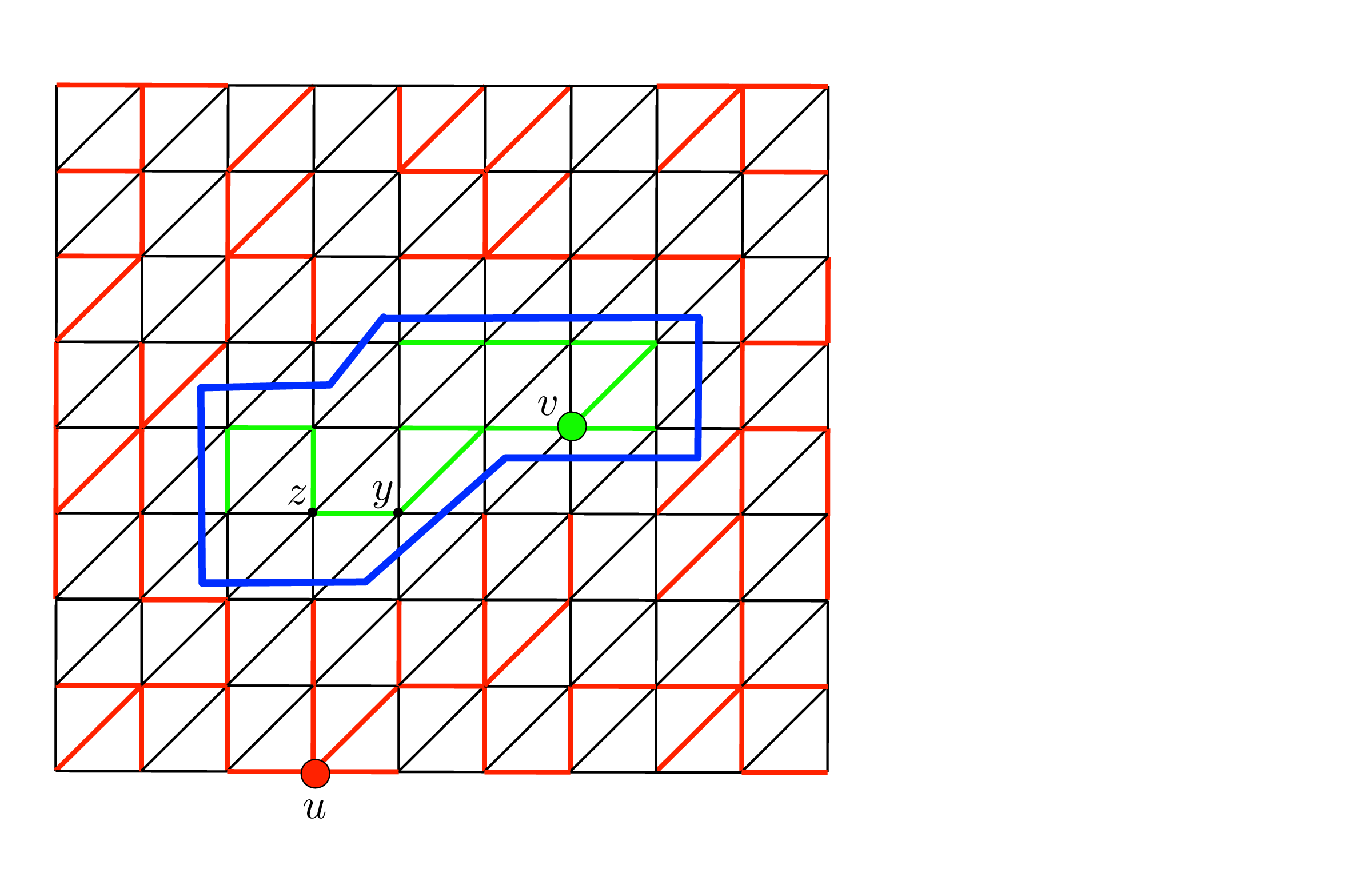}
\caption{\small Illustration of the changes in a bisector $\beta^*(u,v)$ at some critical value $\delta$. $T_u$ is red, $T_v$ is green, and $\beta^*(u,v)$ is blue. The tense arc at $\delta$ is $yz$. Left: the situation just above the critical value $\delta$. The vertex $z$ belongs to $\Vor(u)$. Right: the situation just below $\delta$. The vertex $z$ and all its descendants in $T_u$ switch into $\Vor(v)$. The edges that leave and enter the bisector form a single subpath of the old and new bisector, respectively.  
\label{fig:pivot}}
\end{center}
\end{figure}

We compute and store,  for each vertex $p$, and for each pair of sites $u$, $v$, the unique value 
of $\delta$ at which $p$ moves from $\Vor(u)$ to $\Vor(v)$.
Denote this value as $\delta^{uv}(p)$. Namely, $\delta^{uv}(p) = d(u,p)-d(v,p)$. Thus, for $\delta > \delta^{uv}(p)$, $p \in \Vor(u)$, and for $\delta < \delta^{uv}(p)$, $p \in \Vor(v)$.
Consider an
 arc $pq$ of $P$.
If $\delta^{uv}(p)=\delta^{uv}(q)$ (i.e.\ $pq$ is not the tense arc) then both endpoints of the arc move together
from $\Vor(u)$ to $\Vor(v)$, so the dual arc never becomes an arc of
$\beta^*(u,v)$.
If $\delta^{uv}(p)<\delta^{uv}(q)$ then $q$ moves first (as $\delta$ decreases) and right after time
$\delta^{uv}(q)$ the arc dual to $pq$ becomes an arc of $\beta^*(u,v)$.
It stops being an arc of $\beta^*(u,v)$ at $\delta^{uv}(p)$.
Similarly,
if $\delta^{uv}(q)<\delta^{uv}(p)$ then $p$ moves first and right after time
$\delta^{uv}(p)$ the arc dual to $qp$ becomes an arc of $\beta^*(u,v)$.
It stops being an arc of $\beta^*(u,v)$ at $\delta^{uv}(q)$.

We compute the bisectors $\beta^*(u,v)$ by adding and deleting arcs at the appropriate values of $\delta$.
By Lemma~\ref{lem:contig-bisector}, the arcs that leave and enter $\beta^*(u,v)$ at any particular value of $\delta$ form a single subpath of $\beta^*(u,v)$. We can infer the appropriate position of each arc by ensuring that the order of tails of the arcs of $\beta^*(u,v)$ respects the preorder traversal of $T_u$.\footnote{By preorder traversal we mean breadth first search where descendants of a node are visited according to their cyclic order in the embedding.} 
Note that there are $O(b^2)$ pairs of sites in $S$,
so for each vertex $p \in V$ we compute and store $O(b^2)$  values of $\delta$,
 for a total of $O(rb^2)$ storage.
To compute  the bisectors
for each pair of sites, we perform $O(r)$ updates to the corresponding persistent search tree. 
We have thus established the following towards the preprocessing part of Theorem~\ref{thm:vor}.
\begin{theorem}\label{thm:bisectors}
Consider the settings of Theorem~\ref{thm:vor}. One can compute in $\tilde O(rb^2)$ time and space the persistent binary search tree representations of all possible bisectors for all pairs of sites in $S$. 
\end{theorem}

We remark that the process we described in this section is in fact a special case of the algorithm for multiple source shortest paths~\cite{CabelloCE13} in the case of just two sources.

 % structure
\section{Additional structure of Voronoi diagrams}\label{sec:structure}
In this section we study the structural properties that will be used in the fast construction of a Voronoi diagram from the precomputed bisectors (Section~\ref{sec:voronoi}).

%\subsection{Arc labels.}\label{sec:arclabels}
We first define labels for the arcs of a piece $P$. This is an extension of preorder traversal labels of a spanning tree of $P$ to labels to all arcs of $P$. We will use these labels to prove Lemma~\ref{lem:contig-bisector}, as well as in the weakly bitonic search in Section~\ref{sec:trichrom}.

Let $c \in S$ be a site, and let $T_c$ be the shortest path tree rooted at $c$. 
We define \emph{labels} for the arcs of $P$ with respect to the site $c$.
For the sake of definition only, we modify $P$ as follows. 
For each arc $uv$ of $P$, we create an
artificial vertex $x_{uv}$, and an artificial arc $ux_{uv}$.
The arc $ux_{uv}$ is embedded so that it immediately precedes $uv$ in the counterclockwise cyclic order of arcs around $u$. See Figure~\ref{fig:preorder}.
Let $T'_c$ be the tree obtained by adding to $T_c$ all the artificial arcs. Note that the artificial arcs are leaf arcs in $T'_c$. For each arc $uv\in P$, define its
\emph{label} $\pre_c(uv)$ to be the preorder index of the artificial vertex $x_{uv}$ in $T'_c$,
in a \emph{CCW-first} search of $T'_c$.\footnote{%
  A CCW-first search is a DFS that visits the next unvisited child of a node $u$  in counter-clockwise order. The counter-clockwise order starts and ends at the incoming arc of $u$. For the root vertex $c$, which has no incoming arc, the traversal starts from an imaginary artificial incoming arc embedded in the hole on which $c$ lies.}
Similarly, define the label $\pre_c(vu)$ of the reverse arc $vu$ to be the preorder index of the vertex $x_{vu}$ in $T'_c$.
We define $\pre_c(e^*)$ to be $\pre_c(e)$. The goal of the artificial arcs and vertices is to enable
us to extend the definition of the preorder labels to all arcs of $P$ and their reverses.
Note that this order is consistent with the usual preorder on just the arcs of $T_c$.

\begin{figure*}[h]
\begin{center}
\includegraphics[scale=0.7]{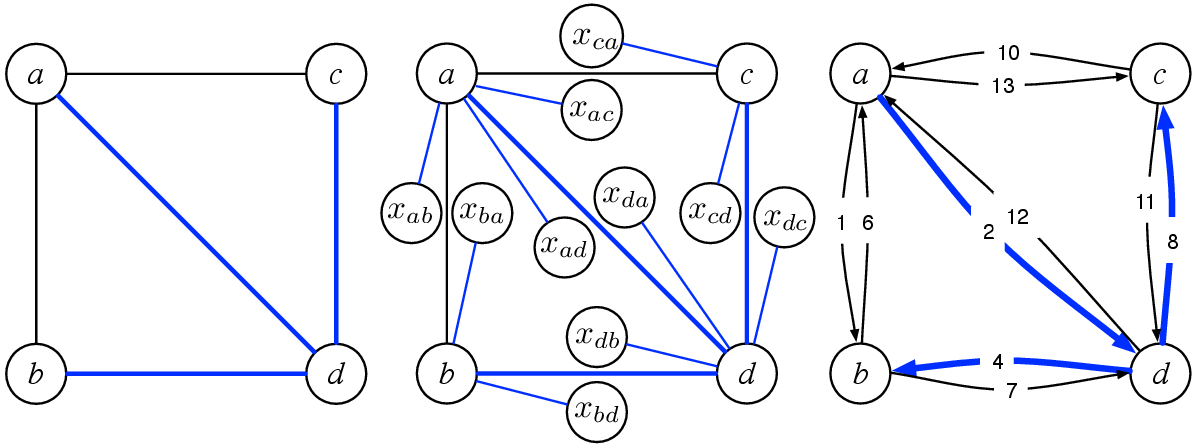}
\caption{\small Illustration of the definition of labels $\pre_c(\cdot)$ for all arcs.
Left: A graph with a spanning tree $T_a$ (in blue), rooted at vertex $a$.
Middle: The same graph with artificial arcs and vertices. The tree $T'_a$ is shown in blue.
Right: The labels $\pre_a(\cdot)$ for all arcs of the original graph.
\label{fig:preorder}}
\end{center}
\end{figure*}

\begin{lemma}
\label{lem:consecutive}
Consider $\VD(\{g,b\},\{\wt(g),\wt(b)\})$ for two sites $g,b$ and additive weights $wt(g),\wt(b)$. Let $c \in \{g,b\}$.  Let $u \neq c$ be a node in the cell $\Vor(c)$ in $\VD(\{g,b\},\{\wt(g),\wt(b)\})$. Let $p$ be the parent of $u$ in $T_c$. Then the following hold:
\begin{enumerate}
\item
The arcs in $R^*_u := \{ e^* \in \beta^*(g,b) \mid \pre_c(pu) < \pre_c(e^*) < \pre_c(up) \}$ form a subpath of $\beta^*(g,b)$.
\item
The labels $\pre_c(e^*)$ are strictly monotone along $R^*_u$.
\end{enumerate}
\end{lemma}

\begin{proof}
Consider the cells $\Vor(g)$ and $\Vor(b)$ of $g$ and $b$ in $\VD(\{g,b\},\{\wt(g),\wt(b)\})$.
They classify the faces of $P$ into three types: (i) those that are not incident to any vertex of $\Vor(g)$,
(ii) those that are not incident to any vertex of $\Vor(b)$, and (iii) those that are incident to a vertex of each set.
Faces of type (iii) are the faces dual to the vertices of $\beta^*(g,b)$. We denote by $f_g$ the union of the faces of
type (i) and (iii), and by $f_b$ the union of the faces of type (ii) and (iii). It follows that for every arc
$e^* \in \beta^*(g,b)$, such that $e=vw$ (so $v$ belongs to $\Vor(g)$), both $v$ and $w$ lie in the face $f_g$,
where $v$ lies on its boundary and $w$ in its interior. See Figure~\ref{fig:consecutive} for an illustration.
Viewed as sets of edges in $P$, $\bd f_g$ and $\bd f_b$ are cycles, which we denote by $\beta_g$ and $\beta_b$, respectively.
Note that the arcs with tail on $\beta_g$ and head on $\beta_b$ are exactly the dual arcs of $\beta^*(g,b)$.
Note also that, since the restriction of $T_g$ to $\Vor(g)$ is a connected subtree of $T_g$ (Lemma~\ref{lem:vorsubtree}),
a branch of $T_g$ that enters $f_g$ (through $\beta_g$) does not leave it.

\begin{figure*}[h]
\begin{center}
\includegraphics[scale=0.43, clip=true, trim = 15mm 15mm 15mm 15mm]{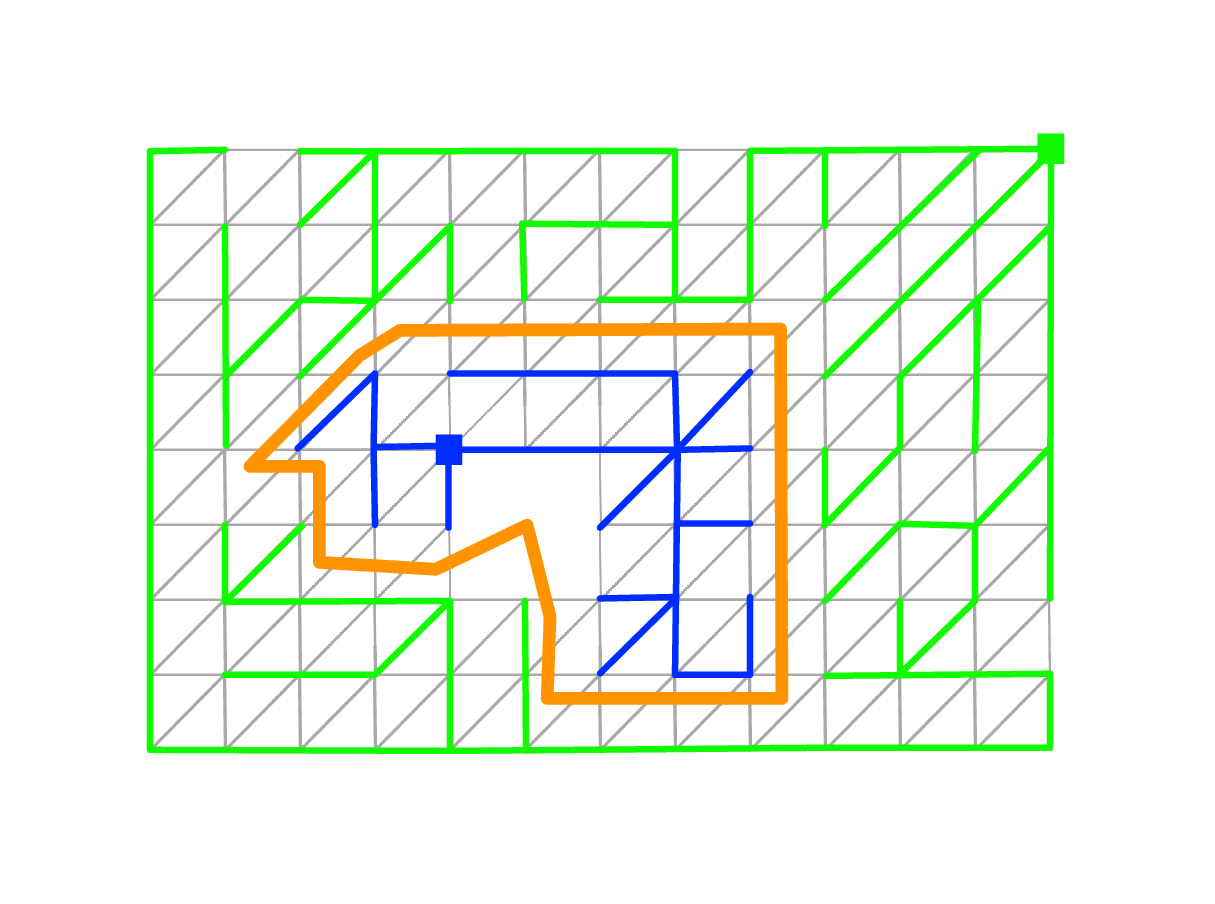}
\hspace{0.3in}
\includegraphics[scale=0.43, clip=true, trim = 15mm 15mm 15mm 15mm]{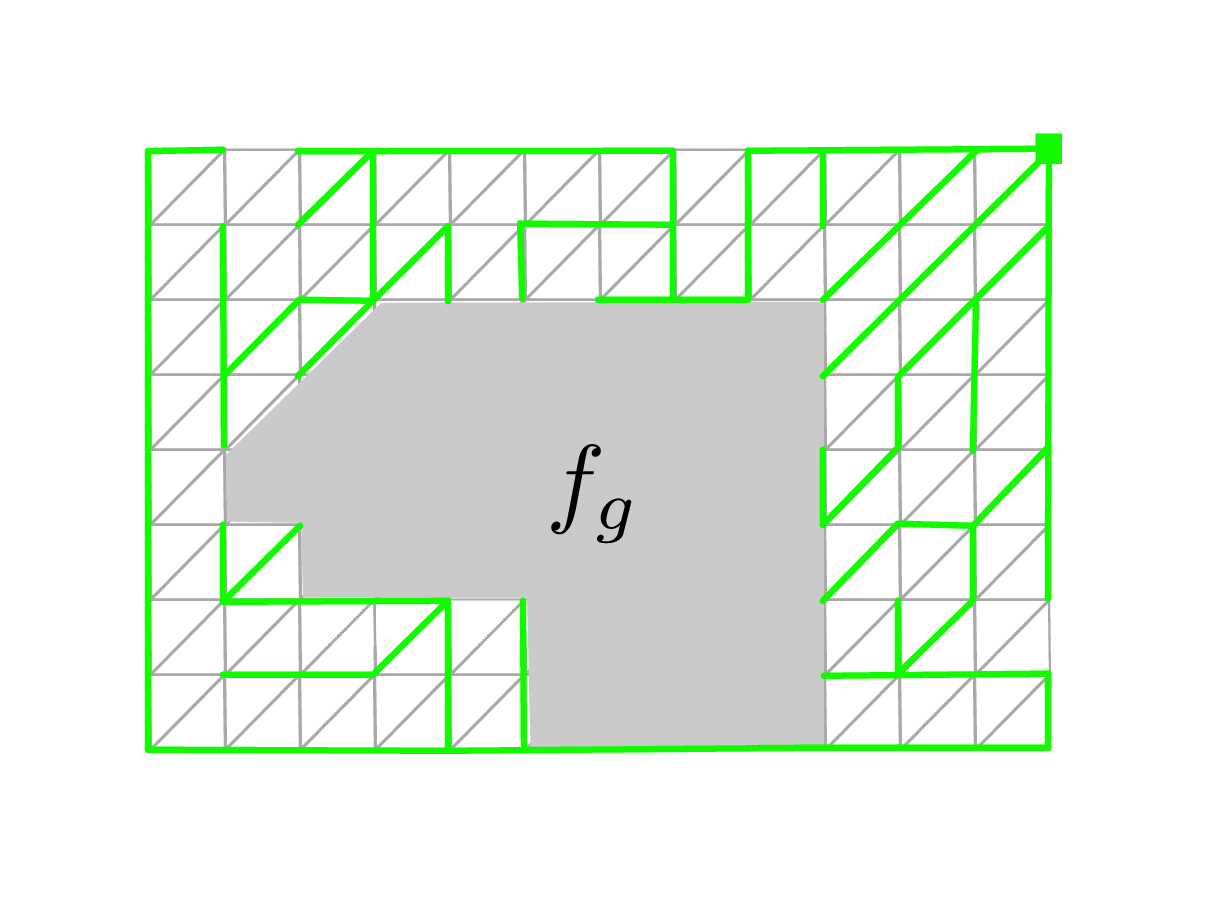}
\includegraphics[scale=0.43, clip=true, trim = 15mm 15mm 15mm 15mm]{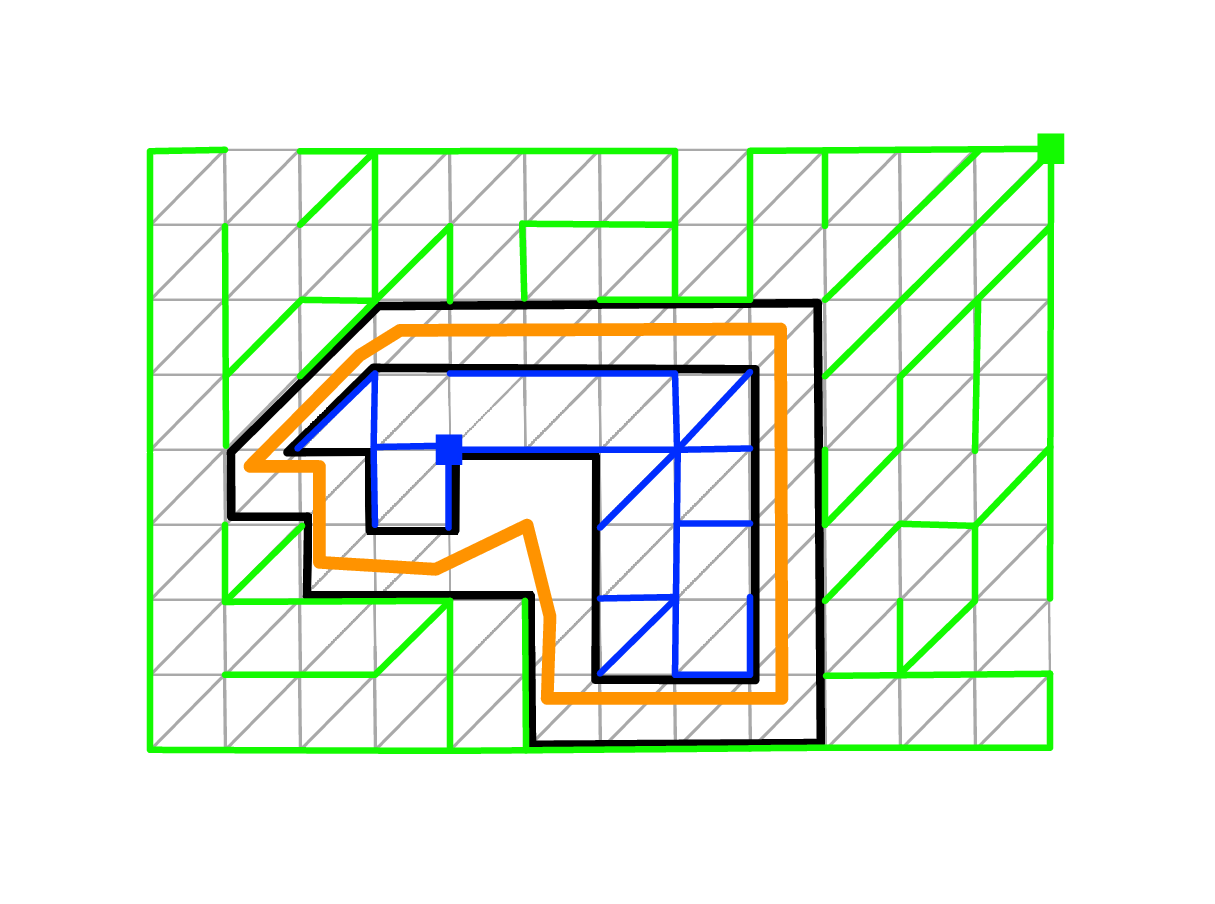}
\hspace{0.3in}
\includegraphics[scale=0.43, clip=true, trim = 15mm 15mm 15mm 15mm]{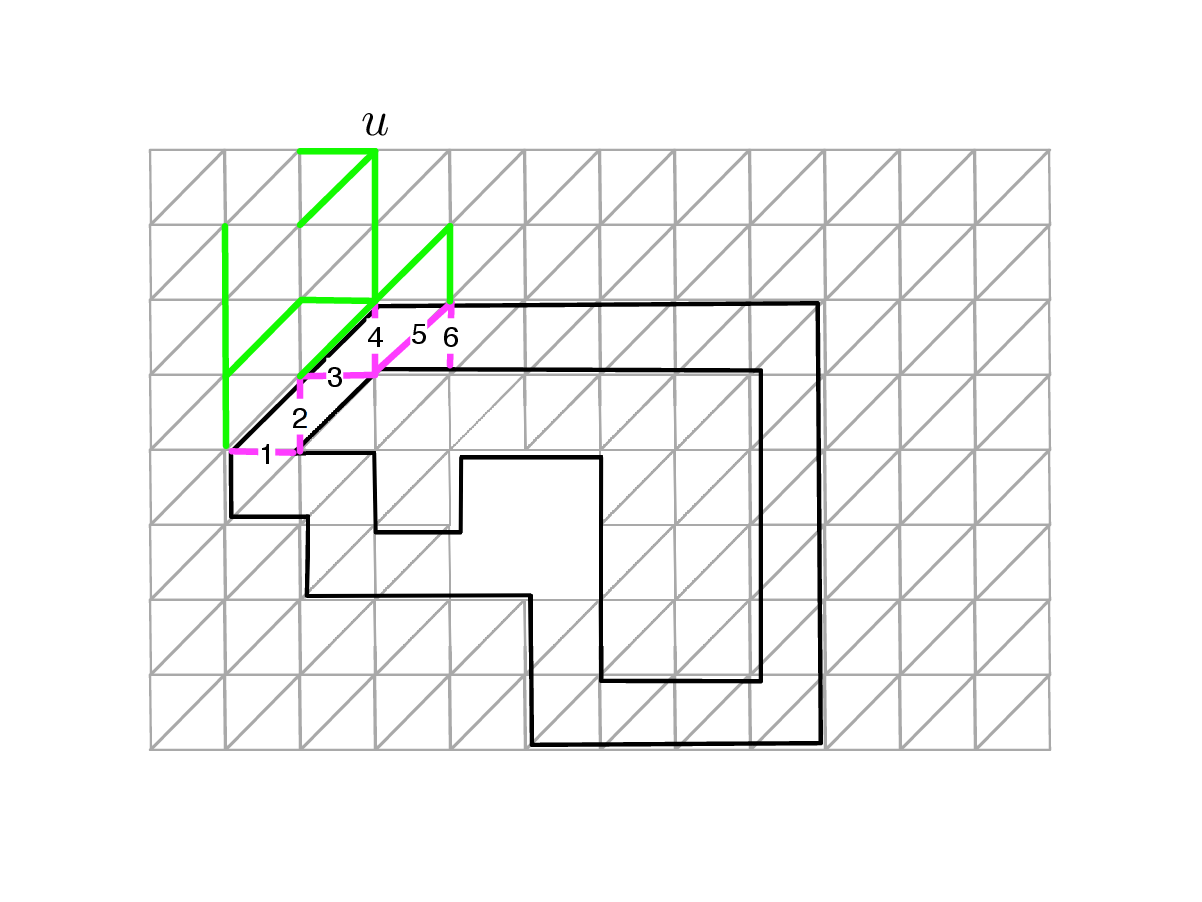}
\caption{\small Illustration of the proof of Lemma~\ref{lem:consecutive}.
Top left: A graph $P$ with two sites $g$ (green square) and $b$ (blue square) lying on two holes. The bisector $\beta^*(g,b)$ is shown in orange.
The subtrees of $T_g$ and $T_b$ spanning the nodes of $\Vor(g)$ and $\Vor(b)$ are shown in green and blue, respectively.
Top right: The Voronoi region $\Vor(g)$ is shown in green. The face $f_g$ is grey.
Bottom left: The cycles $\beta_g$ and $\beta_b$ corresponding to the boundaries of
$f_g$ in $\Vor(g)$ and of $f_b$ in $\Vor(b)$, respectively, are indicated in black.
Bottom right: A node $u \in \Vor(g)$ is indicated along with the subtree of $T_g$ rooted at $u$.
The primal edges of $R^*_u$ are shown in magenta. Their relative order with respect to
the labels $\pre_g(\cdot)$ is indicated.}
\label{fig:consecutive}
\end{center}
\end{figure*}

Assume, without loss of generality, that $c=g$. Consider the pair of arcs
$e^*_1, e^*_3 \in \beta^*(g,b)$ such that $\pre_g(pu) < \pre_g(e_1) < \pre_g(e_3) < \pre_g(up)$
(that is, $e_1^*$, $e_3^*\in R^*_u$), so that $\pre_g(e_1)$ is smallest and $\pre_g(e_3)$ is largest
under the above constraints. Denote $e_i = v_iw_i$, for $i=1,3$ (so $v_i$ is the endpoint in $\Vor(g)$).
By definition of preorder, both $v_1$ and $v_3$ are descendants of $u$ in $T_g$.
Consider the (undirected) cycle $\gamma$ formed by the $u$-to-$v_1$ path in $T_g$, the $u$-to-$v_3$ path in $T_g$,
and one of the two $v_1$-to-$v_3$ subpaths of $\beta_g$, chosen so that $\gamma$ does not enclose the
blue site $b$. 
Denote this subpath of $\beta_g$ by $\beta_g^u$. Since no branch of $T_g$ exits $f_g$, and no pair of paths in $T_g$ cross each other, it follows that the vertices of $\beta_g^u$ are precisely the descendants of $u$ in $T_g$ that belong to $\beta_g$. Hence, by definition of the preorder labels, $R^*_u$ is exactly the set of arcs of $\beta^*(g,b)$ whose corresponding primal edges have their tail in $\beta_g^u$. I.e., a subpath of $\beta^*(g,b)$, showing (1).

Let $e,e'$ be two arcs of $R^*_u$ such that $e$ appears before $e'$ on $\beta^*(g,b)$. Let $vw$, $v'w'$ be the corresponding primal arcs, respectively. Suppose w.l.o.g.\ that $\beta^*(g,b)$ is a CCW cycle. Since no branch of $T_g$ exits $f_g$, and no pair of paths in $T_g$ cross each other, the fact that $e$ precedes $e'$ in $\beta^*(g,b)$ implies that the $g$-to-$x_{vw}$ path in $T'_g$ is CCW to the $g$-to-$x_{v'w'}$ path in $T'_g$, so $\pre_g(e')<\pre_g(e)$, which proves (2).
\hfill \fullqed \end{proof}

We first provide the deferred proof of Lemma~\ref{lem:contig-bisector}, repeated here for convenience, and then continue with additional properties and terminology.

\contigbisector*
\begin{proof}
Let $\Vor^-(u)$ (resp., $\Vor^+(v)$) be the primal Voronoi cell of $u$ (resp., $v$) right before (resp., after) $\delta$.
Let $yz$ be the tense edge that triggers the switch, so $z$ is the root of the subtree that
 moves from $\Vor^-(u)$ to $\Vor^+(v)$ at  $\delta$.
Let $\beta^{*-}(u,v)$ and $\beta^{*+}(u,v)$ denote, respectively, the $uv$-bisector immediately
before and after $\delta$.
By Lemma~\ref{lem:consecutive} the preorder numbers (in $T_u$, say) of the edges along
$\beta^{*-}(u,v)$ are monotonically increasing and therefore
the arcs of $\beta^{*-}(u,v)$ whose tails are in the subtree of $z$ must form a continuous portion of
$\beta^{*-}(u,v)$.
An analogous argument applies to the arcs of $\beta^{*+}(u,v)$ whose heads are in the subtree of $z$.
\hfill \fullqed \end{proof}

\begin{lemma} \label{lem:bisectors}
Let $f^*$ and $g^*$ be two Voronoi vertices of $\VD^*(S)$, which are consecutive on
the common boundary between the cells $\Vor^*(u)$ and $\Vor^*(v)$. Then the path
between $f^*$ and $g^*$ along this boundary in $\VD^*(S)$ is a subpath
of $\beta^*(u,v)$.
\end{lemma}
\ifdefined\fullver
\begin{proof}
Let $xy$ be any edge whose dual is in on the path of $\VD^*(S)$ between $f^*$ and $g^*$, such that
 $x\in \Vor(u)$ and $y\in \Vor(v)$. Then $x$ is closer to
$u$ than to $v$ and $y$ is closer to $v$ than to $u$. It follows that
$x\in \Vor(u)$ and $y\in \Vor(v)$ also when $S$ contains only the two sites $u$ and $v$. It follows that the path between $f^*$ and $g^*$ in $\VD^*(S)$ is also a subpath of $\beta^*(u,v)$.
\hfill \fullqed \end{proof}
\else

We omit the (trivial) proof.
\fi
We next generalize the notion of bisectors to sets of sites.
Let $h_g, h_b$ be (not necessarily distinct) holes of $P$.
Let $G \subset S$ be a set of ``green'' sites incident to $h_g$ and let
 $B \subset S$ be a set of ``blue'' sites incident to $h_b$; when $h_g=h_b$ we require that $G$ and $B$ be separated along $\bd h_g$.
Define $\beta^*(G,B)$ to be the set of edges of $P^*$ whose corresponding primal arcs have their tail in $\Vor(g_i)$ and head in $\Vor(b_j)$, for some $g_i \in G$, and $b_j \in B$.

\begin{lemma} \label{lem:GBcycle}
If $h_g \neq h_b$, or if $h_g = h_b$ and the sets $G,B$ are separated along the boundary of $h_g$, then
$\beta^*(G,B)$ is a non-self-crossing cycle of arcs of $P^*$. If $|G|>1$ (resp., $|B| > 1$) then $h^*_g$ (resp., $h^*_b$)  may have degree greater than 2 in $\beta^*(G,B)$.
 All other dual vertices have degree $0$ or $2$ in $\beta^*(G,B)$. Furthermore, if $h_g = h_b$ and if $\beta^*(G,B)$ is nonempty, then $\beta^*(G,B)$ contains $h^*_g$ (possibly multiple times).
\end{lemma}
\begin{proof}
Embed a ``super-green'' vertex $g$ inside $h_g$, and a ``super-blue'' vertex $b$ inside $h_b$, and assign weight $0$ to both $g$ and $b$. This can be done without violating planarity also when $h_g=h_b$,
since in this case $G,B$ are separated along $\bd h_g$.  Connect $g$ (resp., $b$) to the green (resp., blue) sites with arcs $gg_i$ (resp., $bb_i$) of weight $\wt(g_i)$ (resp., $\wt(b_i)$). By Lemma~\ref{lem:bicycle}, the bisector $\beta^*(g,b)$ is a simple cycle in this auxiliary graph. However, this cycle can go through the artificial faces created inside the holes $h_g, h_b$. Deleting the artificial arcs in the primal is equivalent to contracting them in the dual, which contracts all the artificial faces into the dual faces $h^*_g$  and $h^*_b$. This contraction turns $\beta^*(g,b)$ into $\beta^*(G,B)$, as is easily checked, and might give rise to non-simplicities of $\beta^*(G,B)$ at $h^*_g$ and $h^*_b$. See Figure~\ref{fig:nonsimple-bis}.

\begin{figure}[h]
\begin{center}
\includegraphics[scale=0.5, clip=true, trim = 0mm 150mm 70mm 0mm]{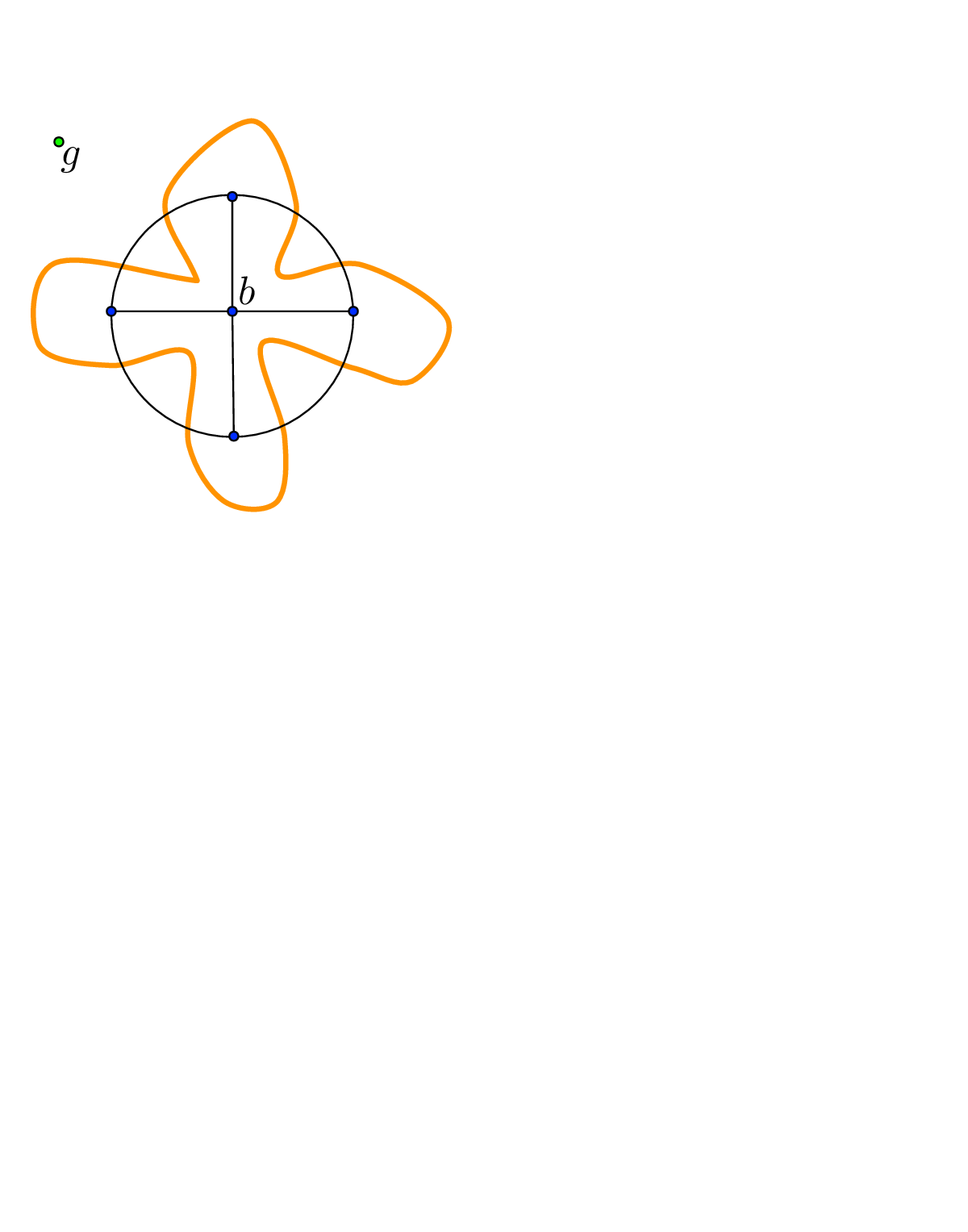}
\hspace{0.1in}
\includegraphics[scale=0.5, clip=true, trim = 0mm 150mm 70mm 0mm]{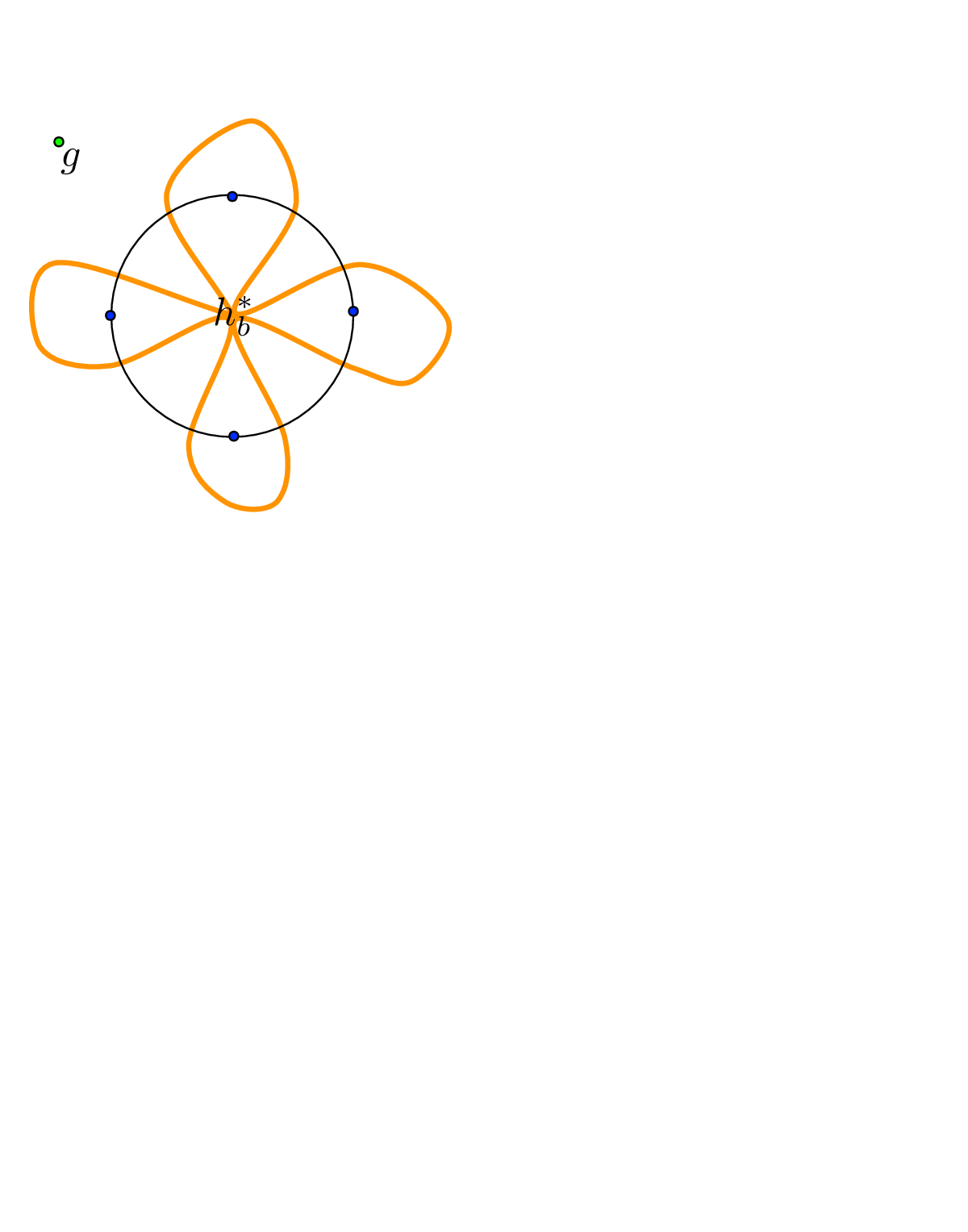}
\caption{\small The structure of the bisector $\beta^*(G,B)$ when $|G|=1$. Left: a set $B$ of four blue sites on the hole $h_b$ (black cycle). An artificial blue site $b$ is embedded in $h_b$ and connected to all blue sites. The bisector $\beta^*(g,b)$ is a simple cycle. Right: When removing the artificial arcs the bisector $\beta^*(G,B)$ visits $h^*_b$ multiple times.
\label{fig:nonsimple-bis}}
\end{center}
\end{figure}

The proof of the final property is identical to the one in Lemma~\ref{lem:bicycle}. Namely, If $h_g=h_b$ and $\beta^*(G,B)$ is nonempty, then the simple cut defined by the partition $(\Vor(g),\Vor(b))$ must contain at least one arc $e$ on the boundary of $h_g$. Therefore, $e^*$ is an arc of $\beta^*(u,v)$ that is incident to $h_g^*$, and multiplicities can arise when there are  several such arcs $e$.
\hfill \fullqed \end{proof}

Let $h_g,h_b$ be (not necessarily) distinct holes. Let $g\in S$ be a site on $h_g$, and let $B \subset S$ be a set of sites on $h_b$. Consider $\VD(\{g\}\cup B)$. Let $\beta^*(g,B)$ be the set of edges of $P^*$ whose corresponding primal arcs have their tail in $\Vor(g)$ and their head in $\Vor(b_j)$ for some $b_j \in B$.
\begin{lemma} \label{lem:gBcycle}
For any $b \in B$, $\beta^*(g,B)$ contains at most a single segment (contiguous subpath) of $\beta^*(b,B)$.
\end{lemma}
\begin{proof}
By Lemma~\ref{lem:GBcycle}, $\beta^*(g,B)$ is a non-self-crossing cycle.  Designate an arbitrary arc of $\beta^*(g,B)$ as its beginning (just to define a linear order on the arcs of $\beta^*(g,B)$).  
Let $e_1^*$ (resp., $e_2^*$) be the first (resp., last) arc of $\beta^*(g,B)$ such that the primal arc $e_1$ (resp., $e_2$) has an endpoint $x_1$ (resp., $x_2$) in $\Vor(b)$.  
Let $y_1$ (resp., $y_2$) be the other endpoint $e_1$ (resp., $e_2$), namely, the one belonging to $\Vor(g)$ in $\VD(\{g\}\cup B)$.
Consider the cycle $C$ formed by the $b$-to-$x_1$ path of $T_b$, $e_1$, the $y_1$-to-$g$ path in $T_g$, the $g$-to-$y_2$ path in $T_g$, $e_2$, and the $x_2$-to-$b$ path in $T_b$. By choice of $e_1^*,e_2^*$, $C$ encloses all arcs of $\beta^*(g,b) \cap \beta^*(g,B)$, but since the only sites in $\{g\}\cup B$ enclosed by $C$ are $g$ and $b$, and since every Voronoi cell is connected, all vertices enclosed by $C$ belong to either $\Vor(g)$ or to $\Vor(b)$ in $\VD(\{g\}\cup B)$. Therefore, all edges of $\beta^*(g,B)$ between $e^*_1$ and $e^*_2$ belong to $\beta^*(g,b)$, proving the claim. See Figure~\ref{fig:beta-gB}. 
\end{proof}

\begin{figure}[h]
\begin{center}
\includegraphics[scale=0.7]{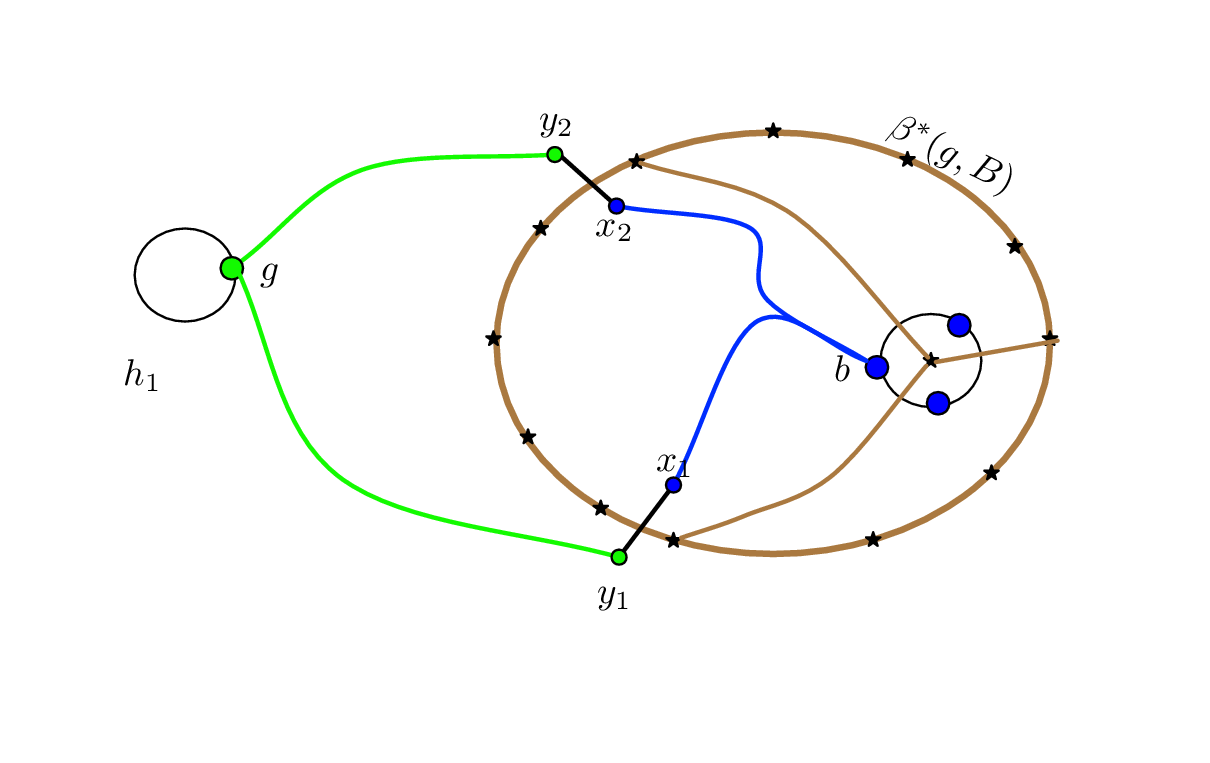}
\caption{\small Illustration of the proof of Lemma~\ref{lem:gBcycle}. $\VD^*(g,B)$ is shown in brown. The paths of $T_g$ and $T_b$ comprising the cycle $C$ are shown in green and blue, respectively.
\label{fig:beta-gB}}
\end{center}
\end{figure}

 % trichromatic
\section{Computing Voronoi vertices}\label{sec:trichrom}

Consider the settings of Theorem~\ref{thm:vor}. We will henceforth only deal with a single subgraph $P$, so to simplify notation we denote the size of $P$ by $n$ (rather than $r$). 
Recall that, by Lemma~\ref{lem:2vvert}, a Voronoi diagram with three sites has at most two Voronoi vertices. In this section we prove the following theorem. 

\begin{theorem}\label{thm:tri}
Let $P$ be a directed planar graph with real arc lengths, $n$ vertices, and no negative length cycles. Let $S$ be a set of  sites that lie on the boundaries of a constant number of faces (holes) of $P$. 
One can preprocess $P$ in $\tilde O(n|S|^2)$ time so that, given any three sites $r,g,b \in S$ with additive weights 
$\wt(\cdot)$, one can find the (at most two) Voronoi vertices of $\VD^*(\{r,g,b\})$ in $\tilde O(1)$ time.
\end{theorem}

We will actually prove a more general theorem that extends the single site $b$ to a subset of sites $B$ on the same hole. 
Let $r,g$ be two sites, and let $B \subseteq S \setminus\{r,g\}$ be a set of sites on some hole $h$. Consider adding an artificial site $v_B$ embedded inside $h$ and connected to all sites in $B$ with artificial arcs whose lengths are the corresponding additive weights of the sites in $B$. Denote by $P'$ be the resulting graph. We define the bisector  $\beta^*(g,B)$ in $P$ to be the bisector $\beta^*(g,v_B)$ in $P'$ (ignoring the artificial arcs). Similarly, we define the diagram $\VD(r,g,B)$ in $P$ to be the diagram $\VD(\{r,g,v_B\})$ in $P'$. The cell $\Vor(B)$ contains all vertices closer (by additively weighted distance) to some $b_i \in B$ than to any other site $u' \in \{r,g\}$. By Lemma~\ref{lem:2vvert}, the dual diargram $\VD^*(r,g,B)$ consists of at most two vertices, each corresponding to a face in $P$ that contains a vertex in $\Vor(r)$, a vertex in $\Vor(g)$, and a vertex in $\Vor(v_B)$.

\begin{theorem}\label{thm:tri-extended}
Consider the settings of Theorem~\ref{thm:tri}. One can preprocess $P$ in $\tilde O(n|S|^2)$ time so that the following procedure takes $\tilde O(1)$ time. The inputs to the procedure are: (i) two sites $r,g\in S$, and a set $B\subset S\setminus\{r,g\}$ of sites on a single hole $h$, with respective additive weights $\wt(\cdot)$, (ii) the site $x \in \{g\} \cup B$ minimizing the (additive) distance from $x$ to $r$, and (iii) a representation of the bisector $\beta^*(g,B)$ that allows one to retrieve the $k$-th vertex of $\beta^*(g,B)$ in $\tilde O(1)$ time. The output of the procedure are the (at most two) vertices of $\VD^*(r,g,B)$.
\end{theorem}

\subsection{Overview} \label{sec:overview}
\ifdefined\fullver
We begin with the case of Theorem~\ref{thm:tri}, when $B$ consists of a single site $b$. In this case, the Voronoi diagram has at most three cells.
\else
Due to space constraints, we only prove Theorem~\ref{thm:tri} in this extended abstract. The modifications to the proof required for establishing Theorem~\ref{thm:tri-extended} will appear in the full version. 
In the case of Theorem~\ref{thm:tri}, there are three sites, and the Voronoi diagram has at most three cells.
\fi
We call a face $f$ of $P$
\emph{trichromatic} if it has an incident vertex in each of the three  cells of the  diagram.
\emph{Monochromatic} and \emph{bichromatic} faces are defined similarly.
(This definition also includes faces that are holes.) We say that a vertex $v$ of $P$ is
\emph{red}, \emph{green}, or \emph{blue} if $v$ is in the Voronoi cell of $r$, $g$, or $b$, respectively.
By definition, the Voronoi vertices that we seek are precisely those dual to the trichromatic faces of $P$.
Let $\beta^*(g,b)$ denote the bisector of $g$ and $b$ (with respect to the additive weights $\wt(g),\wt(b)$).
By Lemma~\ref{lem:bicycle}, $\beta^*(g,b)$ is a simple cycle in $P^*$.
For $c \in \{g,b\}$, and for each vertex $v \in V(P)$, define $\tilde{\delta}^{rc}(v) := \wt(c)+d(c,v) - d(r,v)$. 
Equivalently, $\tilde{\delta}^{rc} = \wt(c) - \delta^{rc}$.
Define $\Delta^{r}(v) := \min\{\tilde{\delta}^{rg}(v),\tilde{\delta}^{rb}(v)\}$. That is, $\Delta^{r}(v)$ is the maximum
weight that can be assigned to $r$ so that $v$ is red (and $v$ will be red also for any smaller assigned weight).
Indeed, if $\wt(r) > \tilde{\delta}^{rg}(v)$, say,  then $\wt(r) + d(r,v) > \wt(g) + d(g,v)$, then $v$ is not red.
For each edge $uv$ of $P$, define $\Delta^{r}(uv) := \max \{\Delta^{r}(u),\Delta^{r}(v)\}$. That is,
$\Delta^{r}(uv)$ is the maximum weight that can be assigned to $r$ so that at least one endpoint of
$uv$ is red. For an edge $e^*$, dual to a primal edge $e$, we put $\Delta^r(e^*) = \Delta^r(e)$.

For any real $x$, we denote by $\VD_x$ the Voronoi diagram of $r,g,b$, with
respective additive weights $x,\wt(g),\wt(b)$. 
We  define
\begin{align*}
Q^*_{\geq x} & := \{ e^* \in \beta^*(g,b) \mid  \Delta^{r}(e^*) \geq x \} \ , \\
Q^*_{= x} & := \{ e^* \in \beta^*(g,b) \mid  \Delta^{r}(e^*) = x \} \ .
\end{align*}
The following lemma proves that
$Q^*_{\geq x}$ is a subpath of $\beta^*(g,b)$.
\begin{lemma} \label{lem:Delta-consecutive}
For any $\infty > x > -\infty$, the
edges  of $Q^*_{\geq x}$
form a subpath  of $\beta^*(g,b)$. Furthermore, if $Q^*_{\geq x}$ is non-empty and
does not contain all the edges of $\beta^*(g,b)$, then the trichromatic faces in $\VD_x$ are the duals
of the endpoints of $Q^*_{\geq x}$; that is, these are the faces whose duals have exactly one incident edge in $Q^*_{\ge x}$.
\end{lemma}
\begin{proof}
Let $Q^*$ be a maximal subset of edges in $Q^*_{\geq x}$ that form a (contiguous) subpath
of $\beta^*(g,b)$.
Assume that there exists at least one edge of $\beta^*(g,b)$ that is
not in $Q^*$; otherwise $Q^* = \beta^*(g,b)$, which we assume not to be the case.

Enumerate the dual vertices incident to the edges of $Q^*$ as $f^*_1, f^*_2, \ldots, f^*_j$
in their (cyclic) order along $\beta^*(g,b)$.
The vertex $f^*_1$ has an incident edge $f^*_1f^*_2$ in $Q^*_{\geq x}$, and another incident edge,
call it $f^*_0f^*_1$, that is in $\beta^*(g,b)$ but not in $Q^*_{\geq x}$.
In the primal, the face $f_1$, dual to $f^*_1$, has an incident edge $uv$ that is dual to $f^*_1f^*_2$,
such that, by construction, at least one of $u,v$ is red, and another incident edge $u'v'$, dual to
$f^*_0f^*_1$, such that none of $u'$, $v'$ is red (note that when $f_1$ is a triangle, one of $u',v'$
coincides with one of $u,v$). Moreover, since $f^*_0f^*_1$ is an edge of the bisector $\beta^*(g,b)$,
exactly one of $u',v'$ is blue and the other one is green. Therefore, the face $f_1$ is trichromatic.
A similar argument shows that $f_j$ is also trichromatic. Note that the argument does not rely on the faces being triangles, so it also applies in the presence
of holes. See Figure~\ref{fig:tri}.

\begin{figure}[htb]
\begin{center}
\includegraphics[scale=0.7, clip=true, trim = 0mm 20mm 0mm 20mm]{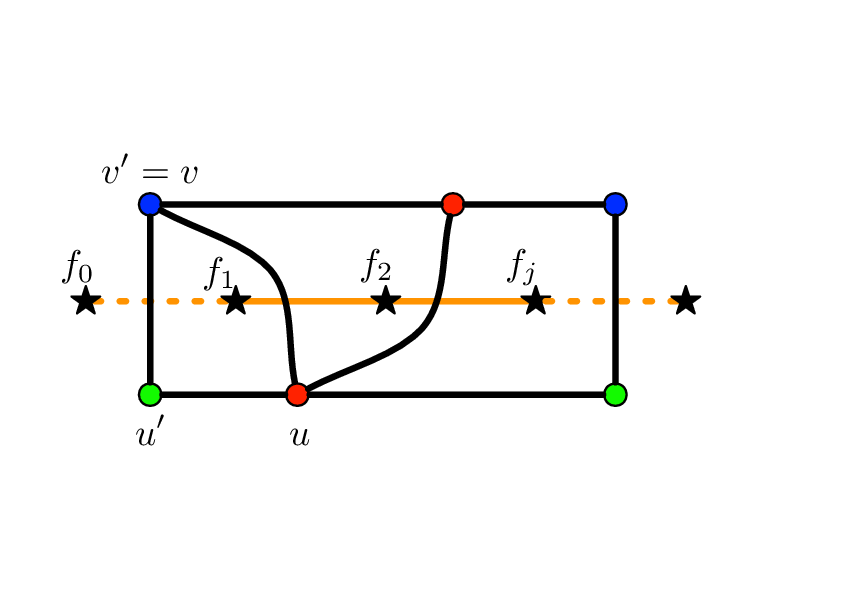}
\caption{\small Illustration of the proof of Lemma~\ref{lem:Delta-consecutive}. A set of consecutive edges in $\beta^*(g,b)$ is shown in orange. Edges of $Q^*$ are solid, and edges not in $Q^*$ are dashed. Dual vertices are indicated by stars.
Primal vertices (circles) are colored according to the Voronoi region they belong to. 
 Primal edges are shown in black. In this example $u,u'$ are distinct vertices, but $v'=v$. 
\label{fig:tri}}
\end{center}
\end{figure}

This shows that every maximal subset of edges in $Q^{*}_{\geq x}$ that forms a subpath of $\beta^*(g,b)$
gives rise to two distinct trichromatic faces.
Since, by Lemma~\ref{lem:2vvert}, there are at most two trichromatic faces we must have that $Q^* = Q^*_{\geq x}$ and the lemma follows.
\hfill \fullqed \end{proof}

%\subsection{Warmup: the case of sites on a single hole}\label{sec:tri-single}
By Lemma~\ref{lem:Delta-consecutive}, to find the trichromatic faces of $\VD_x$ it suffices to find the endpoints of $Q^*_{\geq x}$.
We begin with describing the case in which $r,g$ and $b$ are all incident to the same hole $h$. The multiple hole case follows the same principles, but is much more complicated. 

In the single hole case, if there exists a trichromatic face in $\VD_x$, then $h^*$ is trichromatic (because all three sites are incident to $h^*$. 
We therefore treat the cycle $\beta^*(g,b)$ as a path starting and ending at $h^*$. This defines a natural order on the arcs and vertices of $\beta^*(g,b)$.  
Lemma~\ref{lem:Delta-consecutive} implies, in the single hole case, that the
edges  of $Q^*_{\geq x}$
form a prefix or a suffix of $\beta^*(g,b)$.

We define $\Dbisr$ to be the cyclic sequence
$\Bigl(\Delta^{r}(e^*) \mid e^* \in \beta^*(g,b)\Bigr)$. Note that $\Delta^{r}(e^*)$ is not known at preprocessing time since it depends on three sites and on their relative weights.
\begin{corollary}\label{cor:monotone}
When $r,g,b$ all lie on a single hole, the sequence $\Dbisr$ is weakly monotone.
\end{corollary}
\begin{proof}
In the single hole case $Q^*_{\geq x}$ is a prefix or a suffix of $\beta^*(g,b)$, for each $\infty > x >-\infty$. The lemma follows because, by definition, we have
 $Q^*_{\geq x'}\subseteq Q^*_{\geq x}$ for every pair $x'\ge x$. 
\hfill \fullqed \end{proof}

To find the trichromatic vertices of $\VD_x$ one needs to find the endpoints of $Q^*_{\geq x}$. One of the endpoints is $h^*$. By Corollary~\ref{cor:monotone}, one can find the other endpoint by binary searching for $x$ in $\Dbisr$. Note that a single element of $\Dbisr$ can be computed on the fly in constant time given the shortest path trees rooted at $r,g$ and $b$. Also note that our representation of bisectors supports retrieving the $k$-th arc of the bisector in $\tilde O(1)$ time. Therefore, the binary search can be implemented in $\tilde O(1)$ time as well. This proves Theorem~\ref{thm:tri} for the case of a single hole.

%\subsection{The case of sites on multiple holes}\label{sec:tri-multi}
In the remainder of this section we treat the general case where sites are not necessarily on the same hole. In this case $Q^*_{\geq x}$ is a subpath of $\beta^*(g,b)$, but not necessarily a prefix or a suffix. As a consequence, $\Dbisr$ is no longer weakly monotone, but \emph{weakly bitonic}. This makes the binary search procedure much more involved. We first establish the bitonicity of $\Dbisr$ and describe the bitonic search. We then elaborate on the various steps of the search. 

Using standard notation, we say
 that a linear sequence is strictly (weakly) \emph{bitonic} if it consists of a strictly (weakly)
decreasing sequence followed by a strictly (weakly) increasing sequence. A cyclic sequence is strictly
(weakly) bitonic if there exists a cyclic shift that makes it strictly (weakly) bitonic; this shift starts and ends at the maximum (a maximum) element of the sequence.
Recall that 
we defined $\Dbisr$ to be the cyclic sequence
$\Bigl(\Delta^{r}(e^*) \mid e^* \in \beta^*(g,b)\Bigr)$. Note that $\Delta^{r}(e^*)$ is not known at preprocessing time since it depends on three sites and on their relative weights.
The generalization of Corollary~\ref{cor:monotone} to the case of sites on multiple holes is:
\begin{corollary}\label{cor:bitonic}
The cyclic sequence $\Dbisr$ is weakly bitonic.
\end{corollary}
\begin{proof}
By Lemma \ref{lem:Delta-consecutive}, $Q^*_{\geq x}$ is a subpath of $\beta^*(g,b)$, for each $\infty > x >-\infty$, and by definition we have
 $Q^*_{\geq x'}\subseteq Q^*_{\geq x}$ for every pair $x'\ge x$. This clearly implies the corollary.
\hfill \fullqed \end{proof}

We will show (Lemma~\ref{lem:findmax-extended}) that we can find the maximum of $\Dbisr$ in
$\tilde O(1)$ time. This will allow  us to turn the weakly bitonic cyclic sequence $\Dbisr$ into a weakly bitonic linear sequence.
We will then use binary search on $\Dbisr$ to find the endpoints of $Q^*_{\ge \wt(r)}$, which, by the
second part of Lemma~\ref{lem:Delta-consecutive}, are the trichromatic faces of the Voronoi diagram
of $r,g,b$ with respective additive weights $\wt(r) ,\wt(g), \wt(b)$.
The search might of course fail to find these vertices when they do not exist, either because $\wt(r)$
is too small, in which case $\Vor^*(r)$ completely ``swallows'' $\beta^*(g,b)$, or because $\wt(r)$
is too large, in which case $\beta^*(g,b)$ appears in full in $\VD^*(\{r,g,b\},\wt)$. In both these cases, there are either no Voronoi vertices, or there is a single Voronoi vertex which is dual to a hole. 

Before we show how to find the maximum of $\Dbisr$, 
we briefly discuss a general strategy for conducting binary search on linear bitonic sequences. This search is not trivial, especially when the sequence is only weakly bitonic. We first consider the case of strict bitonicity, and then show how to extend it to the weakly bitonic case.

\subsubsection{Searching in a strictly bitonic linear sequence}
Given a strictly bitonic linear sequence $\sigma$ and a value $y$, one can find the two ``gaps'' in $\sigma$ that contain $y$, where each such gap is a pair of consecutive elements of $\sigma$ such that $y$ lies between their values. (For simplicity of presentation, and with no real loss of generality, we only consider the case where $y$ is not equal to any element of the sequence.)
This is done by the following variant of binary search. 

The search consists of two phases. In the first phase the interval that the binary search
maintains still contains both gaps (if they exist), and in the second phase we have already managed
to separate between them, and we conduct two separate standard binary searches to identify each of them.

Consider a step where the search examines a specific entry $x=\sigma(i)$ of $\sigma$.
\begin{description}
\item{(i)}
If $x>y$, we compute the discrete derivative of $\sigma$ at $i$. If the derivative is
positive (resp., negative), we update the upper (resp., lower) bound of the search to $i$. (This rule holds for both phases.) 
\item{(ii)}
If $x < y$ and we are in the first phase, we have managed to separate the two gaps, and we move on
to the second phase with two searches, one with upper bound $i$ and one with lower bound $i$.
If we are already in the second phase,
we set the upper (resp., lower) bound to $i$ if we are in the lower (resp., higher) binary search.
\end{description}

\subsubsection{Searching in a weakly bitonic linear sequence}\label{sec:weak-bitonic-search}
This procedure does not work for a weakly bitonic sequence (consider, e.g., a sequence all of whose
elements, except for one, are equal). This is because the discrete derivative in (i) might be locally 0 in the weakly bitonic case. This difficulty can be overcome if, given an element $i$ in $\sigma$
such that $\sigma(i) = x$,
we can efficiently find the endpoints of the maximal interval $I$ of $\sigma$ that contains $i$ and all its
elements are of value equal to $x$ (note that in general $\sigma$ might contain up to two intervals of elements
of value equal to $x$, only one of which contains $i$). We can then compute the discrete derivatives at the endpoints of $I$, and use them to guide the search,
similar to the manner  described for the strict case.

Unfortunately, now focusing on the specific context under consideration, given an edge $\hat{e}^* \in \beta^*(g,b)$
such that $\Delta^r(\hat{e}^*)=x$, we do not know how to find the maximal interval $I$ of $\beta^*(g,b)$ such $\hat{e}^* \in I$ and for every edge $e^* \in I$, $\Delta^r(e^*)=x$.
Instead, we provide a procedure that returns an interval $I^+$ of (the cyclic) $\beta^*(g,b)$ that contains 
all edges $e^* \in \beta^*(g,b)$ for which $\Delta^r(e^*)=x$, and no edge $e^*$ for which $\Delta^r(e^*)<x$. Note that $I^+$ might contain edges $e^* \in \beta^*(g,b)$ for which $\Delta^r(e^*)>x$. When there is just one interval of elements of value equal to $x$ then $I^+$ starts or ends with an edge $e^*$ for which $\Delta^r(e^*)=x$. 
In this case, $I^+ = I$, or $\hat{e}^*$ is in the increasing part of $\Dbisr$, or $\hat{e}^*$ is in the increasing part of $\Dbisr$, depending on whether the values of $\Delta^r(\cdot)$ at the two endpoints of $I^+$ are equal, the first is smaller than the last, or the last is smaller than the first, respectively. 
When there are two intervals of elements of value equal to $x$, then $I^+$ also contains all edges $e^* \in \beta^*(g,b)$ for which $\Delta^r(e^*)>x$, and, in particular, an edge $e^*_{max}\in \beta^*(g,b)$ maximizing $\Delta^r(\cdot)$. In this case if $\hat{e}^*$ appears before (resp., after) $e^*_{max}$ in $I^+$ then $\hat{e}^*$ is in the increasing (resp., decreasing) part of $\Dbisr$, and we should set the upper (resp., lower) bound of the search in (i) to the beginning (resp., end) of the interval $I^+$.

In Section~\ref{sec:trimax} we describe the procedure that finds an edge $e^*_{\max}\in \beta^*(g,b)$ with a largest value of $\Delta^r$.
The mechanism for finding $I^+$ with the properties described above is described in Section~\ref{sec:mechanism}.

 % trimax
\subsection{Finding the maximum in $\Dbisr$} \label{sec:trimax}
We now describe a procedure for finding $x_{max} := \max(\Dbisr)$, or more precisely, finding some edge $e^*_{\max} \in \beta^*(g,b)$ such that $\Delta^r(e^*_{\max}) = x_{max}$.
Let $h_r$ be the hole to which the site $r$ is incident.
We check whether $r$ belongs to $\Vor(g)$ or to $\Vor(b)$ in $\VD(\{g,b\})$ by comparing the distances from $g$ to $r$ and from $b$ to $r$. 
Assume that $r$ belongs to $\Vor(g)$ (the case where $r$ belongs to $\Vor(b)$ is symmetric).
Let $T^*_g$ be the cotree of the shortest-path tree $T_g$. Define the label $\ell_g^{h_r}(f^*)$, for each dual vertex $f^*$, to be equal to the number of edges on the $f^*$-to-$h_r^*$ path in $T^*_g$. Note that, because the number of holes is constant, these labels can be computed during preprocessing, when we compute the tree $T^*_g$, without changing the asymptotic preprocessing time. Furthermore, we can augment the persistent search tree representation of the bisectors with these labels, so that, given $\beta^*(g,b)$ and $r$, we can retrieve in 
$\tilde O(1)$ time the vertex of $\beta^*(g,b)$ minimizing $\ell_g^{h_r}$.

\begin{lemma}\label{lem:trimax}
The value of $\Delta^r(\cdot)$ is $x_{max}$ for at least one of the two arcs of $\beta^*(g,b)$ incident to the dual vertex $f^* \in \beta^*(g,b)$ minimizing $\ell_g^{h_r}(\cdot)$.
\end{lemma}
\ifdefined\fullver
\begin{proof}
Consider the dual vertex $f^* \in \beta^*(g,b)$ minimizing $\ell_g^{h_r}(\cdot)$. 
If $\beta^*(g,b)$ goes through $h_r^*$ then $f^* = h^*_r$. Otherwise, 
let $Q^*$ be the $f^*$-to-$h_r^*$ path in $T^*_g$. 
By  choice of $f^*$, $Q^*$ is internally disjoint from $\beta^*(g,b)$. Hence for every $e$ s.t. $e^* \in Q^*$, both endpoints of $e$ belong to the cell $\Vor(g)$ in $\VD(g,b)$ (under our assumption that $r$ belongs to $\Vor(g)$). See Figure~\ref{fig:trimax}. 

\begin{figure}[h]
\begin{center}
\includegraphics[width=0.7\textwidth, clip=true, trim = 0mm 20mm 0mm 0mm]{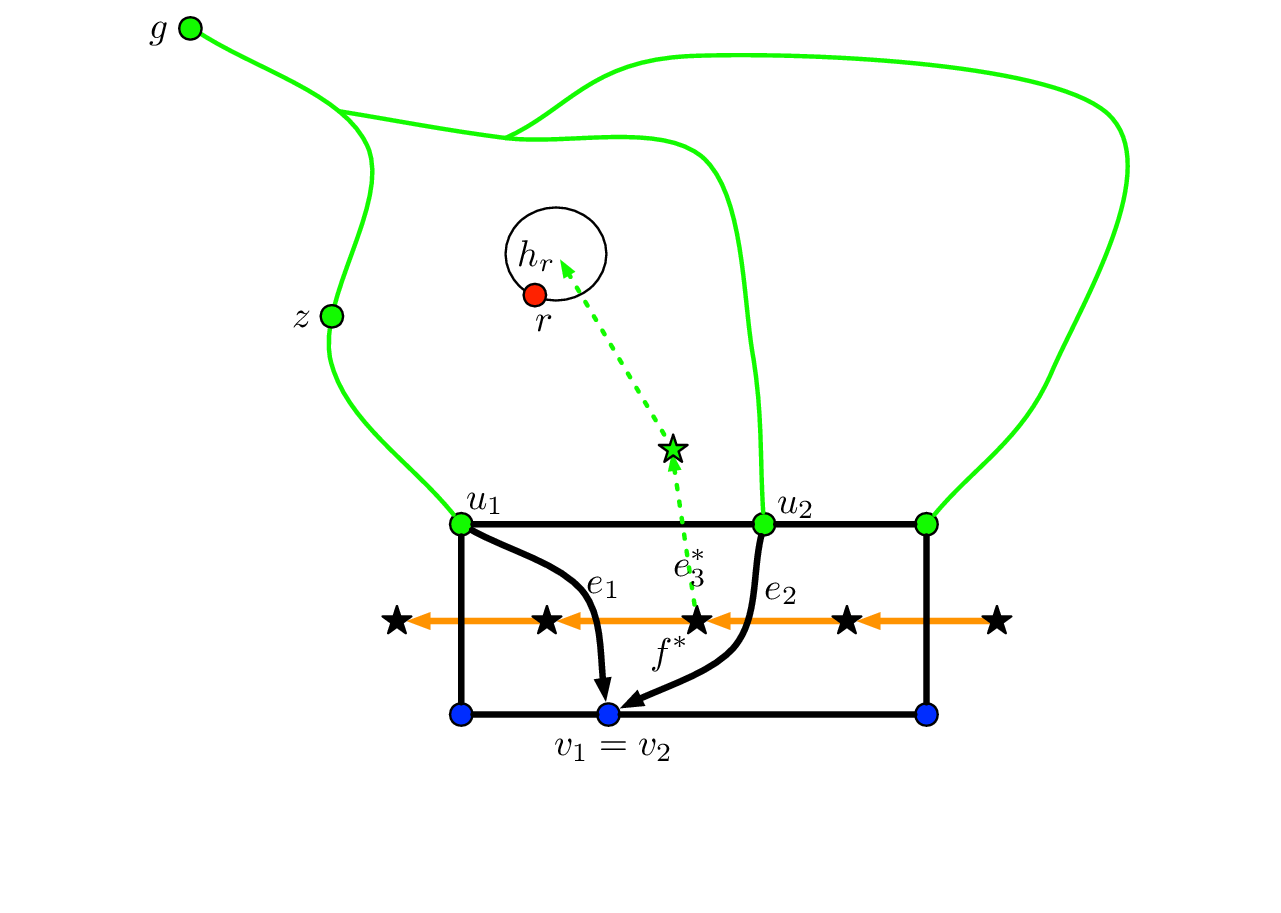}
\caption{Illustration of the proof of Lemma~\ref{lem:trimax}. Part of $\beta^*(g,b)$ is shown in orange. Primal vertices are represented by circles. Green circles belong to $\Vor(g)$ and blue to $\Vor(b)$. The blue site is not shown. Dual vertices are indicated by stars. The face $f^*$ minimizing $\ell_g^{h_r}(\cdot)$ is indicated, as well as the edges $e_1,e_2$, and their endpoints. Parts of the tree $T_g$ are shown in solid green. The only part of $T^*_g$ shown is  the $f^*$-to-$h_r^*$ path (dashed green). As the additive weight $x$ of the red site $r$ decreases, the Voronoi cell of $r$ expands. The node $z$ is the first node on the fundamental cycle of $e_3$ that enters the Voronoi cell of $r$. This happens at critical value $x_{max}$. At that time $u_1$ also becomes red, so $\Delta^r(u_1)=x_{max}$.
\label{fig:trimax}}
\end{center}
\end{figure}

Since $f^* \in \beta^*(g,b)$, there are exactly two edges of $f$ whose duals are in $\beta^*(g,b)$. Namely, edges with one endpoint in $\Vor(g)$ and one endpoint in $\Vor(b)$. Denote these edges by $e_1=u_1v_1, e_2=u_2v_2$, where $u_1, u_2 \in \Vor(g)$. 
If $Q^*$ is empty (i.e, $f^* = h^*_r$) then consider the maximal subpath of $h_r$ that belongs to $\Vor(g)$. This subpath is nonempty since we assume $r \in \Vor(g)$. If $r$ is the only vertex of $h_r$ in $\Vor(g)$, then $u_1=u_2=r$, and the lemma is immediate since $r$ is an endpoint of both $e_1$ and $e_2$. We therefore assume that if $Q^*$ is empty, then $u_1 \neq u_2$.
If $Q^*$ is not empty, let $e^*_3$ be the edge of $Q^*$ incident to $f^*$. Observe that $e_1,e_2,e_3$ are all edges of the face $f$, and that both endpoints of $e_3$ are in $\Vor(g)$. Hence, $u_1 \neq u_2$ also in this case.

Let $F_g$ denote the $u_1$-to-$u_2$ subpath of the face $f$ that belongs to $\Vor(g)$. ($F_g$ consists of the single edge $e_3$, unless $f^*$ is a hole.) Note that, if $f^* \neq h_r^*$ then $F_g$ contains $e^*_3$, and if $f^* = h^*_r$ then $F_g$ contains $r$.
Let $C$ be the cycle formed by the root-to-$u_1$ path in $T_g$, the root-to-$u_2$ path in $T_g$, and $F_g$. 
Note that the cycle $C$ does not enclose $f^*$, and, since $F_g$ contains either $r$ or the first edge of $Q^*$, the cycle $C$ does enclose $r$. By its definition, the cycle $C$ consists entirely of edges and vertices that belong to $\Vor(g)$. Recall that $f^*$ is a vertex of $\beta^*(g,b)$, so $f$ has an incident vertex $v$ that belongs to $\Vor(b)$. The vertex $v$ is not on $C$ because $v \notin \Vor(g)$, and is not strictly enclosed by $C$ since $f$ is not enclosed by $C$. This shows that the site $b$ is not enclosed by $C$ (otherwise, the $b$-to-$v$ shortest path must intersect $C$, but this is impossible since all vertices of this path belong to $\Vor(b)$, while all vertices of $C$ belong to $\Vor(g)$). It follows that all the vertices enclosed by $C$ belong to $\Vor(g)$. This implies that $C$ does not enclose any arc $e'$ whose dual is an arc of $\beta^*(g,b)$, because such an $e'$ has one endpoint in $\Vor(b)$. 

Let $z$ be a vertex that maximizes $\tilde \delta^{rg}(\cdot)$ among the vertices of $C$ that are not internal vertices of $F_g$. Note that, since $r$ is enclosed by $C$,  $\tilde \delta^{rg}(z) \geq \tilde \delta^{rg}(x)$ for any vertex $x$ that is not enclosed by $C$. This is because any $r$-to-$x$ path intersects $C$ at a vertex that is not an internal vertex of $F_g$. Thus $\tilde \delta^{rg}(z) \geq x_{max}$. 	
	By construction of $C$, $z$ is an ancestor of either $u_1$ or $u_2$. Assume, without loss of generality that $z$ is an ancestor of $u_1$ in $T_g$. Therefore, $\tilde \delta^{rg}(u_1) \geq \tilde \delta^{rg}(z)$ (because a green vertex becomes red no later than any of its ancestors in $T_g$). But $u_1$ is an endpoint of $e_1$.  	
	So $\Delta^r(e_1^*) = \Delta^r(e_1) \geq \tilde \delta^{rg}(u_1) \geq \tilde \delta^{rg}(z) \geq x_{max}$. But, since $e_1^* \in \beta^*(g,b)$, by definition of $x_{max}$ we have $x_{max} \geq \Delta^r(e_1^*)$. Therefore, $\Delta^r(e_1^*)=x_{max}$. 
\hfill \fullqed \end{proof}

To summarize, to find an arc of $\beta^*(g,b)$ maximizing $\Delta^r(\cdot)$, the algorithm does the following: (1) it finds the site $c\in \{g,b\}$ closer to $r$; (2) it finds the dual vertex $f^*$ minimizing $\ell_c^{h_r}(\cdot)$ on $\beta^*(g,b)$ in %$O(\log n)$ 
$\tilde O(1)$ time, and (3) it returns the arc of $\beta^*(g,b)$ incident to $f^*$ with larger $\Delta^r$-value. This establishes the following lemma.
\begin{lemma}\label{lem:findmax}
Consider the settings of Theorem~\ref{thm:tri}.	We can preprocess $P$ in $O(n|S|^2)$ time  so that one can find an arc of $\beta^*(g,b)$ maximizing $\Delta^r(\cdot)$ in %$O(\log n)$ 
$\tilde O(1)$ time.  
\end{lemma}

\subsubsection{The case of multiple sites} \label{subseq:trimax-multi}
We next consider the same problem in the more general setting of two individual sites, say $r$ and $g$, and a set $B = \{b_1, \dots, b_k\}$ of sites on hole $h_b$.
We assume the bisector $\beta^*(g,B)$ is represented by a binary search tree over the segments of bisectors $\beta^*(g,b_i)$ for $b_i \in B$ that form $\beta^*(g,B)$. Thus, we can access the $k$-th arc or vertex of $\beta^*(g,B)$ in $\tilde O(1)$ time.
Recall the definition  
$\tilde{\delta}^{rc}(v) := \wt(c)+d(c,v) - d(r,v)$. 
In the context of multiple sites on a hole we redefine $\Delta^{r}(v) := \min\{\tilde{\delta}^{rg}(v),\tilde{\delta}^{rb}(v)\}$, where the vertex $b$ is the vertex of $B$ closest (in additive distance) to $v$. 
We wish to find an edge of $\beta^*(g,B)$ maximizing $\Delta^r$.

Consider first the case where $g$ is closer (in additive distance) to $r$ than any $b \in B$. The treatment of this case is identical to that of the single site case. Let $h_r$ be the hole to which site $r$ is incident.
We find the dual vertex $f^*$ minimizing $\ell_g^{h_r}(\cdot)$ over all vertices of $\beta^*(g,B)$ in $\tilde O(1)$ time using the decorations for $\ell_g^{h_r}$ in the binary search tree representation of $\beta^*(g,B)$.
As in the case of single sites we return the arc of $\beta^*(g,b)$ incident to $f^*$ with larger $\Delta^r(\cdot)$.

Consider now the case where there exists a site of $B$ that is closer to $r$ than $g$. Let $b$ be the site of $B$ minimizing the additive distance to $r$. Let $B'$ be the subsequence of $B$ (ordered by the cyclic order around $h_b$)
consisting of all sites $b'$ such that
 there is a segment of
$\beta^*(g,b')$ in $\beta^*(g,B)$. By non-crossing properties of shortest paths, the cyclic order of these segments along
 $\beta^*(g,B)$ is consistent with the cyclic order of the sites of $B'$ along $h_b$.
Consider the case where $b \in B'$.
This case is similar to the single site case.
Let $\gamma^*$ be the segment of  $\beta^*(g,b)$ contained in $\beta^*(g,B)$. Note that $\gamma^*$ is simple, as it is a subpath of a simple cycle.
We find the dual vertex $f^*$ minimizing $\ell_{b}^{h_r}(\cdot)$ on $\gamma^*$.
Let $f^*_1$ and $f_2^*$ be the two endpoints of $\gamma^*$. We return
the arc of $\beta^*(g,B)$ incident to one of $f^*$, $f^*_1$ and $f_2^*$ with larger $\Delta^r(\cdot)$. The reason one needs to consider $f^*_1$ and $f^*_2$ is that the hole $h_r$ may be located in the region of the plane bounded between the branch of $T_b$ from $b$ to a vertex of $f_1$ and the boundary of the Voronoi cell of $b$ in $\VD^*(B)$. See Figure~\ref{fig:trimax-multi} (middle). 

\begin{figure}[h]
\begin{center}
\includegraphics[width=0.31\textwidth]{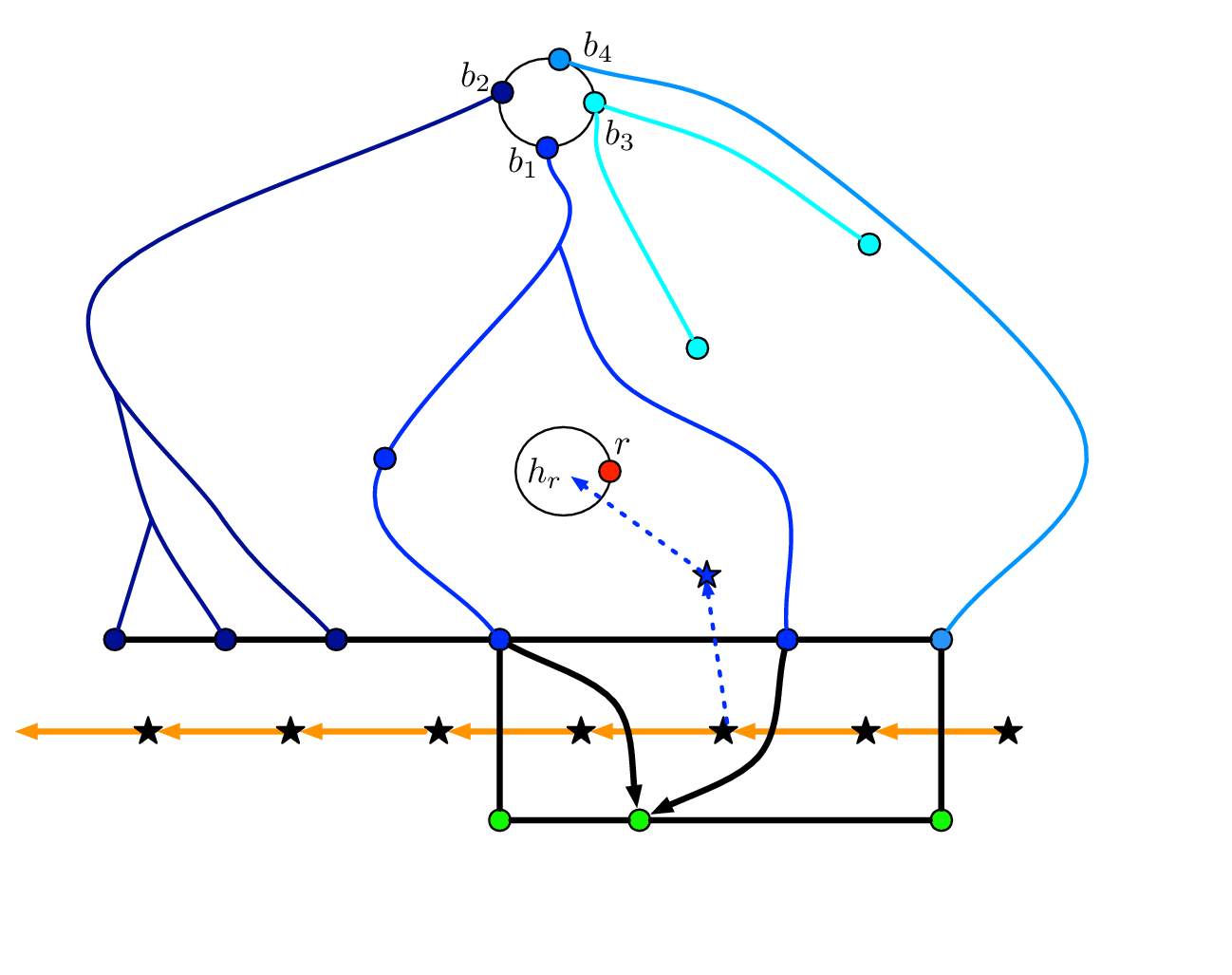}
\includegraphics[width=0.31\textwidth]{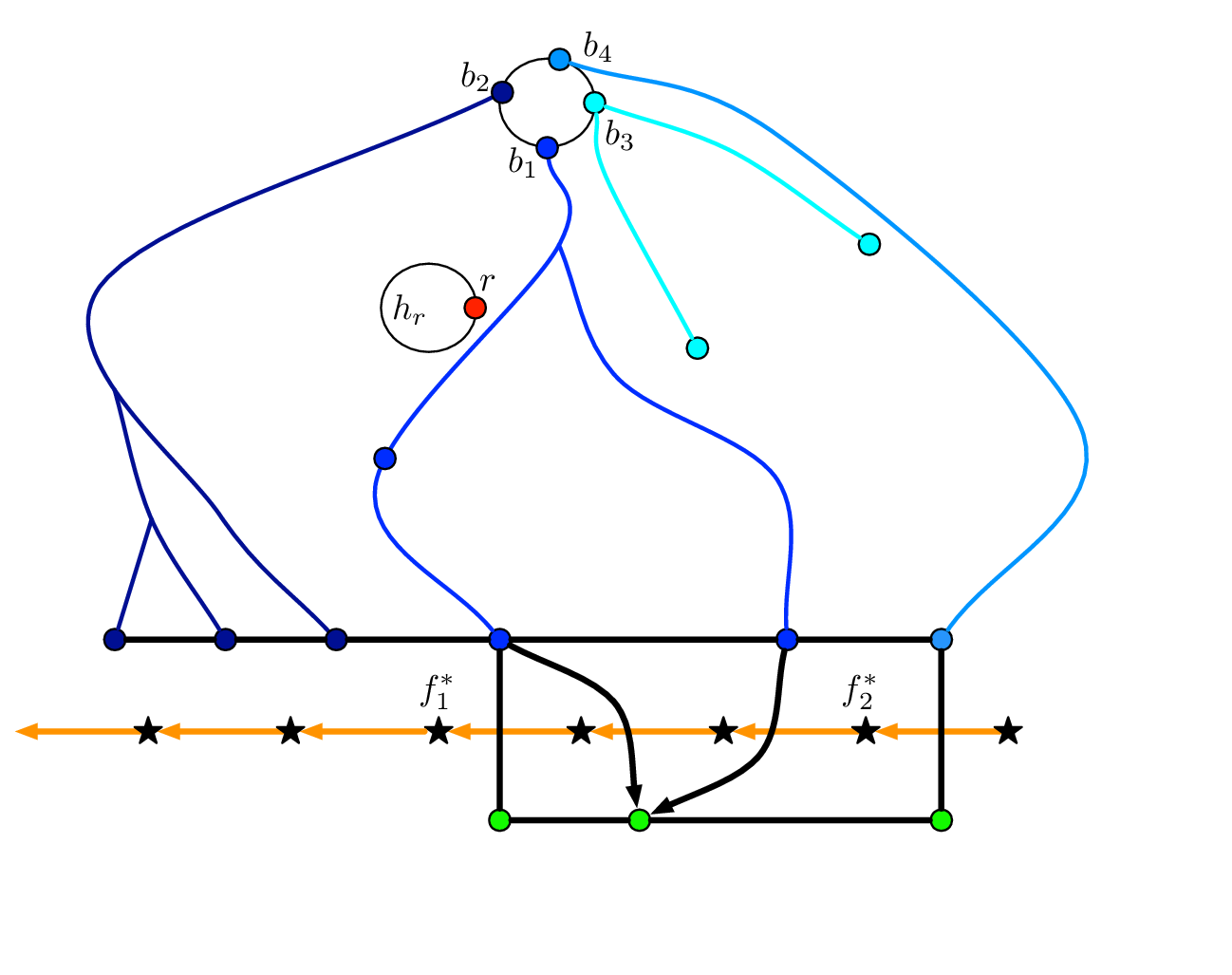}
\includegraphics[width=0.31\textwidth]{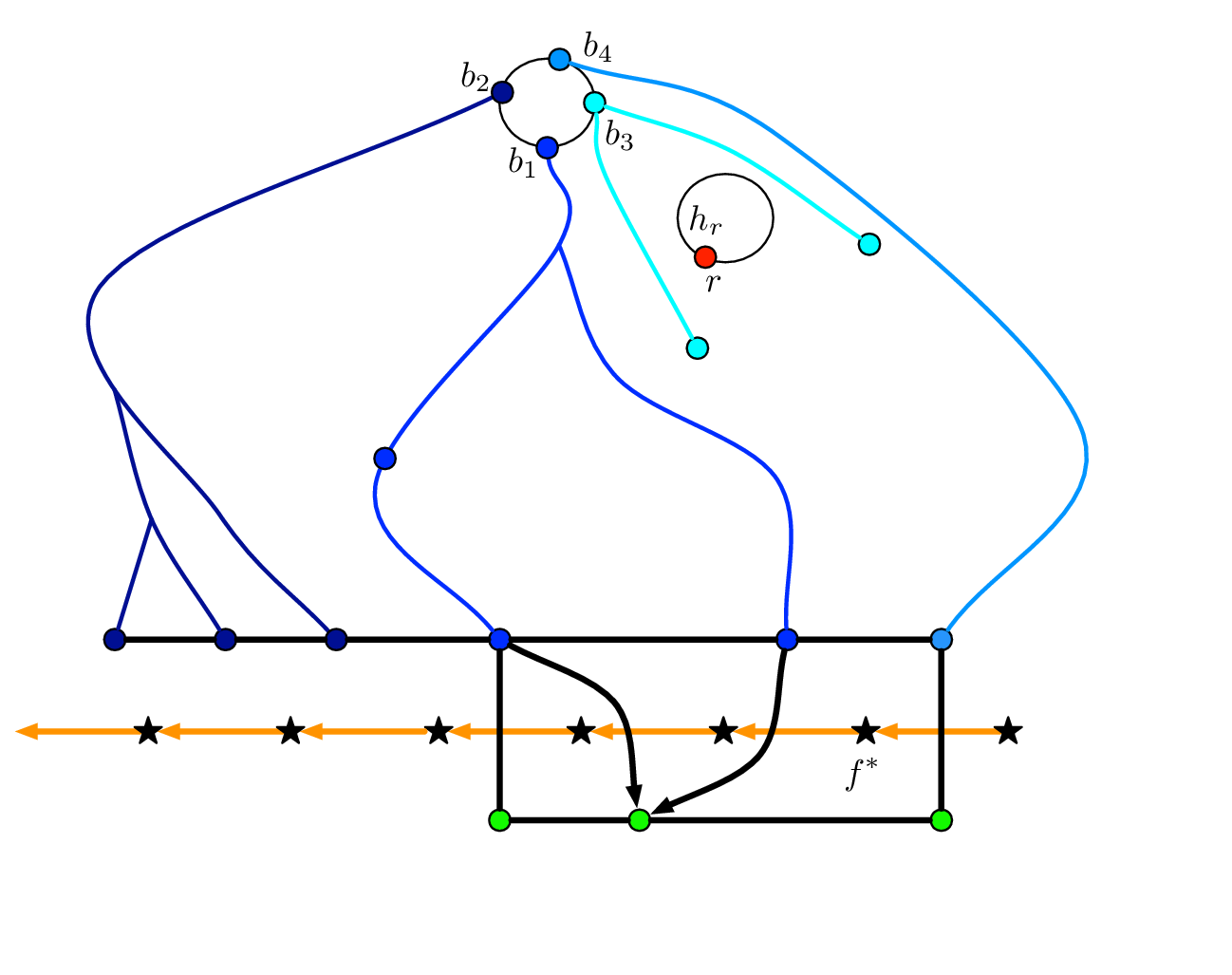}
\caption{Illustration of finding an edge maximizing $\Delta^r$ in the case of a set $B$ of multiple sites. On the left and middle are two scenarios that can occur when $b \in B'$. On the left the hole $h_r$ is sandwiched between two branches of $T_{b}$, both belonging to $\Vor(B)$. The method used in the case of a single site works here. On the middle, the hole $h_r$ is sandwiched between branches of different trees in $\Vor(B)$, so $f^*_1$ must also be considered. On the right is an illustration of the case where the hole $h_r$ belongs to a Voronoi cell of a site $b \notin B'$.
\label{fig:trimax-multi}}
\end{center}
\end{figure}

Finally, consider the case where $b \not\in B'$. In this case
we find in $\tilde O(1)$ time, using the search tree representation of $\beta^*(g,B)$, the  predecessor $b_{pre}$ and the  successor $b_{suc}$ of $b_i$ in $B'$. 
Let $f^*$ be the common dual vertex on the segment of $\beta^*(g,b_{pre})$
and $\beta^*(g,b_{suc})$ in $\beta^*(g,B)$.
We return
the arc of $\beta^*(g,B)$ incident to  $f^*$ with larger $\Delta^r(\cdot)$. See Figure~\ref{fig:trimax-multi} (right).
Recall that the procedure is given as input the site in $B \cup \{g\}$ minimizing the additive distace to $r$, so we find out in $\tilde O(1)$ time which of the above three cases applies.
The correctness of this procedure is argued in an similar manner to the proof of Lemma~\ref{lem:trimax}.
\begin{lemma}\label{lem:findmax-extended}
Consider the settings of Theorem~\ref{thm:tri-extended}.	We can preprocess $P$ in $O(n|S|^2)$ time  so that one can find an arc of $\beta^*(g,B)$ maximizing $\Delta^r(\cdot)$ in $O(\log n)$ time.  
\end{lemma}

\else
We omit the proof, which will appear in the full version.
\fi

\subsection{The mechanism}\label{sec:mechanism}
We now return to presenting our mechanism for finding $I$ described in Section~\ref{sec:weak-bitonic-search}.
 To this end, we need to exploit some structure of the
shortest path trees rooted at the three sites, and the
evolution of the Voronoi diagram $\VD_x := \VD(\{r,g,b\})$, with the weights $\wt(g),\wt(b)$
kept fixed and $\wt(r)=x$, as $x$ decreases from $+\infty$ to $-\infty$. For $c \in \{r,g,b\}$, let
$\Vor_x(c)$ denote the current version of $\Vor(c)$ (for the weight $\wt(r)=x$); recall that it is a
subtree of $T_c$ that spans the vertices in $\Vor(c)$ in $\VD_x$. These subtrees form a spanning
forest $F^x$ of $P$.
We call any edge not in $F^x$ with one endpoint in $\Vor_x(c_1)$ and the other in $\Vor_x(c_2)$,
for a pair of distinct sites $c_1 \neq c_2$, \emph{$c_1c_2$-bichromatic}.
To handle the case where
$\Vor_x(r)$ is empty (e.g., at $x=+\infty$), we think of adding a super source $s$ connected to each
site $c$ with an edge $sc$ of weight $\wt(c)$ (in general, these edges cannot be embedded in the plane
together with $P$), and define the edge $sr$ to be bichromatic.

We now define some bichromatic arcs to be tense at certain critical values of $x$ in a way similar to the definition
of tense arcs in Section \ref{sec:bisectors}. The difference is that here
an 
arc $uv$ is tense at $x$ if $v$ becomes closer to $r$ at $x$ than to {\em both} $g$ and $b$ (rather than to just one
other site as in Section \ref{sec:bisectors}).
Specifically,
we say that an 
arc $uv$ is \emph{tense} at $x$ if $uv$ is $rc$-bichromatic just before $x$ (i.e., for values slightly larger than $x$), and $uv$
is an arc of the full tree $T_r$,
and $x + d_{T_r}(r,v) = \wt(c) + d_{T_c}(c,v) = d(s,v)$.
The values of $x$ at which some arc becomes
tense are a subset of the critical values at which the $rb$-bisector or the $rg$-bisector change.
Recall that we assume that at each critical value $x$ only one arc is tense.
At the critical value $x$, the red tree $\Vor_x(r)$ takes over the node $v$ of the tense arc $uv$.
Let $\Vor_x^+(c)$ be the primal Voronoi cell of $c$ just before the critical value $x$.
For any descendant $w$ of $v$ in $\Vor_x^+(c)$, we have
\begin{align*}
d(s,w) & = \wt(c) + d_{T_c}(c,v) + d_{T_c}(v,w) \\
& = x + d_{T_r}(r,v) + d_{T_c}(v,w) ,
\end{align*}
so the red tree also takes over the entire subtree of $v$ in $\Vor_x^+(c)$. This establishes
the following lemma.
\begin{lemma}\label{lem:single-tight}
At any critical value $x$ of $\wt(r)$ the
 vertices that change color at $x$ are precisely the descendants of $v$ in $\Vor_x^+(c)$ where $uv$ is the unique 
 tense arc at $x$.
\end{lemma}

Consider the sequence $\Dbisr$, and let $x' > x$ be two consecutive values in it.
Recall that both $Q^*_{\geq x'}$ and $Q^*_{\geq x}$ are contiguous (cyclic) subsequences
of $\beta^*(g,b)$ (Lemma~\ref{lem:Delta-consecutive}). Assume $Q^*_{\geq x'} = \beta^*(g,b)[j \dots k]$,
and $Q^*_{\geq x} = \beta^*(g,b)[i \dots l]$ (where the indices are taken modulo the length of $\beta^*(g,b)$).
Since, by definition, $Q^*_{\geq x'} \subset Q^*_{\geq x}$, the set
$Q^*_{=x} := Q^*_{\geq x} \setminus Q^*_{\geq x'}$ corresponding to elements of $\Dbisr$
with value exactly $x$, forms two intervals $[i \dots j), (k \dots l]$ of $\beta^*(g,b)$, at least one of which is non-empty.
Let $uv$ be the unique 
tense arc at critical value $x$. Let $c\in\{g,b\}$ be such that
$uv$ is $rc$-bichromatic just before $x$. By the preceding arguments, including Lemma~\ref{lem:single-tight},
the nodes $w$ with $\Delta^{r}(w)=x$ are descendants of $v$ in
the subtree of $T_c$ spanning
$\Vor_{x}^+(c)$.
Furthermore, the
 subtree of $T_c$ spanning $\Vor_{\infty}(c)$ (this is the cell of $c$ in $\VD(\{g,b\})$ when $r$ is at distance $\infty$ from all vertices) contains the subtree of
$T_c$ spanning
$\Vor_x^+(c)$ and additional nodes $w$ for which $\Delta^{r}(w)=\tilde{\delta}^{rc}(w) \geq x$ that
 switched to the red tree at critical values greater than $x$ (recall that $x$
is decreasing).

Recall the definition of the path $R^*_v$ in the statement of Lemma~\ref{lem:consecutive} (with respect to the bisector $\beta^*(g,b)$).
It follows, by the discussion above, that any edge $e^*$ of $R^*_v$ has $\Delta^{r}(e^*) \geq x$,
and all edges $e^*$ with $\Delta^{r}(e^*)=x$ belong to $R^*_v$. Let $e^*_1$, and $e^*_2$ be the first
and last edges of $R^*_v$, respectively. By Lemma~\ref{lem:consecutive}, $e^*_1$ and $e^*_2$ are
also the edges with minimum and maximum values of $\pre_c(\cdot)$ in $R^*_v$, respectively.
It follows that at least one of $\Delta^{r}(e^*_1),\Delta^{r}(e^*_2)$ is $x$ (and the other
is $\ge x$), and that, moreover, either $e^*_1$ or $e^*_2$ is an extreme edge in the maximal
interval of edges $Q^*_{\ge x}$ of $\beta^*(g,b)$.

Exploiting these properties, we design the following procedure {\sc GetInterval($\hat e^*$)}.
The input is an edge $\hat{e}^* \in \beta^*(g,b)$ with $\Delta^{r}(\hat{e}^*) = x$.
The output is an interval $I$ of  extreme edges $ e^*$ of $Q^*_{\geq x}$
such that (1)  either the first or the last edge of $I$ is of value $x$, (2)
$I$ contains all edges of value $x$ in $\beta^*(g,b)$ (and $\hat{e}^*$ in particular).

\bigskip
\bigskip
\noindent
{\sc GetInterval($\hat e^*$)}
\begin{enumerate}
\item Let $\hat e=vw$ be the primal edge of $\hat e^*$. Find the endpoint of $\hat e$ whose
      $\Delta^{r}(\cdot)$ value is $x$ (Lemma~\ref{lem:single-tight} implies that
      there is only one such endpoint). Suppose, without loss of generality, that $\Delta^{r}(v)=x$.
\item Let $c \in \{g,b\}$ be the site such that 
$\Delta^{r}(v)=\tilde{\delta}^{rc}(v)=x$.
\item Find the ancestor $u$ of $v$ in $T_c$ that is nearest to the root, such that
      $\tilde{\delta}^{rc}(u)=x$.
\      Note that the node $u$ is an endpoint of the unique tense arc at
      critical value $x$. Finding $u$ can be done by binary search on the root-to-$v$ path in $T_c$ since,
      Lemma~\ref{lem:single-tight} implies that all the ancestors of $u$ on this path have strictly smaller $\tilde \delta^{rc}$-values.
\item \label{GI:Rv} 
Let $p$ be the parent of $u$ in $T_c$.
Return the interval of $\beta^*(g,b)$ consisting of arcs whose $\pre_c(\cdot)$ numbers are in the interval $(\pre_c(pu), \pre_c(up))$ (that is, return $R^*_v$). Since, by Lemma~\ref{lem:consecutive}, the cyclic order on $\beta^*(g,b)$ is
      consistent with $\pre_c(\cdot)$, we can find this interval by a binary search if we use $\pre_c(\cdot)$ as an additional key in the search tree representing $\beta^*(g,b)$.
\end{enumerate}

To efficiently implement {\sc GetInterval($\hat e^*$)},
we retrieve and compare the distances
from $r$, $g$, and $b$ to $v$ and $w$. This gives $\Delta^{r}(v)$, $\Delta^{r}(w)$, and thereby
$\Delta^{r}(\hat e)$. It also reveals the
Voronoi cell in $\VD(\{g,b\},\{ \wt(g),\wt(b) \})$ containing $v$.
All this is done easily in $O(1)$ time.
To carry out the binary search to find
 the ancestor $u$
of $v$ in $T_c$ that is nearest to the root, such that
      $\tilde{\delta}^{rc}(u)=x$, we use a level ancestor data structure on $T_c$ which we prepare during  preprocessing.
Each query to this data structure takes  $O(1)$ time.

Using the weak bitonicity of $\Dbisr$ (Corollary~\ref{cor:bitonic}), the procedure {\sc GetInterval}, and the procedure for finding
$e^*_{\max}$ described in Section \ref{sec:trimax},
we can find the trichromatic faces that are dual to the Voronoi vertices
of $\VD^*(\{r,g,b\})$, under the respective weights $\wt(r),\wt(g),\wt(b)$ by the variation of binary search described before.

Several additional enhancements of the preprocessing stage are needed to support an efficient implementation
of this procedure. First, we need to store $T_c$ for each individual site $c$.
Second, for each bisector $\beta^*(c_1,c_2)$ we need to store in its persistent search tree
representation two secondary keys $\pre_{c_1}(\cdot)$ and $\pre_{c_2}(\cdot)$. As we discussed
above, these keys are consistent with the cyclic order of $\beta^*(c_1,c_2)$. Clearly, all these
enhancements do not increase the preprocessing time asymptotically.
We have thus established Theorem~\ref{thm:tri}.

\ifdefined\fullver

\subsection{Dealing with a group $B$ of sites}

In this section we consider the more general scenario of Theorem~\ref{thm:tri-extended}, where there is a single red site $r$, a single green site $g$,
and a set $B$ of blue sites that are on some  hole $\H$. Recall the definition of $\VD^*(r,g,B)$ provided just before the statement of Theorem~\ref{thm:tri-extended}. 
The goal is to find the trichromatic vertices of $\VD^*(r,g,B)$, under the
currently assigned weights; by a \emph{trichromatic vertex} of $\VD^*(r,g,B)$ we mean a primal face with at least one vertex
in $\Vor(r)$ (a red vertex), at least one vertex in $\Vor(g)$ (a green vertex), and at least one vertex
in $\Vor(b)$, for some $b\in B$ (a blue vertex). Handling this situation is similar to the simpler
case of three individual sites, but requires several modifications.
Note that, by creating a ``super-blue'' site inside $\H$, and by connecting it to all the blue sites,
Lemma~\ref{lem:2vvert} still holds, so, as in the simpler case, there are at most two
trichromatic vertices.

Next, Lemma~\ref{lem:Delta-consecutive} still holds, but its proof needs a slight modification
because Lemma~\ref{lem:bicycle} might not apply. Instead, we need to use Lemma~\ref{lem:GBcycle}, which only guarantees that $\beta^*(g,B)$ is a non-self-crossing cycle, which might pass multiple
times through $\H^*$.

\begin{proof}({\em of Lemma~\ref{lem:Delta-consecutive} for a bisector $\beta^*(g,B)$}) 
The original proof of Lemma~\ref{lem:Delta-consecutive} considers a maximal contiguous interval of
edges of $\beta^*(g,b)$ whose corresponding primal edges have at least one red endpoint. It argues that the
extreme vertices of such a path must be trichromatic. Since in a simple cycle extreme vertices of maximal subpaths are distinct, and there are at most two trichromatic vertices (by Lemma~\ref{lem:2vvert}), it follows there can only be one such maximal subpath. Here, if $\H^*$ appears more than once in $\beta^*(g,B)$, and if $\H^*$ is trichromatic, we could have two edge-disjoint
maximal subpaths $I_1$, $I_2$ of $\beta^*(g,B)$ that share $\H^*$ as an endpoint. 

We prove by contradiction that this does not happen. Assume that $\H^*$ is
trichromatic. Then, by Lemma~\ref{lem:2vvert}, there is at most one other trichromatic vertex $x^*\neq \H^*$ along $\beta^*(g,B)$.
It is impossible that both $I_1$ and $I_2$ share $x^*$ too; by Lemma~\ref{lem:GBcycle}, $h^*$ is the only vertex that can have degree greater than 2 in $\beta^*(g,B)$, and so we could have merged
$I_1$ and $I_2$ into a larger subpath, contradicting their maximality. Hence, among
the four endpoints of $I_1$ and $I_2$, at least three are $\H^*$. This, combined with the fact
that $\beta^*(g,B)$ is non-self-crossing, imply the following property: We can choose an endpoint
of $I_1$ equal to $\H^*$ and an endpoint of $I_2$ equal to $\H^*$, choose edges $e_1^*$,
$e_2^*$ incident to the first endpoint, with $e_1^* \notin I_1$ and $e_2^* \in I_1$,
and choose edges $e_3^*$, $e_4^*$ incident to the second endpoint, with $e_3^*\notin I_2$ 
and $e_4^* \in I_2$, so that the cyclic order of these edges around $\H^*$ is
$(e_1^*,e_2^*,e_3^*,e_4^*)$. Note that $e^*_3 \notin I_1$, by the maximality of $I_1$ and $I_2$. As in the original proof of Lemma~\ref{lem:Delta-consecutive}, for each of
$e_2^*$, $e_4^*$, its primal edge has at least one red endpoint, and for each of
$e_1^*$, $e_3^*$, its primal edge has one green endpoint and one blue endpoint.
The two red endpoints can be connected by a path in $\Vor(r)$, consisting exclusively
of red vertices, and the two green endpoints can be connected by a path in $\Vor(g)$,
consisting exclusively of green vertices. However, by the cyclic order of these edges
along $\bd\H$, the red and the green paths must cross one another, which is impossible
because $\Vor(r)$ and $\Vor(g)$ are disjoint. This contradiction shows that
Lemma~\ref{lem:Delta-consecutive} continues to hold in this case too. See Figure~\ref{fig:tri2}.
\end{proof}

\begin{figure}[htb]
\begin{center}
\includegraphics[scale=0.75, clip=true, trim = 0mm 130mm 70mm 10mm]{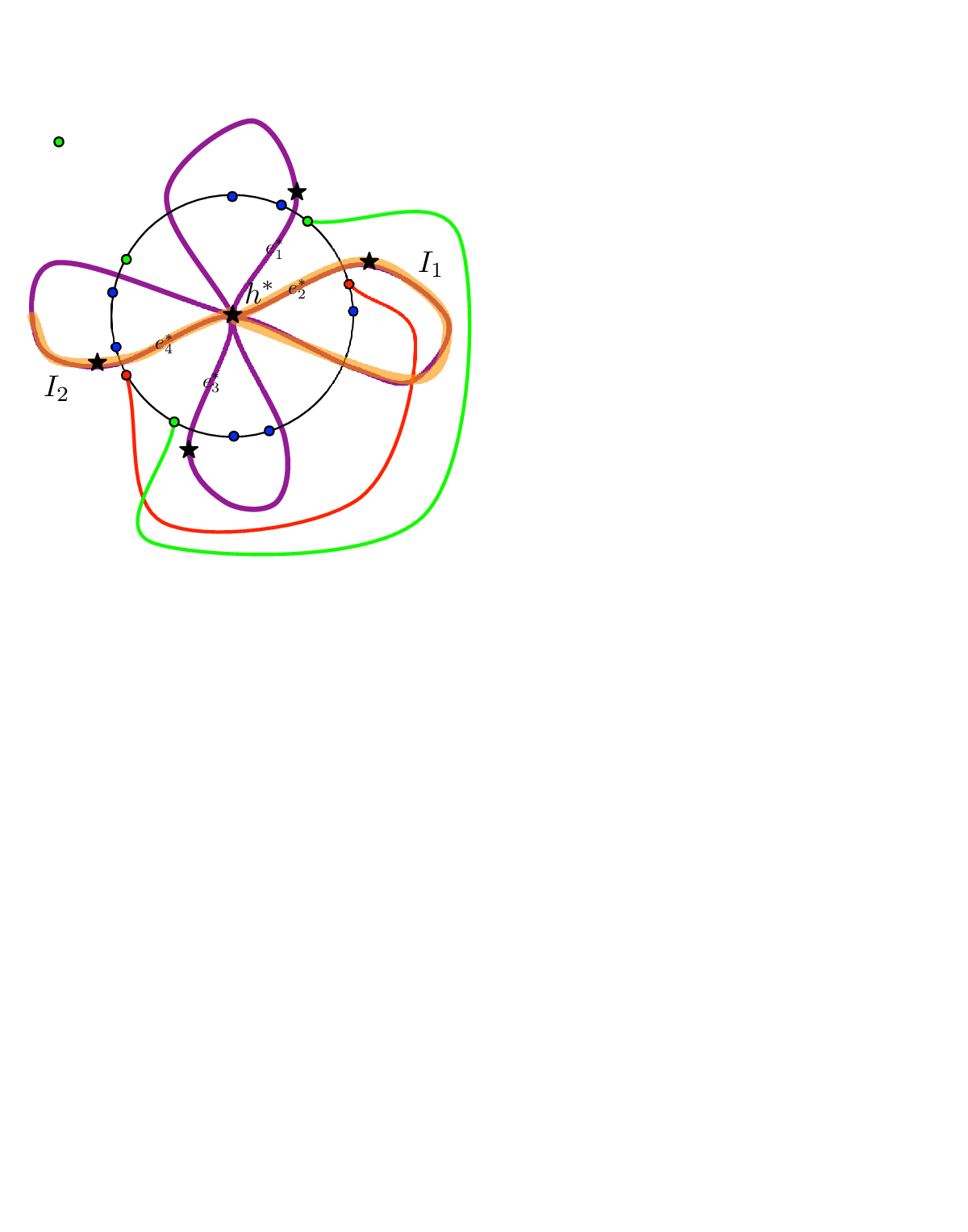}
\caption{Modification to the proof of Lemma~\ref{lem:Delta-consecutive}. The bisector $\beta^*(g,B)$ is shown in purple. Two non-consecutive subpaths, $I_1, I_2$, of $\beta^*(g,B)$ formed by edges of $Q^*_{\geq x}$ are highlighted in orange. These subpaths meet at the dual vertex $h^*$. The edges $e^*_1,e^*_2,e^*_3,e^*_4$ are indicated (their endpoints (dual vertices) are shown as stars), as well as a possible coloring of their primal endpoints. A contradiction arises since the red path and the green path intersect.
\label{fig:tri2}}
\end{center}
\end{figure}

We also need to slightly revise the statement of Lemma~\ref{lem:consecutive} as follows:
Let $g$ be a site and $B$ be a set of sites on a single hole $\H$. Let $b\in  B$, and let $X$ be the segment of $\beta^*(g,b)$ along $\beta^*(g,B)$. Note that, by Lemma~\ref{lem:gBcycle}, $X$ is a well defined single segment of $\beta^*(g,B)$. 
Let $v \neq g$ be a node in the cell $\Vor(g)$ in $\VD(g \cup B)$. Let $q$ be the parent of $v$ in $T_g$. Then the following hold:
\begin{enumerate}
\item
The arcs in $R^*_v := \{ e^* \in \beta^*(g,B) \mid \pre_g(pv) < \pre_g(e^*) < \pre_g(vp) \}$ form a subpath of $\beta^*(g,B)$.
\item
The labels $\pre_g(e^*)$ are strictly monotone along $R^*_v$.
\end{enumerate}

 Let $v$ be a node in the cell $\Vor(b)$ in $\VD(g \cup B)$. Let $p$ be the parent of $v$ in $T_b$. Then the following hold:
\begin{enumerate}
\item
The arcs in $
R^*_v := \{ e^* \in X \mid
pre_b(pv) < pre_b(e^*) < pre_b(vp) \}
$
is a contiguous subpath of $X$. %the segment of $\beta^*(g,b)$ along $\beta^*(g,B)$.

\item
The labels $\pre_b(e^*)$ are strictly monotone along $R^*_u$.
\end{enumerate}

The proof for $v \in \Vor(g)$ is the same as the original proof of Lemma~\ref{lem:consecutive}, by considering a super blue site $v_B$ instead of the individual sites in $B$. The proof for $v \in \Vor(b)$ is also the same as the original proof of Lemma~\ref{lem:consecutive}, applied to just  the cell $\Vor(b)$ of the diagram $\VD(g\cup B)$.

Note that, 
by Lemma~\ref{lem:gBcycle}, $\beta^*(g,B)$ consists of at most $|B|$ segments, each belonging to $\beta^*(g,b_j)$ for a different $b_j$. We will perform the bitonic search in two phases, first to locate the (at most two) segments containing a trichromatic vertex, and then locating the trichromatic vertex within each segment. 

In the first phase of the bitonic search we choose at each step an arc $\hat e^*$  of $\beta^*(g,B)$ from a segment of $\beta^*(g,b_j)$ that roughly partitions the number of segments of the active portion of the search equally. Thus, when we handle the arc $\hat e^*$, we know the site $b\in B$ such that $\hat e^*$ belongs to the segment $X$ of $\beta^*(g,b)$ on $\beta^*(g,B)$. 
Suppose that $\Delta^r(\hat{e}^*)=x$, and let  $uv$ be the unique tense arc at $x$.
If $v \in \Vor(g)$, then the procedure {\sc GetInterval} does not depend on the blue sites at all, and proceeds as in the single site case.
If $v \in \Vor(b)$, then, by Lemma~\ref{lem:gBcycle}, the interval that we seek consists only of edges on the boundary of $\Vor(b)$ in $\VD(g,B)$, so it is contained in the segment $X$. We can therefore find this interval by invoking the procedure {\sc GetInterval} with the sites $r,g,b$ as in the 
case of a single blue site $b$. The interval 
$R^*_v$ returned is an interval of extreme edges $ e^*$ of $\{ e^* \in \beta^*(g,b) \mid  \Delta^{r}(e^*) \geq x \}$
such that (1)  either the first or the last edge of $R^*_v$ is of value $x$, and (2)
$R^*_v$ contains all edges of value $x$ in $\beta^*(g,b)$ (and $\hat{e}^*$ in particular).
We are interested in a maximal such interval in $\{ e^* \in \beta^*(g,B) \mid  \Delta^{r}(e^*) \geq x \}$, not in $\{ e^* \in \beta^*(g,b) \mid  \Delta^{r}(e^*) \geq x \}$  (i.e., with respect to $\beta^*(g,B)$, not $\beta^*(g,b)$). Hence, we return the intersection of $R^*_v$ and $X$. 
This can be computed in $\tilde O(1)$ time by checking whether each endpoint of $R^*_v$ belongs to $X$, and if not, truncating $R^*_v$ at the corresponding endpoint of $X$.  

In the second phase, we complete the bitonic search within a single segment of $\beta^*(g,B)$ that belongs to $\beta^*(g,b)$ for a specific site $b$. The search is conducted in the same manner as in the first phase since we know the relevant site $b \in B$.
\else
The proof of Theorem~\ref{thm:tri-extended} builds upon that of Theorem~\ref{thm:tri}. Some additional preprocessing is required to handle a set of nearly consecutive sites $B$ instead of a single site $b$, as well as some changes due to the weaker structure of $\beta^*(g,B)$. The details will appear in the full version.
\fi

 % diagram

\section{Computing the Voronoi diagram} \label{sec:voronoi}

In this section we describe an algorithm that, given access to the pre-computed representation of the bisectors (Theorem~\ref{thm:bisectors}), and the mechanism for computing trichromatic vertices provided by Theorem~\ref{thm:tri-extended},
computes $\VD^*(S)$ in $\tilde{O}(|S|)$ time.
Thus we establish all parts of Theorem~\ref{thm:vor}, except for the mechanism for maximum queries in a Voronoi cell (item $(ii)$ in Theorem~\ref{thm:vor}), which is shown in Section~\ref{sec:prep_max}.
The presentation proceeds through several stages that handle the cases where the sites lie on the boundary of a single hole, of two holes, of three holes, and finally the general case.
\ifdefined\fullver
\else
Due to space constraint we only include the description of the single hole case. The case of multiple holes will appear in the full version.
\fi

\medskip
\noindent
{\bf Representing the diagram.}
We represent $\VD^*(S)$ by a reduced graph in which we contract all the vertices of degree $2$.
That is, we replace each path $p^*=p^*_1,p^*_2,\ldots,p^*_l=q^*$ in $\VD^*(S)$, where $p^*$ and $q^*$
are Voronoi vertices, and $p^*_2,p^*_3,\ldots,p^*_{l-1}$ are vertices of degree $2$, by the single
edge $(p^*,q^*)$. We represent this reduced graph using the standard DCEL data structure for
representing planar maps (see~\cite{dBCOvK} for details). Note that, $\VD^*(S)$ and its reduced representation may be disconnected, and may contain parallel edges and self loops. The standard DCEL data structure can represent disconnected and non-simple planar maps.
By Lemma \ref{lem:bisectors}, the path $(p^*=p^*_1,p^*_2,\ldots,p^*_l=q^*)$ is a contiguous portion
of some bisector $\beta^*(u_1,u_2)$. We store with the contracted edge $(p^*,q^*)$ pointers that
(a) identify the sites $u_1$, $u_2$ whose dual cells are adjacent to this edge, and
(b) point to the first and the last edges of this portion, namely to the edges $(p^*,p^*_2)$
and $(p^*_{l-1},q^*)$.

\ifdefined\fullver
\subsection{Single hole} \label{sec:single}
\fi
\medskip
\noindent
{\bf The divide-and-conquer mechanism.}
Assume that all sites lie on a single face (hole) $\H_1$ of the graph $P$, and let $\H^*_1$ denote its dual vertex. Note that here we only assume that the sites of the diagram are on a single hole, not that there is just a single hole (a possibly non-triangulated face of $P$).
We describe a divide-and-conquer algorithm for constructing $\VD^*(S)$ in $\tilde{O}(|S|)$ time.
Note that in this case the diagram does not contain an isolated loop; that is, its edges form a connected
graph (see below for a more precise statement and justification).

We partition the set $S$ of sites into two contiguous subsets of (roughly) the same size $k\approx |S|/2$.
Each subset is consecutive in the cyclic order around $\bd \H_1$; to simplify the notation, we assume that they have exactly the same size.
The sites in one subset, $G$, denoted as $g_1,\ldots,g_k$ in their clockwise order around
$\H_1$, are referred to as the {\em green} sites, and those of the other subset $R$, denoted as
$r_1,\ldots,r_k$ in their clockwise order around $\H_1$, are the {\em red} sites. If $k > 3$
we recursively compute the Voronoi diagram of $G$, denoted as $\VD^*(G)$, and the Voronoi diagram of $R$,
denoted as $\VD^*(R)$. If $k = 3$ then we compute $\VD^*(G)$ (resp.,  $\VD^*(R)$) using the algorithm of Section \ref{sec:trichrom} and if $k=2$,  $\VD^*(G)$ (resp.,  $\VD^*(R)$) is the bisector of the two sites in $G$ (resp.,  $R$). We now describe how to merge these two diagrams and obtain
$\VD^*(S)$, in $\tilde{O}(|S|)$ time.

\medskip
\noindent
{\bf Constructing \boldmath$\beta^*(G,R)$.}
To merge $\VD^*(G)$ and $\VD^*(R)$, we have to identify all segments of the bisectors
$\beta^*(g,r)$, for $g\in G$ and $r\in R$, that belong to $\VD^*(S)$.
Recall that we denoted the subgraph of $\VD^*(S)$ induced by the edges of these segments by $\beta^*(G,R)$.
The structure of $\beta^*(G,R)$ is specified by Lemma~\ref{lem:GBcycle} (this is the case where $h_g=h_b$ in the statement of the lemma).

Our algorithm consists of two main stages. In the first stage we identify the
edges of $\beta^*(G,R)$ that are incident to $\H^*_1$; that is, edges of $\beta^*(G,R)$
that are dual to edges on $\bd \H_1$.  We denote this subset of $\beta^*(G,R)$ by $\beta_I^*(G,R)$.
In the second stage we compute the paths that comprise $\beta^*(G,R)$, each of which connects
a pair of edges of $\beta_I^*(G,R)$, as obtained in Stage 1. Each such path is a concatenation of
contiguous portions of ``bichromatic''
bisectors of the form $\beta^*(g_i,r_j)$.

\medskip
\noindent
{\bf Stage 1.}
\hypertarget{sec:single-stage1}
We partition $\bd \H_1$ into segments by cutting off its edges whose duals are in $\VD^*(G)\cup \VD^*(R)$.
The vertices in each such segment belong to a single Voronoi cell $\Vor(g_i)$ of $\VD(G)$, and to a
single Voronoi cell $\Vor(r_j)$ of $\VD(R)$, and we then denote this segment by $A_{ij}$.
Consider the bisector $\beta^*(g_i,r_j)$. It is a simple cycle incident to $\H^*_1$, so it has at most two edges that cross $\bd \H_1$. These edges are
stored with $\beta^*(g_i,r_j)$ when it is constructed
by the sweeping procedure at preprocessing. For each segment $A_{ij}$, we check whether one of the edges, $e^*$, of $\beta^*(g_i,r_j)$ that
is incident to $\H^*_1$ is dual to an edge in $A_{ij}$.
 This can be done by assigning to each edge a number according to the cyclic walk along $\H_1$, and checking if the number assigned to $e$ is between the numbers assigned to  the endpoints of $A_{ij}$.
If we find such an edge $e^*$, it belongs to $\beta_I^*(G,R)$, and we add it to this set.

Another way in which an edge $e^*$ can belong to $\beta_I^*(G,R)$ is when
it is a delimiter between two consecutive segments $A_{ij}$ and $A_{i'j'}$
(where we have $i=i'$ or $j=j'$, but not both).
Let $e=xx'$ be the primal edge of $\bd \H$
dual to $e^*$, with $x\in A_{ij}$ and $x'\in A_{i'j'}$. In this case, the site in $G\cup R$
nearest to $x$ must be either $g_i$ or $r_j$, and the site nearest to $x'$ must be either
$g_{i'}$ or $r_{j'}$. Then $e^*$ is an edge of $\beta^*_I(G,R)$ if and only if the site
nearest to $x$ and the site nearest to $x'$ are of different colors
(a test that we can perform in constant time).

Since the preceding arguments exhaust all possibilities, we obtain the following lemma
that asserts the correctness of this stage.
\begin{lemma}
The procedure described above identifies all edges of $\beta_I^*(G,R)$.
\end{lemma}

\medskip
\noindent
{\bf Stage 2.}
By Lemma~\ref{lem:GBcycle}, $\beta^*(G,R)$ is a union of edge-disjoint simple cycles, all passing through $\H^*_1$; we refer to these sub-cycles simply as cycles.
The set $\beta_I^*(G,R)$ produced in Stage 1 consists of the edges incident to $\H^*_1$ on each of these cycles.
Moreover, as already argued, each of these cycles contains exactly two such edges.
Saying it slightly differently (with a bit of notation abuse) ignoring $\H^*_1$, each of these cycles is a simple
path in $P^*$ that starts and ends at a pair of respective edges of $\beta_I^*(G,R)$,
and otherwise does not meet $\bd \H_1$.
Stage 2 constructs these paths one at a time, picking an edge $e_1^*$
of $\beta_I^*(G,R)$, and tracing the path $\gamma$ that starts at $e_1^*$ until it reaches
$\bd \H_1$ again, at another ``matched'' edge $e_2^* \in \beta_I^*(G,R)$; the stage repeats this
tracing step until all the edges of $\beta_I^*(G,R)$ are exhausted.

We assume that each of $\VD^*(G)$ and $\VD^*(R)$ has at least two non-empty cells. The procedure for the other cases is a degenerate version of the one we describe.
We compute the path $\gamma\in \beta^*(G,R)$ containing an initial edge $e_1^* \in \beta_I^*(G,R)$ as follows.
Let $\beta^*(g_i,r_j)$ be the bisector that contains $e_1^*$ in $\beta^*(G,R)$.
It follows that $e_1^*$ is either contained in, or lies on the boundary of
$\Vor^*(g_i)$ in $\VD^*(G)$, and the same holds for $\Vor^*(r_j)$ in $\VD^*(R)$.
We would like to use Theorem~\ref{thm:tri-extended} with $r_j,g_i,G\setminus\{g_i\}$ and with $g_i,r_j,R\setminus\{r_j\}$. 
We have already computed $\VD^*(G)$, so $\beta^*(g_j, G\setminus\{g_i\})$ is available (it is the boundary of the Voronoi cell of $g_j$ in $\VD^*(G)$). Therefore, requirement (iii) in the statement of the theorem is satisfied. To satisfy requirement (ii) we compute, for each site in $S$, the Voronoi cell in $\VD(G)$ containing it. The following short description of this computation does not assume the sites are on a single hole, in preparation for the multi-hole cases. Each hole $h$ is either entirely contained in a single Voronoi cell in $\VD(G)$ or is intersected by Voronoi edges of $\VD^*(G)$. In the latter case we can infer which portion of the sites on $h$ belongs to which cell by inspecting the order of  the edges of $\VD^*(G)$ incident to $h$ along the boundary of $h$. 
We had already mentioned in the description of \hyperlink{sec:single-stage1}{Stage 1} that the edges of a bisector that are incident to a hole are stored with the bisector when it is constructed by the sweeping procedure at preprocessing. In the former case all the sites on $h$ are in the same Voronoi cell of $\VD(G)$, so we can choose an arbitrary site on $h$ and find its Voronoi cell by explicitly inspecting its distance from each site in $G$. Thus, the total time required to find, for all sites $s \in S$ the Voronoi cell of $\VD(G)$ to which $s$ belongs is $\tilde O(|G|)$, which is dominated by the time invested to compute $\VD^*(G)$ in the first place.

We use Theorem~\ref{thm:tri-extended} with $r_j,g_i,G\setminus\{g_i\}$  to find the (one or two) trichromatic vertices of $\VD^*(r_j,g_i, G\setminus\{g_i\})$. Similarly, we use Theorem~\ref{thm:tri-extended} with $g_i,r_j,R\setminus\{r_j\}$  to find the (one or two) trichromatic vertices of $\VD^*(g_i,r_j, R\setminus\{r_j\})$. Note that $\H_1^*$ is a trichromatic vertex in both diagrams, so it is always one of the vertices returned by the algorithm of Theorem~\ref{thm:tri-extended}.

Denote by $I_g$ the segment of the bisector $\beta^*(g_i,r_j)$ 
starting with $e_1^*$ and ending at the first trichromatic vertex of $\VD^*(r_j,g_i, G\setminus\{g_i\})$ encountered along $\beta^*(g_i,r_j)$ ($I_g$ may be all of $\beta^*(g_i,r_j)$ in case $e_1^*$ is incident to the single trichromatic vertex of $\VD^*(r_j,g_i, G\setminus\{g_i\})$).
Let $I_r$ be the segment of the bisector $\beta^*(g_i,r_j)$ starting with $e_1^*$ and ending at the first trichromatic vertex of $\VD^*(g_i,r_j,R\setminus\{r_j\})$ encountered along $\beta^*(g_i,r_j)$ ($I_r$ may be all of $\beta^*(g_i,r_j)$ in case $e_1^*$ is incident to the singletrichromatic vertex of $\VD^*(g_i,r_j,R\setminus\{r_j\})$).
Let $I$ be the shorter of the segments $I_g$ and $I_r$.
The segment $I$ of $\beta^*(g_i,r_j)$ is a Voronoi edge (of the final diagram)
in $\beta^*(G,R)$. If $I$ is the entire $\beta^*(g_i,r_j)$ then $\gamma = I$ and we are done.
Otherwise, each endpoint of $I$ is a Voronoi vertex of $\VD^*(S)$
adjacent to the cells $\Vor(g_i)$, $\Vor(r_j)$,
and to a third cell which is either $\Vor(g_{i'})$ for some green site $g_{i'}\not= g_i$,
or $\Vor(r_{j'})$ for some red site $r_{j'}\not= r_j$. Let $f^*$ be the endpoint of $I$
different from $\H^*_1$. Assume for concreteness that
$f^*$ is adjacent in $\VD^*(S)$ to $\Vor(g_i)$, $\Vor(r_j)$, and to some other
$\Vor(r_{j'})$.
Let $e_2^*$ be the edge of $\beta^*(g_i,r_{j'})$ incident to $f^*$ that is on the boundary of $\Vor (r_{j'})$ in $\VD^*(S)$  (the third edge adjacent to $f^*$ in  $\VD^*(S)$ is an edge of $\beta^*(r_j,r_{j'})$ from $\VD^*(R)$, and is not part of $\beta^*(G,R)$).
We continue tracing $\gamma$ by repeating the above procedure with %$f^*$ taking the role of $\H^*_1$, and 
$e^*_2$ taking the role of $e^*_1$.
We continue tracing $\gamma$
in this manner, identifying the Voronoi edges on $\gamma$ one by one, until we get back to $\H^*_1$ through an edge $a^* \in \beta^*_I(G,R)$.
We then discard $e_1^*$ and $a^*$ from $\beta^*_I(G,R)$, and keep performing this tracing
procedure from another edge of $\beta^*_I(G,R)$, until all edges of $\beta^*_I(G,R)$ are exhausted,
obtaining in this way all the cycles of $\beta^*(G,R)$. See Figure~\ref{fig:single} for an illustration.

\begin{figure}
\begin{center}
\captionsetup{subrefformat=parens}
\begin{subfigure}{.3\textwidth}
  \centering
  \includegraphics[width=0.95\textwidth]{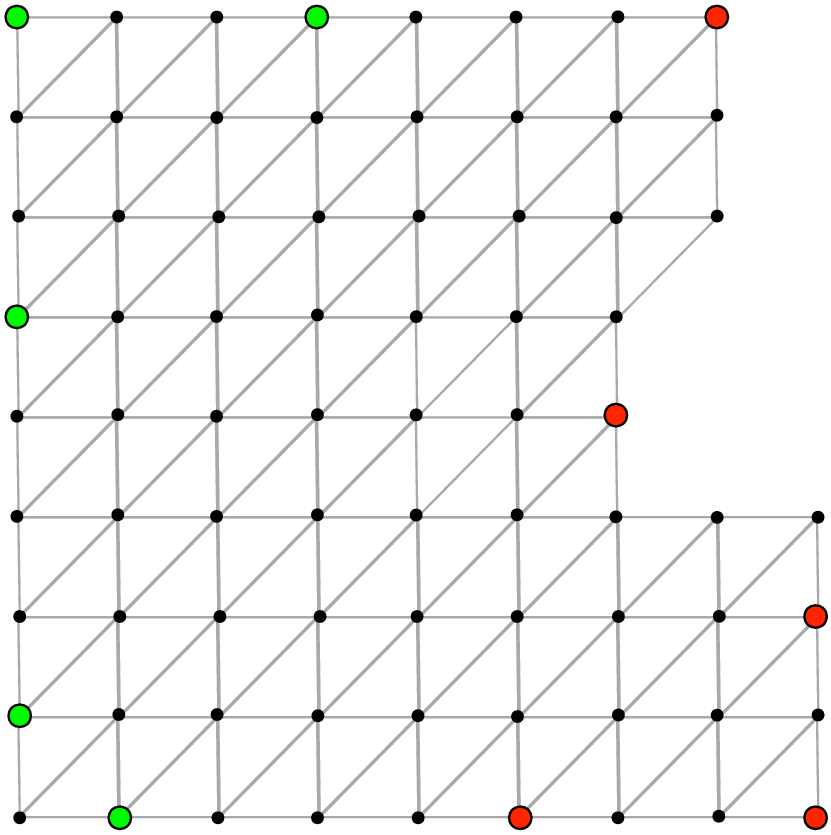}
  \caption{}
  \label{fig:sub1}
\end{subfigure}
\begin{subfigure}{.3\textwidth}
  \centering
  \includegraphics[width=0.95\textwidth]{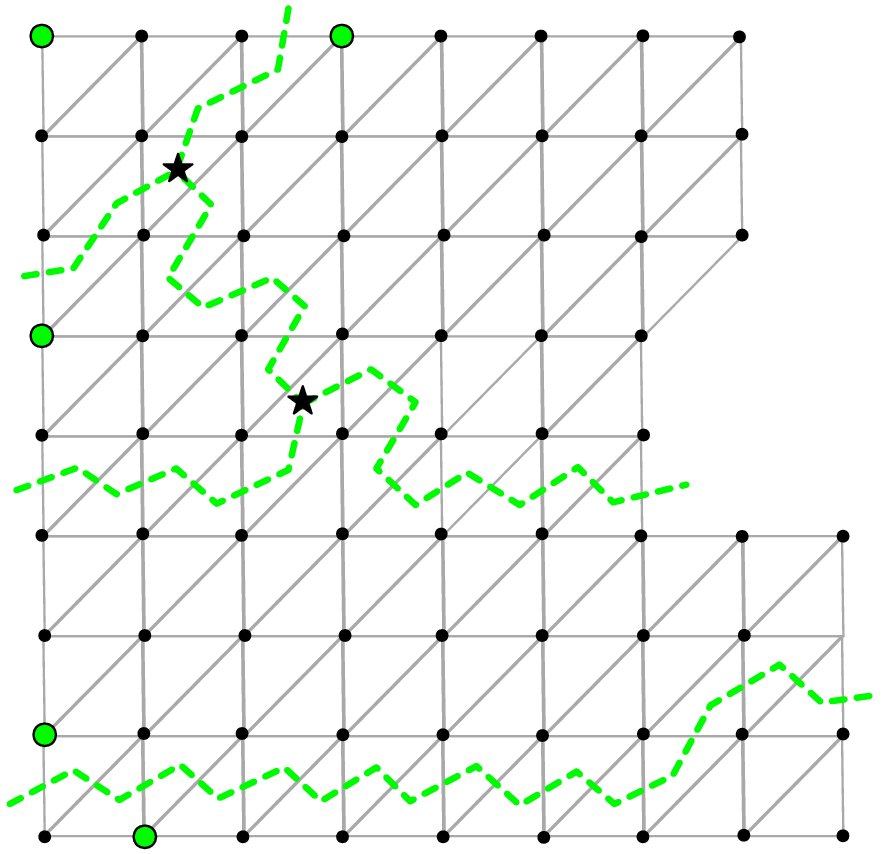}
  \caption{}
  \label{fig:sub1G}
\end{subfigure}
\begin{subfigure}{.3\textwidth}
  \centering
  \includegraphics[width=0.95\textwidth]{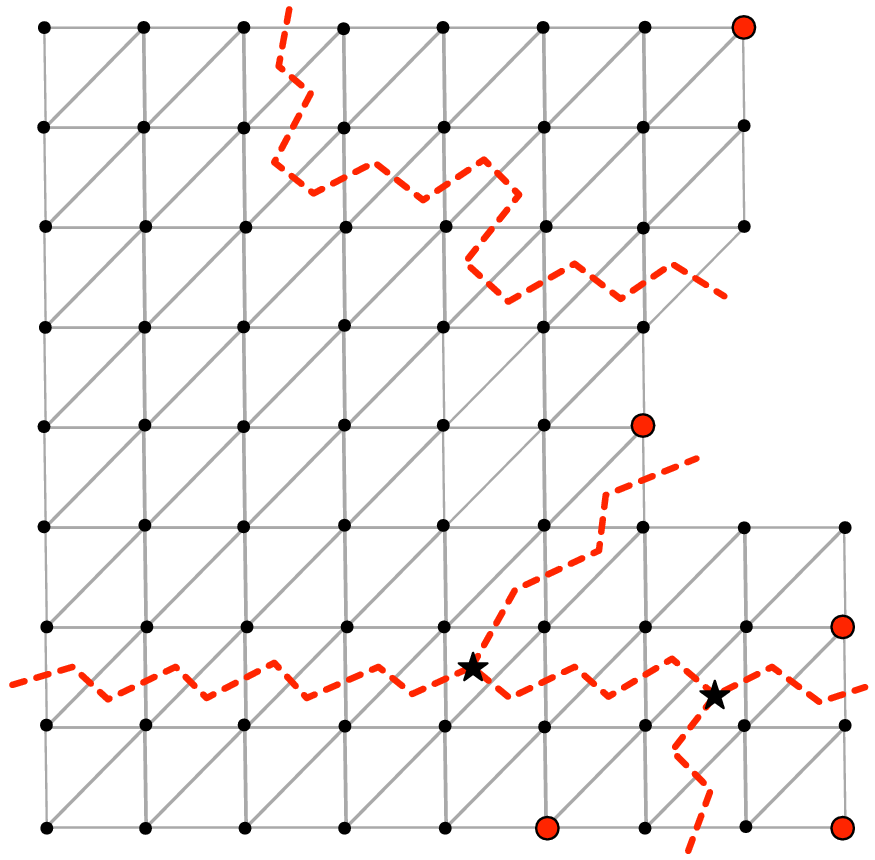}
  \caption{}
  \label{fig:sub1R}
\end{subfigure}
\begin{subfigure}{.3\textwidth}
  \centering
  \includegraphics[width=0.95\textwidth]{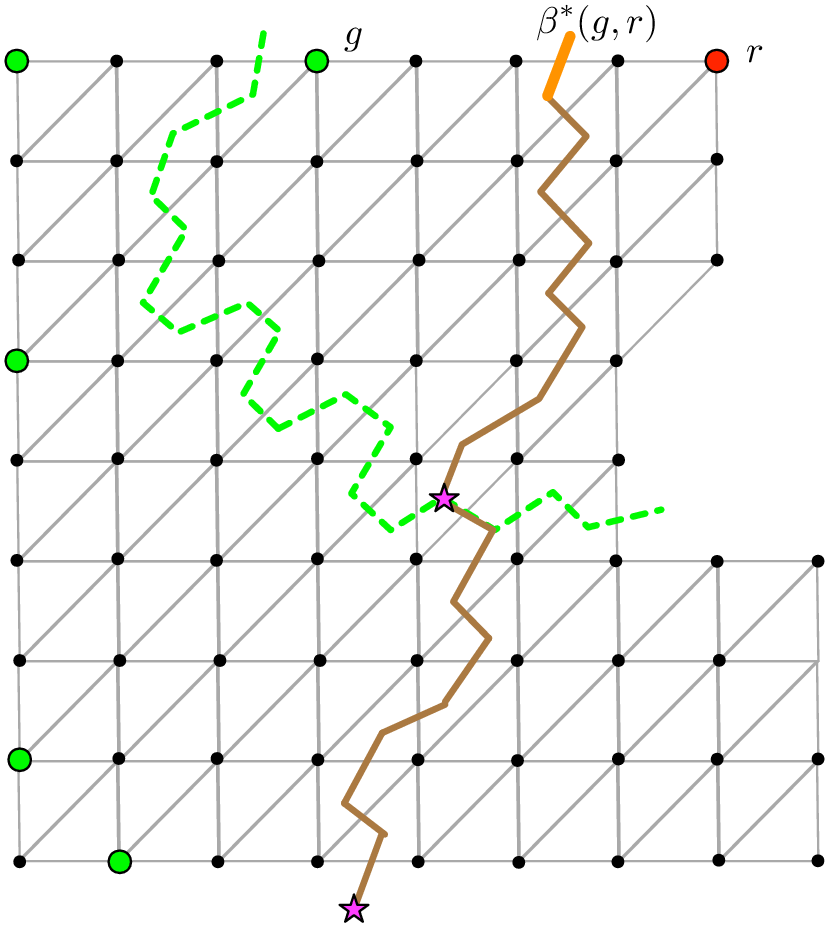}
  \caption{}
  \label{fig:sub2G}
\end{subfigure}
\hspace{0.1\textwidth}
\begin{subfigure}{.3\textwidth}
  \centering
  \includegraphics[width=0.95\textwidth]{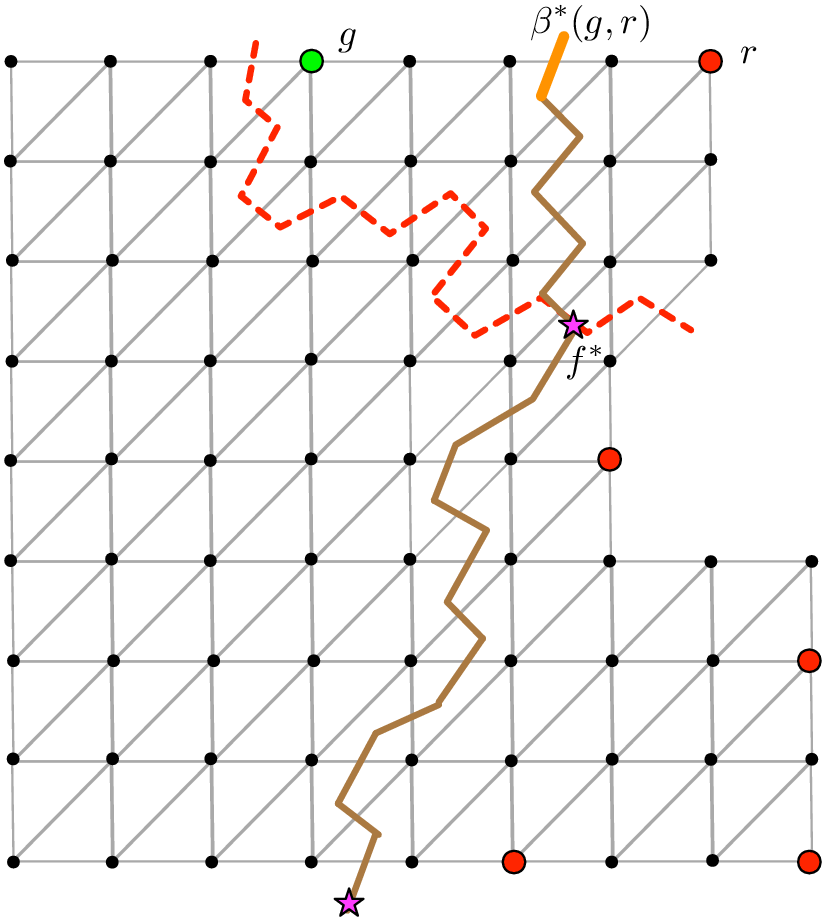}
  \caption{}
  \label{fig:sub2R}
\end{subfigure}
\begin{subfigure}{.3\textwidth}
  \centering
  \includegraphics[width=0.95\textwidth]{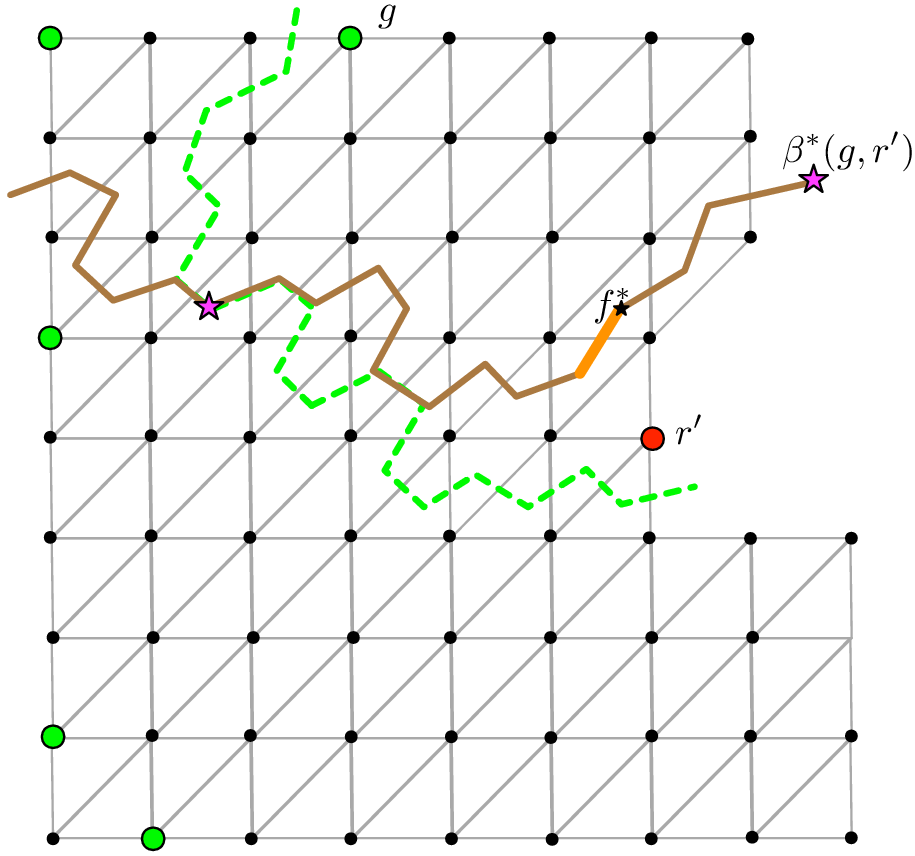}
  \caption{}
  \label{fig:sub3G}
\end{subfigure}
\begin{subfigure}{.3\textwidth}
  \centering
  \includegraphics[width=0.95\textwidth]{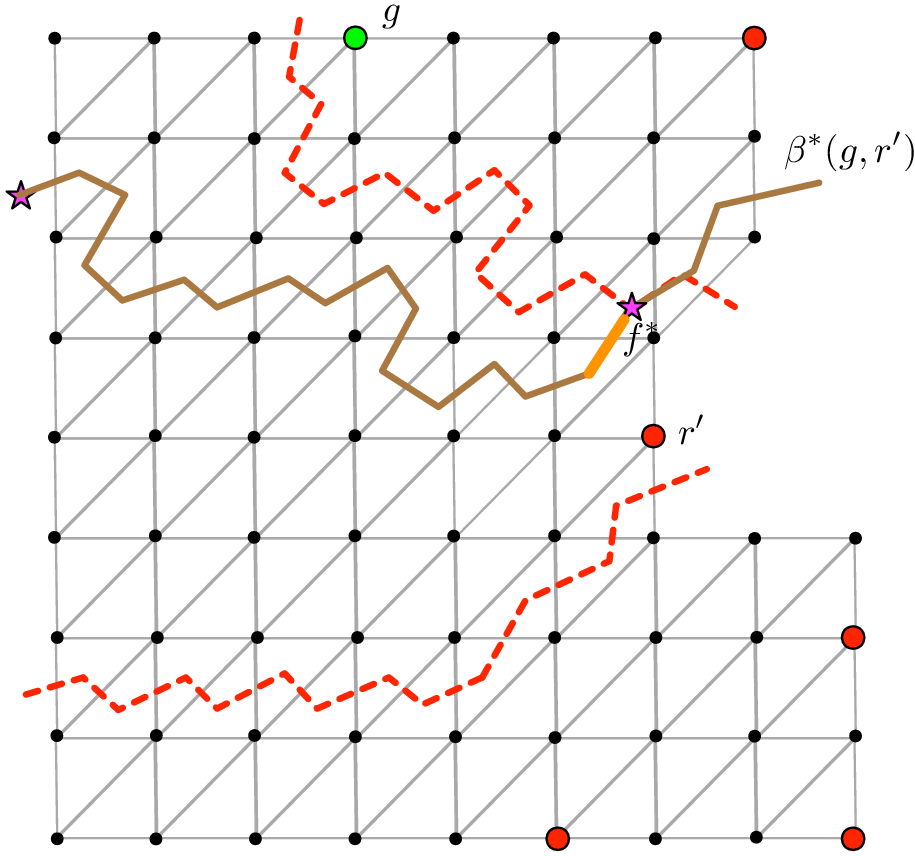}
  \caption{}
  \label{fig:sub3R}
\end{subfigure}
\begin{subfigure}{.3\textwidth}
  \centering
  \includegraphics[width=0.95\textwidth]{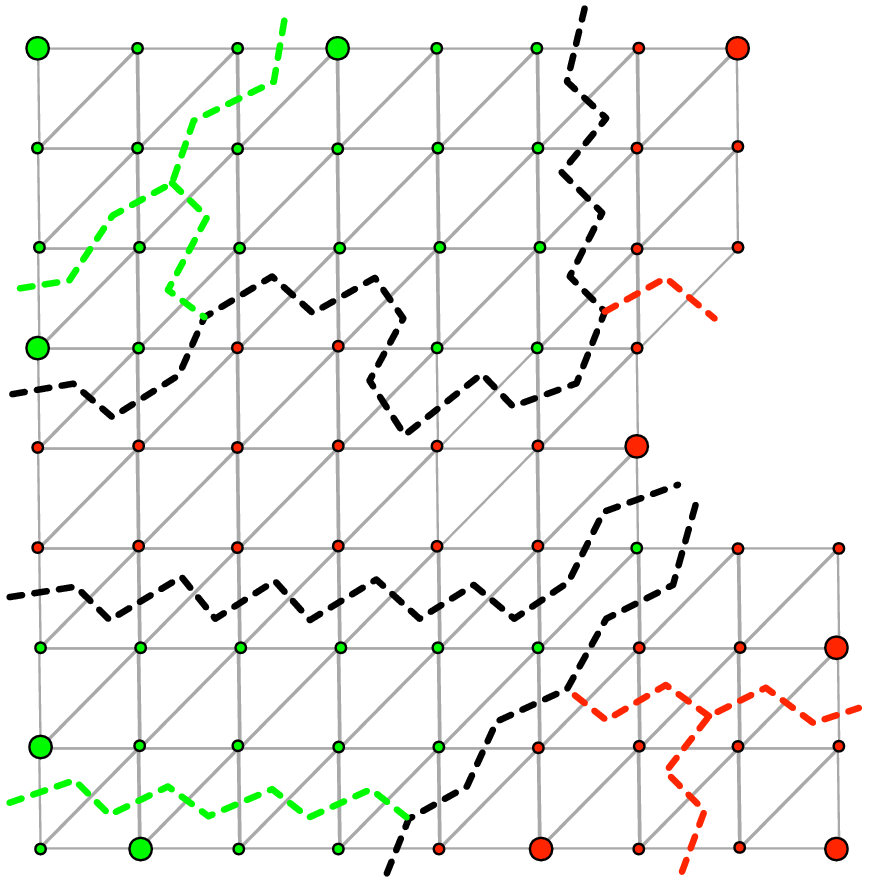}
  \caption{}
  \label{fig:sub4}
\end{subfigure}
\caption{Illustration of Stage 2 of the construction. ~\subref{fig:sub1} A graph with a set $S$ of sites (larger circles) on a single hole, partitioned into green  and red sites. 
\subref{fig:sub1G} $VD^*(G)$.
\subref{fig:sub1R} $VD^*(R)$.
\subref{fig:sub2G} An arc $e_1^*$ of $\beta^*_I(G,R)$ (orange) belongs to $\beta^*(g,r)$ (brown). We compute the trichromatic vertices of $\VD^*(r,g,G\setminus\{g\})$ (magenta stars).
\subref{fig:sub2R} Similarly, we compute the trichromatic vertices of $\VD^*(g,r,R\setminus\{r\})$. Since $I_r$ is shorter than $I_g$, $I_r$ is a Voronoi edge of $\VD^*(S)$ in $\beta^*(G,B)$. 
\subref{fig:sub3G} The arc $e_2^*$ (orange) belongs to $\beta^*(g,r')$. We compute the trichromatic vertices of $\VD^*(r',g,G\setminus\{g\})$.
\subref{fig:sub3R} Similarly, we compute the trichromatic vertices of $\VD^*(g,r',R\setminus\{r'\})$.
\subref{fig:sub4} The complete diagram $\VD^*(S)$. The bisector $\beta^*(G,B)$ is shown in black.
\label{fig:single}}
\end{center}
\end{figure}

\medskip
\noindent
{\bf Wrap-up.}
Once we have computed the simple cycles comprising the non-self-crossing cycle $\beta^*(G,R)$, we can assemble the entire
Voronoi diagram $\VD^*(S)$ by pasting parts of $\VD^*(G)$ and $\VD^*(R)$ to $\beta^*(G,R)$ as follows.
Consider the bisector $\beta^*(G,R)$ as a cycle in the plane. It intersects the embedding of $\VD^*(G)$ at the endpoints of segments of $\beta^*(G,R)$ (each such intersection is either a Voronoi vertex of $\VD^*(G)$ or a degree 2 vertex on some Voronoi edge of $\VD^*(G)$). We use the DCEL representation to cut $\VD^*(G)$ along $\beta^*(G,R)$, keeping just parts of $\VD^*(G)$ that belong to the $G$ side of $\beta^*(G,R)$. We repeat the process for $\VD^*(R)$, and then glue together $\beta^*(G,R)$ and the parts we kept from $\VD^*(G)$ and $\VD^*(R)$. This is done by identifying the endpoints of segments of $\beta^*(G,R)$ that appear in multiple parts.

\medskip
\noindent
{\bf Running time.}
Each call to Theorem~\ref{thm:tri-extended} takes $\tilde{O}(1)$ time, and the
number of calls is proportional to the overall complexity of $\VD(G)$, $\VD(R)$,
and $\VD(S)$, which is $O(|S|)$. It follows that the overall cost of the
divide-and-conquer mechanism is $\tilde{O}(|S|)$.

\ifdefined\fullver

\subsection{Two holes} \label{sec:double}

Assume now that the sites $S$ are located on only two holes (distinguished faces) $\H_1$ and $\H_2$ of $P$. We call the sites on $\bd \H_1$ (resp., on $\bd \H_2$) the
\emph{green sites} (resp., \emph{red sites}), and denote those subsets of $S$ by
$G$ and $R$, respectively. We extend the algorithm of Section \ref{sec:single} to
compute $\VD^*(S)$ in this ``bichromatic'' case.
Let $\H^*_1$ and $\H^*_2$ be the dual vertices of $\H_1$ and $\H_2$, respectively.

Our algorithm first computes the Voronoi diagram $\VD^*(G)$ of the green sites, and the
Voronoi diagram $\VD^*(R)$ of the red sites, using the algorithm of Section \ref{sec:single}.
(That algorithm works in the presence of other holes that have no sites incident to them.)
We then merge $\VD^*(G)$ and $\VD^*(R)$ into $\VD^*(S)$, by
finding the bisector $\beta^*(G,R)$. This bisector consists of the bichromatic edges of $\VD^*(S)$, each of which
corresponds to a segment of a bisector of the form $\beta^*(g,r)$, for some $g\in G$
and $r\in R$. The  structure of $\beta^*(G,R)$ is characterized in Lemma~\ref{lem:GBcycle}.

Our algorithm for computing $\beta^*(G,R)$ has two stages, similar to the two stages of the algorithm in the previous subsection. In the first stage we compute the edges of $\beta^*(G,R)$ that are dual to edges of $\H_1$ or of $\H_2$. We denote this subset of $\beta^*(G,R)$ by $\beta_I^*(G,R)$.
This is done in exactly the same way as the computation of $\beta_I^*(G,R)$ in the previous section.

Note however that here, in contrast with the previous algorithm, $\beta_I^*(G,R)$ may be empty, in case $\beta^*(G,R)$ is a simple cycle which does not contain $\H^*_1$ and $\H^*_2$.
When $\beta_I^*(G,R)=\emptyset$, we identify an edge $e^*$ on the ``free-floating'' cycle $\beta^*(G,R)$, using a special procedure that we describe below.

In the second stage we construct the whole $\beta^*(G,R)$.
If $\beta_I^*(G,R)\not= \emptyset$ (i.e., if $\beta^*(G,B)$ is incident to $h^*_1$ or to $h^*_2$),
then $\beta^*(G,R)$ decomposes into simple sub-cycles each
passing through $h_1^*$ and/or $h_2^*$. Each sub-cycle decomposes into one or two paths between edges of
$\beta_I^*(G,R)$. If the cycle contains only $h_1^*$ (resp., $h_2^*$)
 it becomes a single path when we delete $h_1^*$ (resp., $h_2^*$). This path connects a pair of edges of $\beta_I^*(G,R)$, where the two edges
in such a pair are both dual to edges of $\H_1$ (resp., $h_2$). If the cycle contains both holes then it splits into two paths when we delete $h_1^*$ and $h_2^*$.
Each of these paths connects a pair of edges of $\beta_I^*(G,R)$, where one
is dual to an edge of $\H_1$ and one dual to an edge of $\H_2$.
We compute these paths and cycles exactly as we computed, in Section \ref{sec:single},
the paths of $\beta^*(G,R)$ connecting pairs of edges of $\beta_I^*(G,R)$.

If $\beta^*(G,R)$ is a simple cycle that
does not contain $\H^*_1$ and $\H^*_2$, we compute it using a similar tracing procedure, but this
time starting from the edge $e^* \in \beta^*(G,R)$ that is produced by the special procedure. To complete the description of the computation of $\beta^*(G,R)$, it remains to describe the special procedure.

\medskip
\noindent
{\bf Finding an edge \boldmath$e^*\in \VD^*(G,R)$ when \boldmath$\VD^*(G,R)$ does not meet the holes.} We split this procedure into two cases.

\medskip
\noindent
{\bf Case 1:} There is a path between $\H^*_1$ and $\H^*_2$ in $\VD^*(G)$.
(The case where there is a path between $\H^*_1$ and $\H^*_2$ in $\VD^*(R)$ is treated symmetrically.)
It is easy to verify that there is such a path if and only if some edge of $\VD^*(G)$ passes through
 $\H^*_2$ (as already noted, this information is available from the preprocessing stage).
Let $\pi^*(\H^*_1,\H^*_2)$ be such a path in $\VD^*(G)$ between $\H^*_1$ and $\H^*_2$. 
We say that an edge $e^*$ of is green (resp. red), if both endpoints of $e$ belong to green (resp. red) cells of $\VD(S)$. Note that if an edge is neither green nor red then it belongs, by definition, to $\beta^*(G,R)$, so our goal is to find such an edge. 
Since (we assume that) $\beta^*(G,R)$ is a simple cycle separating $G$ and $R$,
it follows that
the entire $\bd h_1$ (resp., $\bd h_2$) is contained in green (resp., red) cells of $\VD^*(S)$. Hence the first edge of $\pi^*(\H^*_1,\H^*_2)$ is green, and the last edge of $\pi^*(\H^*_1,\H^*_2)$ is red. 
It follows that $\pi^*(\H^*_1,\H^*_2)$ contains a vertex $x^*$ such that the two edges of $\pi^*(\H^*_1,\H^*_2)$ incident to $x^*$ are neither both green nor both red. 
Therefore, the primal face $x$ corresponding to $x^*$ has incident primal vertices both in green and red cells of $\VD(S)$, so some edge (in fact, at least two edges) incident to $x^*$ is an edge of $\beta^*(G,R)$. 
We first find such a vertex $x^*$ by binary search, and then find an incident edge of $\beta^*(G,R)$ by binary search as follows.

\begin{figure}[h]
\begin{center}
\includegraphics[scale=0.3]{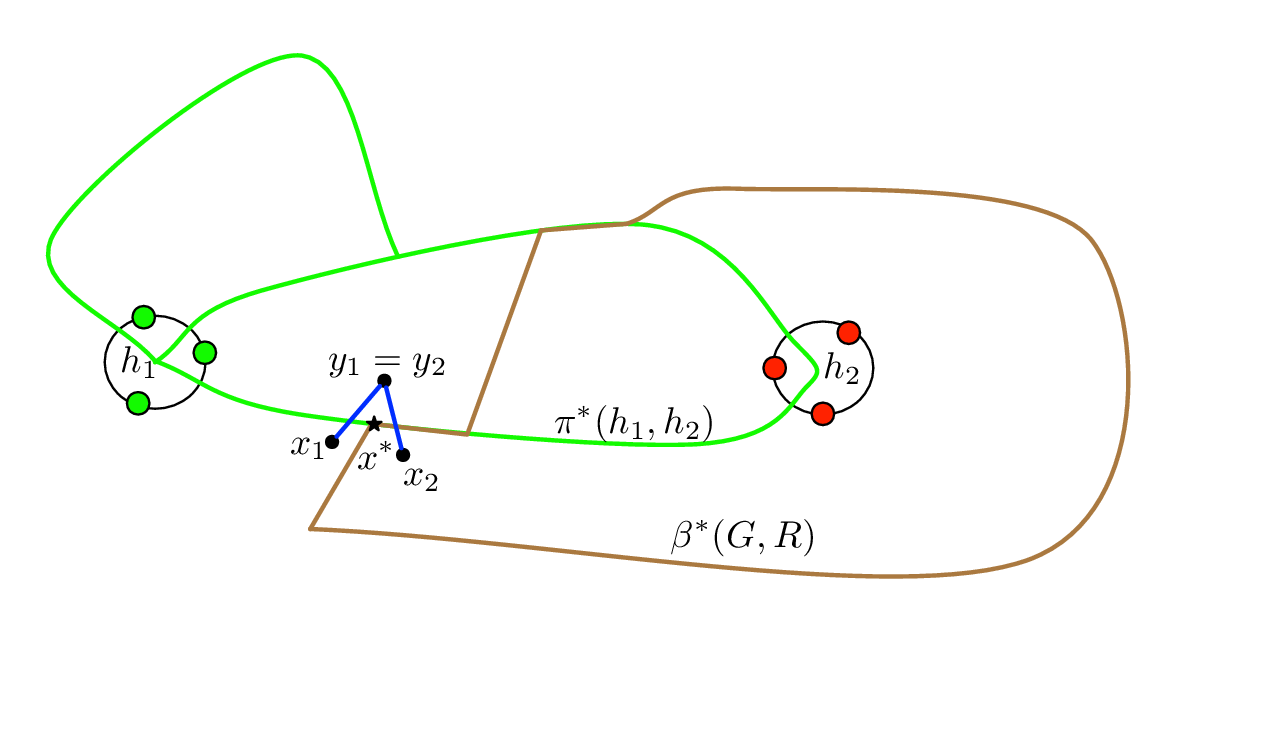}
\caption{Illustration of Case 1. The diagram $\VD^*(G)$ is shown in green. The bisector $\beta^*(G,R)$ is brown. In this example, $x_1$ and $x_2$ are in green cells of $\VD^*(G\cup R)$, and $y_1=y_2$ is in a red cell.
\label{fig:6-2-1}}
\end{center}
\end{figure}

Consider a candidate dual vertex $x^*$ of $\pi^*(\H^*_1,\H^*_2)$.
Denote the dual edges to the two edges of $\pi^*(\H^*_1,\H^*_2)$ incident to $x^*$ by $x_1y_1$ and $x_2y_2$. See Figure~\ref{fig:6-2-1}. 
We find by "brute force" the Voronoi cells to which the four primal vertices $x_1,x_2,y_1,y_2$ belong (i.e., by comparing the distances from each of these primal vertices to all sites in $S$). 
If not all four vertices belong to cells of the same color, the binary search is done. 
If all four vertices are in green (resp. red) cells we continue the search in the subpath of $\pi^*(\H^*_1,\H^*_2)$ between $x^*$ and $\H^*_2$ (resp. between $\H^*_1$ and $x^*$). 
This binary search takes $O(|S|\log n) = \tilde O(|S|)$ time.  

After identifying a vertex $x^*$ we need to identify an edge of $\beta^*(G,R)$ incident to $x^*$. Since, by Lemma~\ref{lem:GBcycle}, the degree of $x^*$ in $\beta^*(G,R)$ is 2, the primal vertices on the boundary of the primal face $x$ dual to $x^*$ can be partitioned into two maximal sets of consecutive vertices, such that all the vertices in one set are in green Voronoi cells in $\VD(S)$, and those in the other set are in red cells. Since we have already located a vertex in each set, the two edges of $\beta^*(G,R)$ incident to $x^*$ can be found by binary search, where at each step of the binary search we consider a single edge $e$ on the boundary of the face $x$, and identify, by "brute force", the Voronoi cells in $\VD(S)$ to which the two endpoints of $e$ belong. This binary search also takes $O(|S|\log n) = \tilde O(|S|)$ time.

\medskip
\noindent
{\bf Case 2:} The hole $\H_2$ is contained in a single cell of $\VD^*(G)$
and the hole $\H_1$ is contained in a single cell of $\VD^*(R)$
(that is, no edge of $\VD^*(G)$ (resp.,  $\VD^*(R)$) crosses any of $\bd \H_2$ (resp., $\bd \H_1$)).
Let $g$ be the green site such that all the vertices of $\bd \H_2$ lie in
$\Vor(g)$ in $\VD(G)$, and let $r$ be the red site such that all the vertices
of $\bd \H_1$ lie in $\Vor(r)$ in $\VD(R)$. 
We use Theorem~\ref{thm:tri-extended} to compute the trichromatic vertices of $\VD^*(g,r,R\setminus r)$ and of $\VD^*(r,g,G\setminus g)$. The requirements for applying the theorem are satisfied for the same reasons as in the single hole case.
There are two subcases depending on whether any trichromatic vertices were found in any of these two applications of Theorem~\ref{thm:tri-extended}.

\medskip \noindent
{\bf Subcase 2(a):}
If no trichromatic vertices were found, then 
either (i) $\beta^*(g,r)$ is empty, or (ii) the boundary of $\Vor^*(r)$ in $\VD^*(R)$
is disjoint from the boundary of $\Vor^*(g)$ in $\VD^*(G)$, and the bisector $\beta^*(g,r)$
separates these two boundaries.  See Figure~\ref{fig:6-2-2aii}. In case (i), suppose, without loss of generality, that $\beta^*(g,r)$ is empty because $\wt(g) + d_P(g,r) < \wt(r)$. Since we assumed that $\beta^*(G,R)$ does not intersect $h_2$, it must be that, for every site $r' \in R$, $\wt(g) + d_P(g,r') < \wt(r')$. Therefore, in this case, $\VD^*(S) = \VD^*(G)$. In case (ii), $\beta^*(G,R)$ is the bisector $\beta^*(g,r)$.

\begin{figure}[h]
\begin{center}
\includegraphics[scale=0.75]{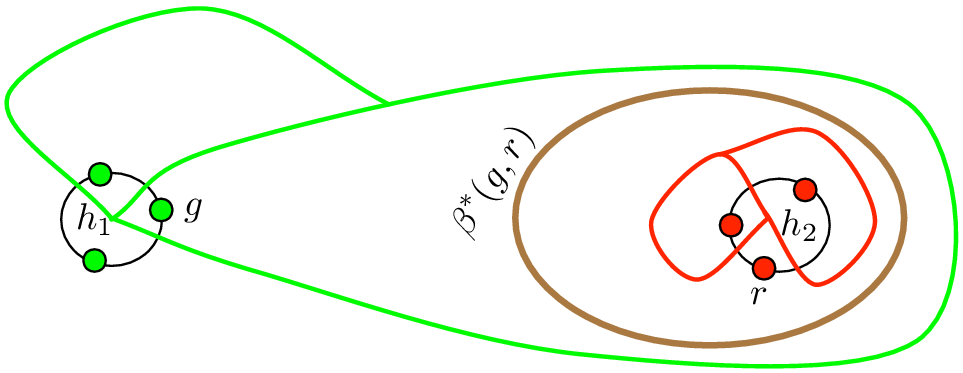}
\caption{Ilustration of Subcase 2(a).ii. The diagram $\VD^*(G)$ is shown in green. The diagram $\VD^*(R)$ is shown in red. The bisector $\beta^*(g,r)$ is brown. In this example, $\beta^*(g,r) = \beta^*(G,R)$.
\label{fig:6-2-2aii}}
\end{center}
\end{figure}

\medskip \noindent
{\bf Subcase 2(b):} 
Otherwise, at least one of the trichromatic vertices of $\VD^*(g,r,R\setminus r)$ and of $\VD^*(r,g,G\setminus g)$ must be a vertex of $\beta^*(G,R)$. To see this consider a curve $\pi^*$ in the dual connecting $h_1^*$ and $h_2^*$  that lies inside or on the boundary of both $\Vor^*(g)$ in $\VD^*(G)$ and $\Vor^*(r)$ in $\VD^*(R)$. Since the bisector $\beta^*(r,g)$ separates $h_1$ from $h_2$, it must cross $\pi^*$ at some vertex $x^*$. Follow $\beta^*(r,g)$ from $x^*$ until it first intersects a Voronoi edge of either $\Vor^*(g)$ or $\Vor^*(r)$ at some dual vertex $y^*$. Note that $y^*$ is a trichromatic vertex of either $\VD^*(g,r,R\setminus r)$ or $\VD^*(r,g,G\setminus g)$. It follows that 
the portion of $\beta^*(r,g)$ between $x^*$ and $y^*$ is a portion of $\beta^*(G,R)$, and that $y^*$ is a dual vertex on $\beta^*(R,G)$.
See Figure~\ref{fig:6-2-2b}.
If $y^*$ is not a hole, we can find which of the 3 incident edges belongs to $\beta^*(G,R)$ by "brute force". If $y^*$ is a hole then one of the two edges of $\beta^*(g,r)$ incident to $y^*$ is also an edge of $\beta^*(G,R)$. Since the edges of $\beta^*(g,r)$ incident to each hole were stored when $\beta^*(g,r)$ was constructed by the sweeping algorithm at preprocessing, we can access each of them and check which is an edge of $\beta^*(G,R)$ by "brute force". 
\begin{figure}[h]
\begin{center}
\includegraphics[scale=0.65]{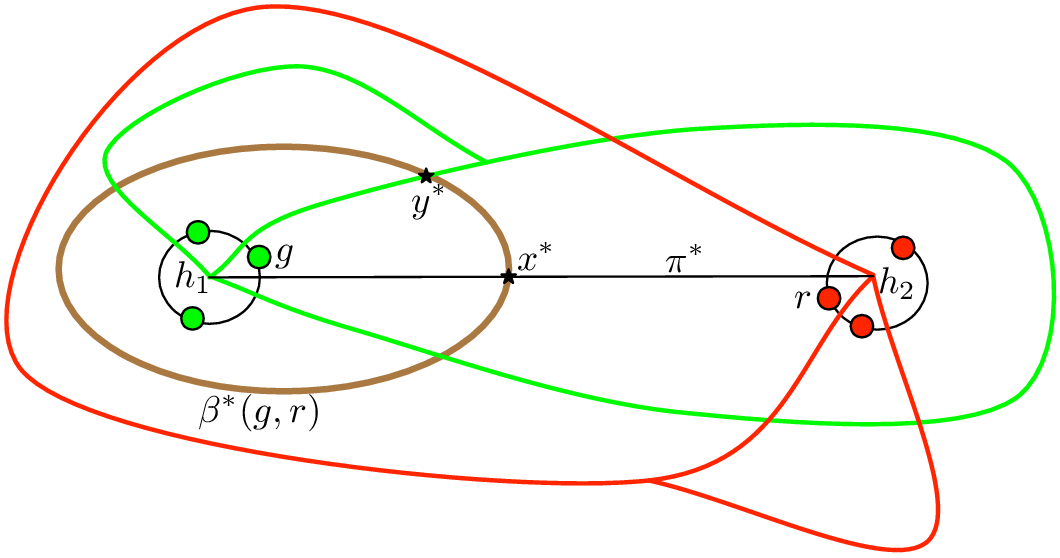}
\caption{Ilustration of Subcase 2(b). The diagram $\VD^*(G)$ is shown in green. The diagram $\VD^*(R)$ is shown in red. The bisector $\beta^*(g,r)$ is brown. 
\label{fig:6-2-2b}}
\end{center}
\end{figure}

Once we have computed $\beta^*(G,R)$, we paste it with the appropriate portions of
$\VD^*(G)$ and $\VD^*(R)$, to obtain $\VD^*(S)$, in a straightforward manner as we combined
$\beta^*(G,R)$ with $\VD^*(G)$ and $\VD^*(R)$ in Section \ref{sec:single}.

\subsection{Three holes} \label{sec:trivert}

The next stage is to compute all the $O(1)$ ``trichromatic'' Voronoi diagrams, where each
such diagram is of the form $\VD^*(S_i\cup S_j\cup S_k)$, for distinct indices
$1\le i<j<k\le \NH$, where $S_i$ (resp., $S_j$, $S_k$) consists of all the sites on
the boundary of $i$-th hole (resp., $j$-th hole, $k$-th hole) of $P$.

Extending Lemma~\ref{lem:2vvert}, we get the following interesting property.
\begin{lemma} \label{lem:2v3vert}
Let $S$, $S'$, $S''$ be three subsets of sites, each appearing consecutively on
the boundary of a single hole of $P$ (when the holes are not distinct, the corresponding
subsets are assumed to be separated along the common hole boundary). Then there are at
most two trichromatic Voronoi vertices in
$\VD^*(S\cup S'\cup S'')$. That is, there are at most two faces of $P$ for which
there exist sites $u\in S$, $v\in S'$, and $w\in S''$, such that each of the cells
$\Vor(u)$, $\Vor(v)$, $\Vor(w)$ contains a single vertex of $f$.
\end{lemma}
\begin{proof}
The proof proceeds essentially as in the proof of Lemma~\ref{lem:2vvert}, by showing that the existence
of three trichromatic vertices would lead to an impossible plane drawing of $K_{3,3}$. Here
one set of vertices are points inside the faces dual to the vertices, and another set
are points inside the holes whose boundaries contain the sets $S$, $S'$, $S''$.
(When the holes are not distinct, we draw a separate point inside the hole for each subset
along its boundary; the fact that these subsets are separated allows us to embed the relevant edges
in a crossing-free manner.)
\hfill \fullqed \end{proof}

Fix three distinct holes, call them $\H_1$, $\H_2$, $\H_3$, and let $S_1$, $S_2$, $S_3$
be the respective subsets of sites, such that $S_i$ lies on $\bd \H_i$, for $i=1,2,3$.
The preceding procedure has already computed the three bichromatic diagrams
$\VD^*(S_1\cup S_2)$, $\VD^*(S_1\cup S_3)$, and $\VD^*(S_2\cup S_3)$.
We now want to merge them into the trichromatic diagram $\VD^*(S_1\cup S_2\cup S_3)$.

We take each bichromatic bisector of the form $\beta^*(u,v)$, for $u\in S_1$
and $v\in S_2$, that contains a Voronoi edge in $\VD^*(S_1\cup S_2)$, and note
that there are only $O(|S_1|+|S_2|)=O(|S|)$ such bisectors.
For each such bisector $\beta^*(u,v)$, we take the Voronoi cell $\Vor^*(u)$
in $\VD^*(S_1\cup S_3)$ (assuming it is not empty), and denote it for specificity as
$\Vor^*_{13}(u)$. 
Let $W_{13}(u)$ be the set of all sites $w\in S_3$ such that there is a segment of the 
bisector $\beta^*(u,w)$ in $\Vor^*_{13}(u)$. 
We would like to find trichromatic vertices of $\VD^*(v,u,W_{13}(u))$ that lie on an edge of $\Vor^*_{13}(u)$. Such vertices are trichromatic vertices of $\VD^*(S_1 \cup S_2 \cup S_3)$. 
To find such vertices we would like to apply Theorem~\ref{thm:tri-extended} to $v,u$, and $W_{13}(u)$. Requirement (ii) in the statement of Theorem~\ref{thm:tri-extended} is satisfied as argued in the single hole case. 
Requirement (iii) might not be satisfied because $\beta^*(u,W_{13}(u))$ might not be fully represented in $\VD^*(S_1 \cup S_3)$ ($\Vor^*_{13}(u)$ contains parts of $\beta^*(u,W_{13}(u))$, but might also contain segments of bisectors between $u$ and other sites in $S_1$).
To satisfy requirement (iii) we compute $\Vor^*(\{u\} \cup W_{13}(u))$ (this is a diagram involving sites on just two holes so we can use the algorithm in Section~\ref{sec:double}). This takes $\tilde O(|\Vor^*_{13}(u)|)$ time. The boundary of the Voronoi cell of $u$ in this diagram is the bisector $\beta^*(u,W_{13}(u))$. 
We can therefore apply Theorem~\ref{thm:tri-extended} to $v,u$, and $W_{13}(u)$ in $\tilde O(1)$ time.
If we get two trichromatic vertices, we test each of them to see whether it lies on a Voronoi edge (edges) of $\Vor^*_{13}(u)$. 
Each such vertex is a Voronoi vertex of $\VD^*(S_1 \cup S_2 \cup S_3)$, as is easily verified, and otherwise it is not.

By checking all bisectors $\beta^*(u,v)$ that contain a Voronoi edge in $\VD^*(S_1\cup S_2)$ we find the zero, one or two trichromatic vertices of $\VD^*(S_1 \cup S_2 \cup S_3)$. Constructing $\VD^*(S_1 \cup S_2 \cup S_3)$ is now done in a similar manner to the wrap-up stage in the single hole case, as follows. Assume first that two trichromatic vertices $u^*$ and $v^*$ were found. Consider $\beta^*(S_1,S_2)$. Recall that it is represented as a binary search tree over segments of bisectors of the form $\beta^*(u_i,v_j)$ for $u_i \in S_1, v_j\in S_2$. Both $u^*$ and $v^*$ are vertices on $\beta^*(S_1,S_2)$. We cut $\beta^*(S_1,S_2)$ at $u^*$ and $v^*$, and keep just the subpath $\beta^*_{12}$ of $\beta^*(S_1,S_2)$ that belongs to $\VD^*(S_1 \cup S_2 \cup S_3)$. In a similar manner we obtain the subpaths $\beta^*_{13}$ of $\beta^*(S_1,S_3)$ and $\beta^*_{23}$ of $\beta^*(S_2,S_3)$ that belong to $\VD^*(S_1 \cup S_2 \cup S_3)$. 
The concatenation of $\beta^*_{12}$ and $\beta^*_{13}$ is the cycle forming the boundary of the Voronoi cell of $S_1$ in $\VD^*(S_1 \cup S_2 \cup S_3)$. 
We therefore cut the DCEL representation of $\VD^*(S_1)$ along the concatenation of $\beta^*_{12}$ and $\beta^*_{13}$ and keep just the side that corresponds to the Voronoi cell of $S_1$ in $\VD^*(S_1 \cup S_2 \cup S_3)$. In a similar manner we obtain the DCEL representation of the the Voronoi cells of $S_2$ and of $S_3$ in $\VD^*(S_1 \cup S_2 \cup S_3)$. We then paste the three cells into the DCEL representation of $\VD^*(S_1\cup S_2 \cup S_3)$ by identifying their common boundaries along $\beta^*_{12}, \beta^*_{13}$ and $\beta^*_{23}$.

The degenerate cases of a single or no trichromatic vertex in $\VD^*(S_1 \cup S_2 \cup S_3)$ are similar to the case of two trichromatic vertices.

The total time to handle each site $u \in S_1$ is $\tilde O(|\Vor_{13}(u)| + |\Vor_{12}(u)|)$. The first term is for computing $\beta^*(u,W_{13}(u))$, and the second term accounts for all the applications of Theorem~\ref{thm:tri-extended} involving $u$. Summing over all $u \in S_1$ we get $\tilde O(|S_1|+|S_2|+|S_3|)$. 
 This completes the description of the construction of $\VD(S_1\cup S_2\cup S_3)$
for the sites on a fixed triple of holes. We repeat this for all (the constant
number of) such triples, and obtain all the trichromatic diagrams.
It follows from the analysis presented above that the overall running time of the
constructions is $\tilde{O}(|S|)$.

\subsection{Handling multiple holes} \label{sec:final}

The procedures presented so far construct all the trichromatic Voronoi diagrams,
namely all the diagrams of the form $\VD^*(S_i\cup S_j\cup S_k)$,
for distinct indices $1\le i<j<k\le \NH$, where $S_i$ (resp., $S_j$, $S_k$) consists
of all the sites on the boundary of the $i$-th hole (resp., $j$-th hole, $k$-th hole)
of $P$. The final step in the construction, presented in this subsection, produces the
overall diagram, of the entire set $S$, from these partial diagrams.
(We assume here that there are at least four holes, otherwise this stage is not needed.)

Clearly, every Voronoi vertex in $\VD^*(S)$ is also a Voronoi vertex in some trichromatic
diagram, so the set of all the trichromatic vertices, over all the $\binom{\NH}{3}$
triples of holes, is a superset of the actual vertices of $\VD^*(S)$. (To be precise,
these sets also include ``monochromatic'' and ``bichromatic'' vertices, determined by
sites that belong to just one hole or to two holes.) To find the real final Voronoi
vertices and edges, we proceed as follows.

For each bisector $\beta^*(u,v)$, say, with $u\in S_1$ and $v\in S_2$, where $S_1$
and $S_2$ are the subsets of sites on the boundaries of two fixed respective holes
$\H_1$, $\H_2$, we mark along $\beta^*(u,v)$ all the Voronoi edges and vertices that show up in
any of the trichromatic diagrams $\VD^*(S_1\cup S_2\cup S_k)$, for $k=3,\ldots,\NH$.
We sort the resulting vertices in their order along $\beta^*(u,v)$, and thereby obtain a
partition of $\beta^*(u,v)$ into ``atomic'' intervals, each delimited by two consecutive (real or fake) vertices.

We then count, for each atomic interval $I$, in how many trichromatic diagrams it
is contained in a Voronoi edge (of that diagram). This is easy to do by computing this number
for some initial interval in brute force (which takes $\tilde{O}(1)$ time), and then by
updating the count by $\pm 1$ as we move from one atomic interval to the next one.

As is clear from the construction, the portions of $\beta^*(u,v)$ that are
Voronoi edges in the full diagram $\VD^*(S)$ are precisely those that are covered by exactly $\NH-2$
trichromatic Voronoi edges if $i\not = j$, or by $\NH-1$ edges if $i=j$.
(This is the number of possible indices $k\not=i,j$ in each of these cases.)

We can therefore identify in this manner all the true Voronoi edges along each bisector,
and also the true Voronoi vertices (which are the endpoints of the true edges). Each true vertex
arises in this manner three times. By matching these three occurrences, we can ``stitch'' together
the edges and vertices of $\VD^*(S)$ into a single planar graph that represents the diagram.
It is also easy to augment this graph in situations where the diagram is not connected, using
standard features of the DCEL data structure~\cite{dBCOvK}.

\fi

 % max

\section{Preprocessing for max queries}\label{sec:prep_max}
In this section we establish the last part (item $(ii)$) of Theorem~\ref{thm:vor}. Recall that $P$ contains $r$ vertices and a set $S$ of $b$ sites. We wish to preprocess $P$ (independently of any weight assignment to the sites) in $\tilde O(rb^2)$ time, so that, given the cell $\Vor^*(v)$ of a site $v$, we can return the vertex $w \in \Vor(v)$ maximizing $d(v,w)$  in  time that is linear (up to polylogarithmic factors) in the number of Voronoi vertices of $\Vor^*(v)$. Cabello described a similar mechanism, but since in his Voronoi diagrams the sites are on a single hole, $\VD^*(S)$ is connected, so the boundary of each Voronoi region is a single cycle. As our treatment is more general, we need to handle the case of sites on multiple holes, where the boundary of Voronoi cells might consist of multiple cycles. Our approach for the single hole case is essentially that of Cabello, although our presentation is different. 
We then extend the technique to handle sites on multiple holes. 
\ifdefined\fullver 

Let $T$ 
be the shortest-path tree rooted at $v$ in $P$ and let $T^*$ be the cotree of $T$. (That is,
$T^*$ is the tree whose edges are the duals of the edges which are not in $T$.)
Let $h_\infty$ 
be the hole
such that $v\in \bd h_\infty$, and let
$h^*_\infty$ be the  vertex dual to $h_\infty$.
We root $T^*$ at $h^*_\infty$ and when we refer to an edge $f^*g^*\in T^*$ then $f^*$ is the parent of $g^*$ in  $T^*$
(i.e., $f^*$ is closer to $h^*_\infty$ than $g^*$).

We label the vertices of $P$ according to their distance (in $P$) from $v$ and we label each face $f$ (and its corresponding dual vertex $f^*$)  that is not a hole of $P$ with the maximum label of a vertex incident to $f$. Note that the trees $T$, $T^*$ and the labels can be easily computed and stored in $\tilde O(rb)$ time.

We will use the following properties of the boundary of the Voronoi cell of $v$ in any Voronoi diagram $VD(S)$ in $P$. 

\begin{lemma}
The boundary of $\Vor^*(v)$ is a set $\mathcal C^*$ of at most $t$ vertex-disjoint non-self-crossing cycles, where $t$ is the number of holes of $P$.
\end{lemma}
\begin{proof}
Add an artificial vertex $x_h$ in each hole $h$ of $P$, and artificial infinite-length edges connecting $x_h$ to every vertex of $h$. The resulting piece $P'$ has no holes. We consider $x_h$ to belong to $\Vor(v)$ in $\VD(S)$ if and only if all vertices of $h$ belong to $\Vor(v)$. It follows that $P' \setminus \Vor(v)$ has at most $t$ connected components. Hence, by the duality of cuts and cycles in planar graphs (see, e.g., \cite{planarbook}), the boundary of $\Vor^*(v)$ is a set $\mathcal C^*$ of at most $t$ mutually non-crossing simple cycles in $P'^*$. 
Since deleting the artificial edges is equivalent to contracting the duals of the artificial edges,  $\mathcal C^*$ becomes a set of at most $t$ mutually non-crossing non-self-crossing cycles in $P^*$. Since no two cycles cross, if any two cycles share a vertex, they can be combined into a single non-self-crossing cycle.
\end{proof}

We assume that $\mathcal C^*$ contains a cycle
$C^*_0$ through $h^*_\infty$ that encloses all the other cycles in $\mathcal C^*$. See Figures~\ref{fig:grid2},~\ref{fig:grid1}.
(If there is no such cycle then we add a dummy cycle enclosing all other cycles, which is  a self loop
through $h^*_\infty$.)
Furthermore, since $\Vor(v)$ is connected, no cycle  $C^*_1\in \mathcal C^*\setminus \{ C^*_0 \}$
is enclosed by a cycle $C^*_2\in \mathcal C^*\setminus \{ C^*_0 \}$.

Our goal is to report the maximum label of a vertex that is enclosed by $C^*_0$ and not enclosed by any $C^* \in \mathcal C^* \setminus \{ C^*_0 \}$. 

We say that a face $f$ of $P$ belongs to $\Vor(v)$ if all vertices of $f$ belong to $\Vor(v)$.
Let  $C^*$ be a cycle in $\mathcal C^*$. 
We say that an arc $f^*g^*$ of $T^*$ {\em penetrates} $\Vor^*(v)$ at $C^*$
(or, in short, {\em penetrates} at $C^*$)
if the face $f$ does not belongs to $C^*$ %to $\Vor(v)$ 
and the face $g$ belongs to $\Vor(v)$. 
%Note that this implies that $f^* \in C^*$.  
We say that an arc $f^*g^*$ of $T^*$
{\em exits} $\Vor^*(v)$ at $C^*$ if the face $g$ belongs to $C^*$ and the face $f$ belongs to $\Vor(v)$. 

\begin{figure}[h]
\begin{center}
\includegraphics[scale=0.7]{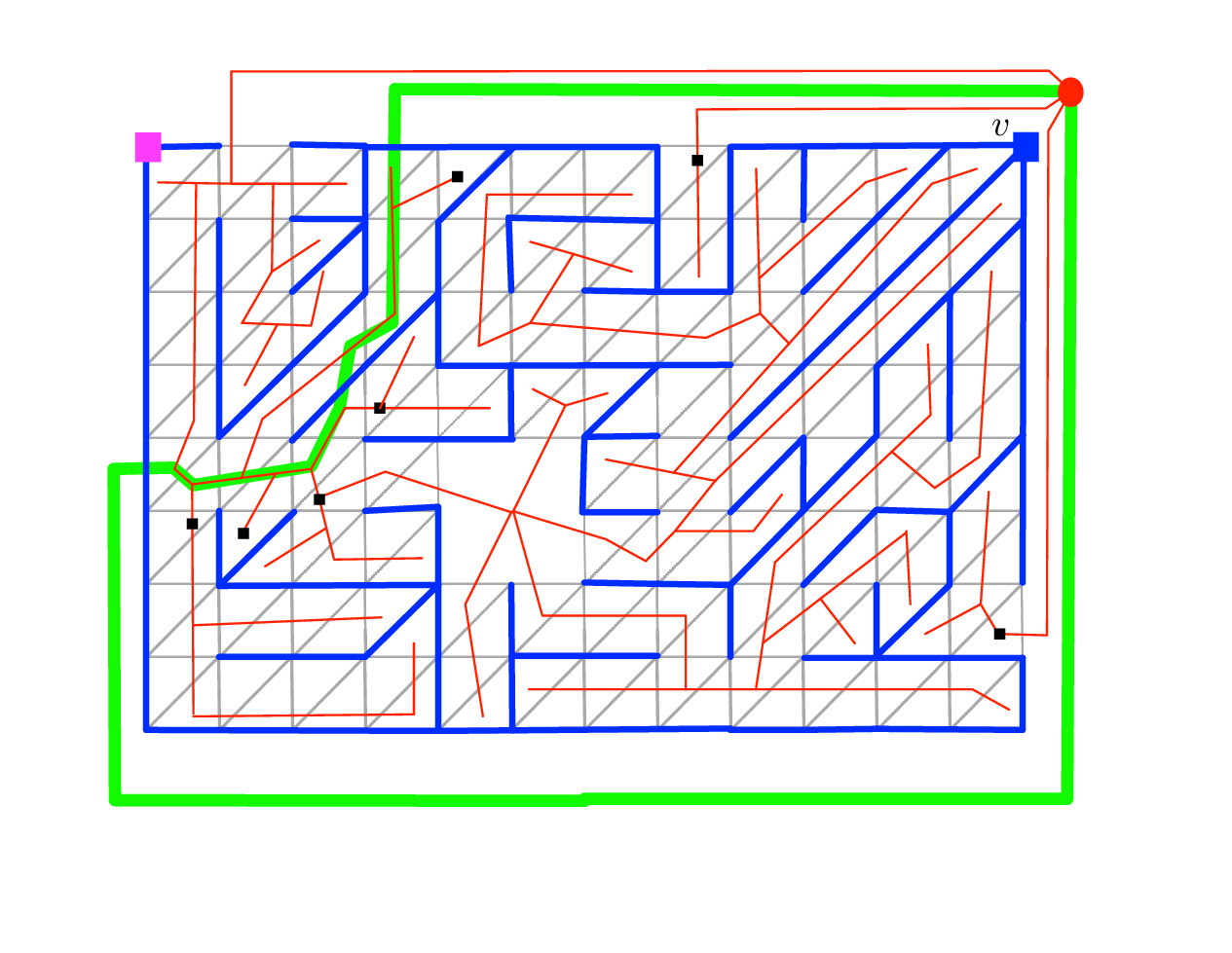}
\caption{A graph with two holes. Two sites are shown, both on the external hole (blue and pink squares). The site on the upper right corner is $v$. In this example, $\Vor^*(v)$ (green) consists of one simple cycle that goes through the infinite face. $T$ is shown in blue and $T^*$ in red. Heads of arcs penetrating at $C^*_0$ are indicated by black squares.
\label{fig:grid2}
}
\end{center}
\end{figure}

\begin{lemma} \label{lem:T*C*0}
Every root-to-leaf path in $T^*$ contains at most one arc penetrating at $C^*_0$ and no arc exiting at $C^*_0$.
\end{lemma}
\begin{proof}
Clearly, for a root-to-leaf path of $T^*$ to contain two arcs penetrating at $C^*_0$, there must be an arc
between them on the path that exits at $C^*_0$. It therefore suffices to show there are no exiting arcs. Since the root $h^*_\infty$ is not strictly enclosed by $C^*_0$, any root-to-leaf path of $T^*$ must penetrate at $C^*_0$ before exiting at $C^*_0$. Suppose for the sake of contradiction that  $f^*_0g^*_0$ is a penetrating arc and $f^*_1g^*_1$ is its descendant exiting arc. Consider the cycle formed by the $f^*_0$-to-$g^*_1$ path in $T^*$ and a $f^*_0$-to-$g^*_1$ subpath of $C^*_0$. This cycle partitions $\Vor(v)$ into two nonempty regions
which gives a
 contradiction  since $T$ does not cross $T^*$, but the subtree of $T$ spanning $\Vor(v)$ is connected.
\hfill \fullqed \end{proof}

To simplify the presentation we
 first assume that $C_0^*$ is the only cycle in $\mathcal C^*$ (this is the case, for example, when all the sites are on a single hole).
 %$h_\infty$ is the only hole in $P$. 
 We will remove this assumption later.
%If we assume that $h_\infty$ is the only hole in $P$ then $C_0^*$ must be  the only cycle in $\mathcal C^*$.
Each vertex $u\in \Vor(v)$ is either adjacent to an arc of $C_0^*$ or to a face $f$ which is strictly enclosed by $C^*_0$.
By Lemma~\ref{lem:T*C*0}, for any penetrating arc $f^*g^* \in T^*$, the subtree of $T^*$ rooted at $g^*$ consists only of faces (of $P$, i.e., vertices of $P^*$) strictly enclosed by $C^*_0$. Conversely, since the root of $T^*$ is not strictly enclosed by $C^*_0$, every face 
which is strictly enclosed by $C^*_0$ belongs to the subtree of $g^*$ for some penetrating arc $f^*g^* \in T^*$. It follows that it suffices to find the maximum label among the labels of the (primal) vertices
of $\Vor(v)$ adjacent to arcs of $C_0^*$ and of the labels of  the dual vertices in each of the subtrees of $T^*$ rooted at an arc
penetrating at $C^*_0$.

\smallskip
\noindent
{\bf Preprocessing.}
For each node  $g^*\in  T^*$ we compute the maximum label of a node in the subtree of $g^*$ in $T^*$ (denoted by $T^*_{g^*}$).
Then, we extend the computation of the
 persistent binary search trees representing the
bisectors $\beta^*(v,w)$ (see Section~\ref{sec:bisectors}) to store, with
each dual vertex $f^*$ on a bisector $\beta^*(v,w)$ that is not a hole, a value $\ell(f^*)$, defined as follows.
Let $g^*$ be the neighbor of $f^*$ that does not belong to $\beta^*(v,w)$ (recall that $f^*$ has exactly three incident arcs, two of which belong to $\beta^*(v,w)$). If the arc $f^*g^*$ is an arc of  $T^*$ that penetrates $\beta^*(v,w)$ (i.e., if $f^*$ is on $\beta^*(v,w)$ and $g^*$ belongs to $\Vor(v)$ in $\VD(\{v,w\})$), we set $\ell(f^*)$ to  the maximum label of a node in $T^*_{g^*}$. Otherwise we set $\ell(f^*)$ to $0$.
Next, with each arc $e^*$ of $\beta^*(v,w)$ we store a value $\ell(e^*)$
which is equal to the label of the endpoint of the primal arc $e$ that is closer to $v$.
In addition, we store at the nodes of the binary search tree representing each bisector $\beta^*(v,w)$,
the maximum of the
 $\ell(\cdot)$ values of the arcs and vertices in its subtree (of the binary search tree).
 This allows us to compute the maximum value of a node  in
a segment of $\beta^*(v,w)$ in logarithmic time. 
In addition, we construct a data structure for range maximum queries (RMQ) for the arcs incident to $h^*_\infty$ in their cyclic order
around $h_\infty$,
where the value $\ell(h^*_\infty g^*)$ of each arc $h^*_\infty g^*$ is defined to be the maximum label of a vertex of
$T^*_{g^*}$.
All these additions to the preprocessing can be done within the $O~O(rb^2)$ time to compute the bisectors in Theorem~\ref{thm:bisectors}.

\smallskip
\noindent
{\bf Query.}
For the query, we are given the representation of $C^*_0$ as a sequence of segments of bisectors $\beta^*(v,\cdot)$, and we need to report the maximum label of a vertex in $\Vor(v)$.
 We query each segment $\gamma^*$ of a bisector
on  $C^*_0$  for the maximum value of $\ell(f^*)$ and $\ell(e^*)$
among the interior vertices $f^*$ and the arcs of $\gamma^*$, respectively. This takes $O(\log r)$ time per segment. In addition, for each two consecutive segments with a shared dual vertex $f^*$, we check whether the arc $f^*g^*$ incident to $f^*$ but not on $C^*_0$ belongs to $T^*$ and penetrates at $C^*_0$. If so, we take  the maximum label of a node in $T^*_{g^*}$ as a candidate value
for the maximum distance as well. 
The reason endpoints of segments are treated differently than interior vertices of a segment is that at an interior vertex $f^*$, the two arcs of the bisector incident to $f^*$ are also arcs of $C_0^*$, whereas when $f^*$ is the meeting point of two segments of two distinct bisectors this is not the case. 
Finally, if $g^*h^*_\infty$ and $h^*_\infty g^{\prime *}$ are two consecutive arcs of $C^*_0$ incident to $h^*_\infty$, we query the RMQ data structure of $h^*_\infty$ for the maximum value associated with the arcs
(strictly) between $h^*_\infty g^*$ and $h^*_\infty g^{\prime *}$ in the cyclic order around $h^*_\infty$ that are enclosed by $C^*_0$.
The desired maximum value is the largest among all these candidate values.

\subsection{Working with multiple cycles}\label{sec:max_holes}
We now remove the assumption $\mathcal C^*$ consists of just $C_0$. The challenge in this case is that a branch of $T^*$ that penetrates $C^*_0$ may go in and out of $\Vor^*(v)$ through other cycles in $\mathcal C^*$. We show, however, that because the sites lie on a constant number of holes, there is a constant number of problematic branches of $T^*$ and we can handle them with a special data structure.
Recall that we would like to report the maximum label of a (primal) vertex
strictly enclosed by $C^*_0$ but not enclosed by any $C^* \in \mathcal C^* \setminus \{ C^*_0 \}$.

Consider traversing a root-to-leaf path $Q$ in $T^*$. It starts with a prefix of faces that are not strictly enclosed by $C^*_0$, followed by an arc $f^*g^*$ that penetrates $\Vor^*(v)$ at $C^*_0$. By Lemma~\ref{lem:T*C*0}, the suffix of $Q$ starting at $g^*$ is strictly enclosed by $C^*_0$. This suffix of $Q$ is not enclosed by any other cycle in $\mathcal C^*$ until it uses an arc that exits $\Vor^*(v)$ at some $C^* \in \mathcal C^*$. From that point on it is enclosed by $C^*$ until it encounters an arc
that penetrates at $C^*$, and so on. Therefore, we would like to report the maximum in any subtree rooted at $g^*$ for any penetrating arc $f^*g^*$ (at any $C^* \in \mathcal C^*)$, but remove from each such subtree the subtrees rooted at $g^*$ for any exiting arc $f^*g^*$. The next two lemmas characterize the exiting arcs.

\begin{figure*}[h!]
\begin{center}
\includegraphics[scale=0.8, clip=true, trim = 0mm 0mm 0mm 0mm]{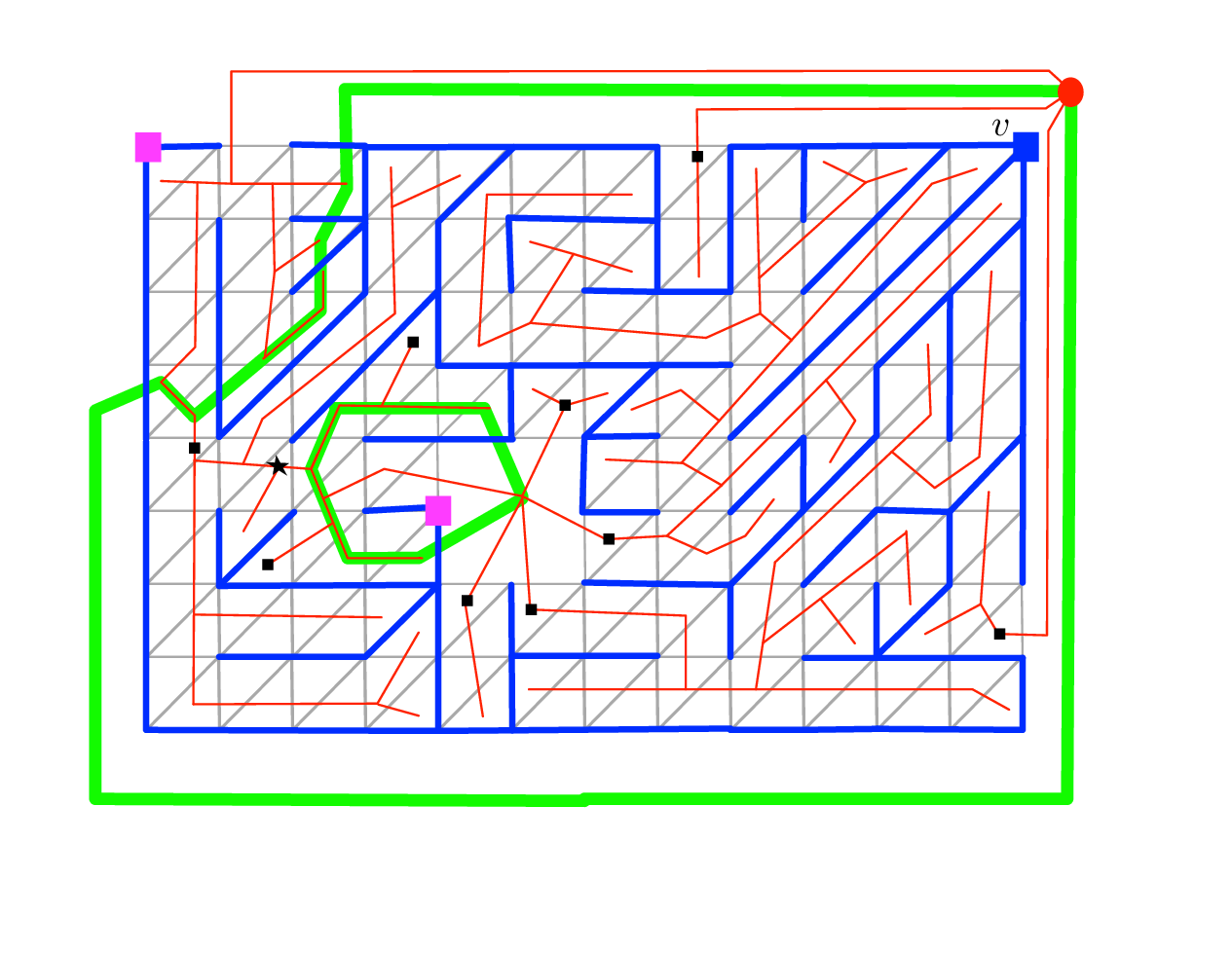}
\caption{A graph with two holes. Three sites are shown (blue and pink squares), two on the external hole and one on the inner hole. The blue site on the upper right corner is $v$. In this example, $\Vor^*(v)$ (green) consists of two cycles. $T$ is shown in blue and $T^*$ in red. Heads of arcs penetrating $\Vor^*(v)$ are indicated by black squares. The tail of an arc exiting $\Vor^*(v)$ is indicated by a black star.
\label{fig:grid1}}
\end{center}
\end{figure*}

\begin{lemma} \label{lem:T*C*}
    For every cycle $C^* \in \mathcal C \setminus C^*_0$ there is exactly one arc of $T^*$ that exits at $C^*$.
\end{lemma}
\begin{proof}
    Since $T^*$ is a spanning tree there is at least one such arc. Assume there are two arcs $f^*_0g^*_0, f^*_1g^*_1$ exiting at $C^*$. The cycle formed by one of the portions of $C^*$ between $g^*_0, g^*_1$ and the unique $g^*_0$-to-$g^*_1$ path  in $T^*$ (not necessarily directed) partitions $\Vor(v)$ into two nonempty regions. A contradiction arises since $T$ does not cross $T^*$, but the subtree of $T$ spanning $\Vor(v)$ is connected.
\hfill \fullqed \end{proof}

\begin{corollary}
Let $h \neq h_\infty$ be a hole, and suppose
that there is a cycle $C^* \in \mathcal C^* \setminus C^*_0$ that either strictly encloses $h^*$ or goes through $h^*$. Then the unique arc of $T^*$ that exits at $C^*$ is on the path from  $h_\infty^*$ to $h^*$ in $T^*$.
\end{corollary}
\begin{proof}
Observe that since $h^* \neq h^*_\infty$ is on $C^*$ or strictly enclosed by $C^*$ then, by definition of the boundary cycle $C^*$, the hole $h$ does not belong to $\Vor(v)$. 
Therefore, the path in $T^*$ from $h^*_\infty$ to $h^*$ must exit $\Vor^*(v)$ at $C^*$ to get to $h^*$. By Lemma~\ref{lem:T*C*}, this is the unique arc exiting at $C^*$.
\hfill \fullqed \end{proof}

\smallskip
\noindent
{\bf Preprocessing (with multiple holes).}
 For every hole $h\in {\mathcal H}$, let $T^*[h^*]$ denote the
path from $h_\infty^*$ to $h^*$ (the vertex dual to $h$) in $T^*$.
Consider the subtree $T_{\mathcal H}^*$ of $T^*$ which is the union of the paths
$T^*[h^*]$ over all holes $h\in {\mathcal H}$.

We construct a data structure that can answer the following queries on
$T_{\mathcal H}^*$. Given
dual vertices $v^*_1,v^*_2, \dots v^*_{k}$ in $T_{\mathcal H}^*$ (where $k \leq t$) 
such that
$v^*_1$ is an ancestor of all the other $v^*_i$'s,
return the maximum label of a vertex in $T^*_{v^*_1} \setminus \bigcup_i T^*_{v^*_i}$.
We obtain this data structure
by contracting all the arcs of $T^*\setminus  T_{\mathcal H}^*$ into their ancestors in
$T_{\mathcal H}^*$, keeping in each vertex $v^*$ of
$T_{\mathcal H}^*$ the maximum label of a vertex in the subgraph contracted to $v^*$.
We then represent $T_{\mathcal H}^*$ (with these maxima as the vertex values)  by an  Euler tour tree \cite{HenzingerK97}.
We can construct this data structure
 within the $\tilde O(r)$ time that it takes to compute $T^*$.

At the vertices $f^*$ of the bisectors $\beta^*(v,w)$ we store the labels  $\ell(f^*)$
which are defined as before with the difference that
we set $\ell(f^*) = 0 $
if the penetrating arc $f^*g^*$ is in $T_{\mathcal H}^*$ (in the query, these arcs will be taken into account separately).
We also store, for each vertex $f^* \in \beta^*(v,w)$, which does not correspond to a hole, a boolean flag $b(f^*)$
which is $1$, 
if the arc $e^*$ incident to $f^*$ that is not an arc of $\beta^*(v,w)$, belongs to $T_{\mathcal H}^*$ and $e^*$ has one endpoint on $\beta^*(v,w)$ and the other endpoint belongs to $\Vor(v)$ in $\VD(v,w)$ (i.e., $e^*$ either penetrates or exits $\Vor(v)$ at $\beta^*(v,w)$.
 We maintain these flags so that we can list all such dual vertices in any segment of  $\beta^*(v,w)$ in $O(\log r)$ time per dual vertex.

Finally, we maintain an RMQ data structure over the arcs of each hole $h$ in their cyclic order around $h$.
The value $\ell(h^*g^*)$ of each arc $h^*g^*$ is defined to be the maximum label of a vertex
in $T^*_g$ if
$h^*g^* \in T^* \setminus T_{\mathcal H}^*$ and $0$ otherwise.
 We also store with $h^*$ the indices of the (constant number of) arcs of $T_{\mathcal H}^*$ that are incident to $h^*$.

\smallskip
\noindent
{\bf Query (with multiple holes).}
We traverse the cycles $C^* \in \mathcal C^*$.
For each cycle $C^*$, we query the representation of
each segment $\gamma^*$ of a bisector $\beta^*(v,\cdot)$ in
$C^*$  and find the maximum value of $\ell(f^*)$ over all interior vertices  $f^*\in \gamma^*$ and arcs $e^* \in \gamma^*$.
This takes  $O(\log r)$ time per segment.
The maximum in each segment is a candidate for the maximum distance from $v$ to a vertex in $\Vor(v)$.
 In addition, for each pair of consecutive segments with a common dual vertex $f^*$, we check whether the third arc $f^*g^*$ belongs to $T^*\setminus T_{\mathcal H}^*$  and penetrates at $C^*$. If so, we consider the maximum label in $T^*_g$ as a candidate
for the maximum as well. Finally,
for each hole $h$ that $C^*$ crosses we identify the pair of consecutive arcs
$g^*h^*$ and $h^*g^{\prime *}$ of $C^*$ which are incident to $h^*$,
and  we query the RMQ of $h^*$ for the maximum value associated with the arcs
which are between $g^*h^*$ and $h^*g^{\prime *}$ in the cyclic order
of the arcs around $h^*$ and inside $\Vor^*(v)$. We pick the largest among all these candidates as a candidate for the maximum distance as well.

It remains to handle vertices of $\mathcal C^*$ that belong to $T_{\mathcal H}^*$.
For this we collect a set $X$ of arcs while traversing the cycles
in
$\mathcal C^*$ as follows.
For each dual vertex $f^*$ such that the flag $b(f^*)$ is $1$
we put the arc $f^*g^*$ incident to $f^*$ which is not in
$C^*$ in $X$.
From the definition of the flags, the arcs in $X$ are in  $T_{\mathcal H}^*$.
 Note that each arc $f^*g^* \in X$
either penetrates or exits a cycle $C^* \in \mathcal C^*$,
and furthermore there
can be multiple penetrating arcs collected at the same cycle
$C^*$.
For each hole $h$ such that $h^*$ is in some cycle $C^*$ we add to $X$ the arcs of
$T_{\mathcal H}^*$ incident to $h^*$ (and in $\Vor(v)$).

If we consider all the arcs in $X$ on a
root-to-leaf path in $T_{\mathcal H}^*$, they must  alternate between penetrating and exiting, with the first
such arc penetrating, and the last such arc exiting.
Since there are $O(1)$ leaves in $T_{\mathcal H}^*$ and also $O(1)$ collected
arcs that exit, we have that $|X|=O(1)$.
We then use the Euler tour tree data structure for $T_{\mathcal H}^*$ to retrieve the contribution
of each subtree of $T^*$ rooted at the vertex internal to $\Vor^*(v)$ of each penetrating arc of $X$.
Let $v_1^*$ be an endpoint, internal to $\Vor^*(v)$, of a penetrating arc in $X$. 
We identify all the exiting arcs $e^*$ in $X$ which are descendants of $v_1^*$, 
such that the $v_1^*$-to-$e^*$ path in $T_{\mathcal H}^*$ contains no other arc of $X$.
Let $v^*_{2},\ldots,v_j^*\in X$ be the endpoints, internal to
$\Vor^*(v)$, of these exiting arcs. We query the Euler tour data structure for $T_{\mathcal H}^*$ with
$v^*_1,v^*_{2},\ldots,v_j^*$ and take the returned value as another candidate for the maximum distance.
Specifically, we identify the queries to the Euler tour data structure for $T_{\mathcal H}^*$  as follows.
Let $X'$ be the set of endpoints, internal to
$\Vor^*(v)$, of the arcs in $X$.
We define the depth of a vertex $g^* \in X'$ to be the number of vertices in
$X'$ that are proper ancestors of $g^*$ in $T_{\mathcal H}^*$. 
Note that vertices at even depth are endpoints of penetrating arcs, and vertices at odd depths are endpoints of exiting arcs.
For every $i$ and for every $g^*\in X$ of depth $2i$, we query the Euler tour data structure for $T_{\mathcal H}^*$ with $v^*_1 = g^*$, and $v^*_2, v^*_3, \dots, v_j^*$ the descendants of $g^*$ in $T_{\mathcal H}^*$ that belong to $X'$ and have depth $2i+1$. Since $|X| = O(1)$, we perform a constant number of queries, and, can compute the appropriate vertices for each query
 in $O(1)$ time. See Figure~\ref{fig:grid-multi} for an illustration.

 \begin{figure*}[h!]
\begin{center}
\includegraphics[scale=0.8, clip=true, trim = 0mm 20mm 0mm 0mm]{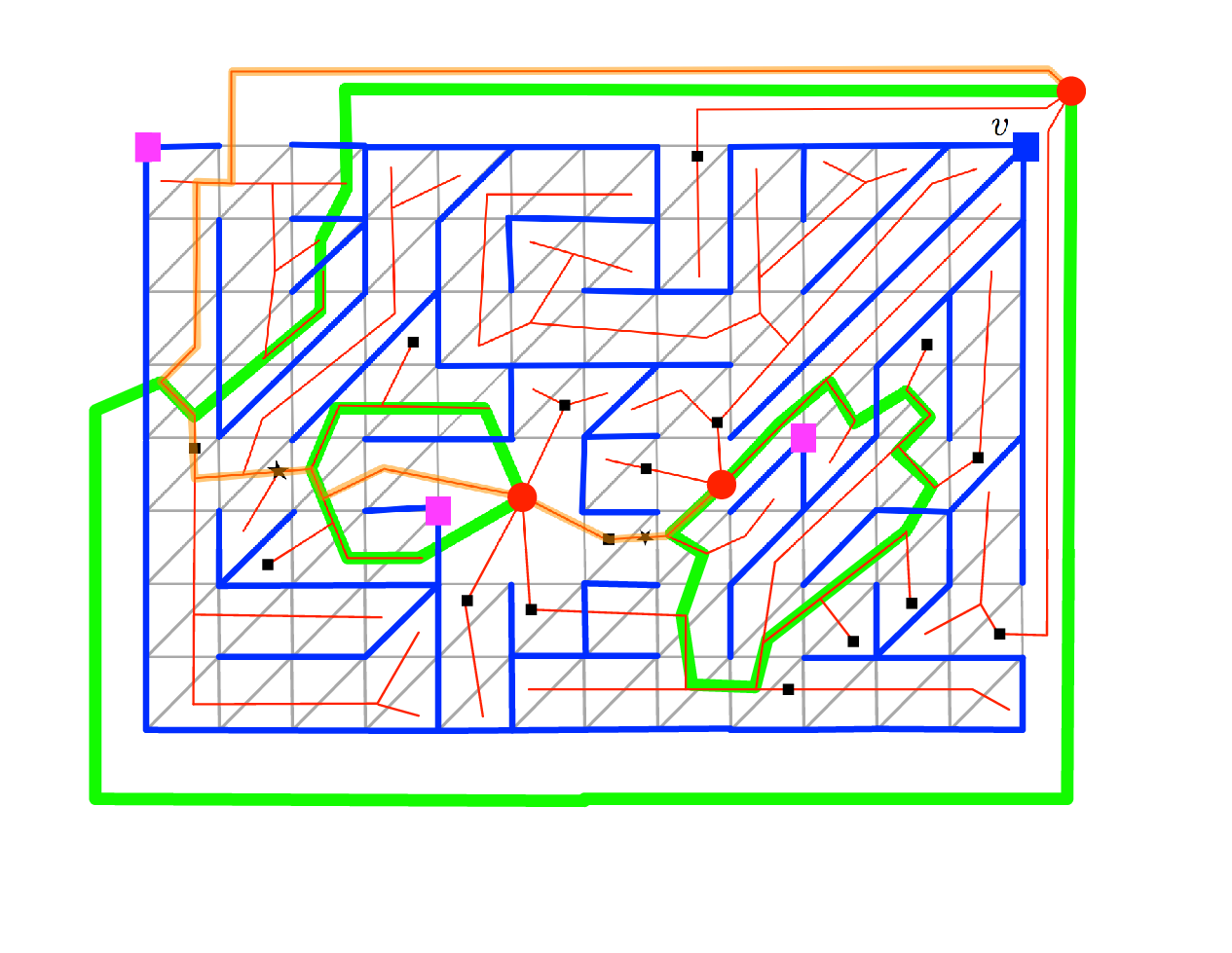}
\caption{A graph with three holes (red circles). Four sites are shown (blue and pink squares), two on the external hole and one on each of the inner holes. The blue site on the upper right corner is $v$. In this example, $\Vor^*(v)$ (green) consists of three cycles. $T$ is shown in blue and $T^*$ in red. The tree $T^*_{\mathcal H}$ (a path in this example) is highlighted in orange. Heads of arcs penetrating $\Vor^*(v)$ are indicated by black squares. The tail of an arc exiting $\Vor^*(v)$ is indicated by a black star. The set $X$ consists the two dual vertices marked by a black square in $T^*_{\mathcal H}$, and the two dual vertices marked by a black star in $T^*_{\mathcal H}$. These are the endpoints of penetrating and exiting arcs that belong to $T^*_{\mathcal H}$. Observe that for any penetrating arc $e^*$ not in $T^*_{\mathcal H}$, the entire subtree of $e^*$ in $T^*$ is in $\Vor(v)$. Candidates from these subtrees of $T^*$ are accounted for as in the single cycle case.
For the penetrating arc $e^*$ in $T^*_{\mathcal H}$ only the portion of the subtree of $e^*$ that is not in the subtree of any of the following exiting arc in $T^*_{\mathcal H}$ is in $\Vor(v)$. Candidates from this portion are accounted for by querying the tree $T^*_{\mathcal H}$.  
\label{fig:grid-multi}}
\end{center}
\end{figure*}

We have thus established the following towards the proof of Theorem~\ref{thm:vor}.
\begin{theorem}
	Consider the settings of Theorem~\ref{thm:vor}. One can preprocess $P$ in $\tilde O(rb^2)$ time so that, for any site $u\in S$, the maximum distance from $u$ to a vertex in $\Vor(u)$ in $\VD(S,\wt)$ can be retrieved in $\tilde O(|\bd \Vor(u)|) $ time, where $|\bd \Vor(u)|$ denotes the complexity of $\Vor(u)$.
\end{theorem}
\else
Due to space constraints the details will appear in the full version.
\fi

\bibliographystyle{abbrv}

\end{document}

%% file: 2verts.pstex_t
\begin{picture}(0,0)%
\includegraphics{2verts.pstex}%
\end{picture}%
\setlength{\unitlength}{3895sp}%
\begingroup\makeatletter\ifx\SetFigFont\undefined%
\gdef\SetFigFont#1#2#3#4#5{%
  \reset@font\fontsize{#1}{#2pt}%
  \fontfamily{#3}\fontseries{#4}\fontshape{#5}%
  \selectfont}%
\fi\endgroup%
\begin{picture}(2295,3008)(1104,-3624)
\put(1696,-1891){\makebox(0,0)[lb]{\smash{{\SetFigFont{11}{13.2}{\rmdefault}{\mddefault}{\updefault}{\color[rgb]{0,0,0}$p_1$}%
}}}}
\put(3098,-1778){\makebox(0,0)[lb]{\smash{{\SetFigFont{11}{13.2}{\rmdefault}{\mddefault}{\updefault}{\color[rgb]{0,0,0}$p_2$}%
}}}}
\put(2004,-2393){\makebox(0,0)[lb]{\smash{{\SetFigFont{11}{13.2}{\rmdefault}{\mddefault}{\updefault}{\color[rgb]{0,0,0}$q_1$}%
}}}}
\put(2625,-2490){\makebox(0,0)[lb]{\smash{{\SetFigFont{11}{13.2}{\rmdefault}{\mddefault}{\updefault}{\color[rgb]{0,0,0}$q_2$}%
}}}}
\put(1284,-2446){\makebox(0,0)[lb]{\smash{{\SetFigFont{11}{13.2}{\rmdefault}{\mddefault}{\updefault}{\color[rgb]{0,0,0}$r_1$}%
}}}}
\put(3384,-2363){\makebox(0,0)[lb]{\smash{{\SetFigFont{11}{13.2}{\rmdefault}{\mddefault}{\updefault}{\color[rgb]{0,0,0}$r_2$}%
}}}}
\put(2109,-751){\makebox(0,0)[lb]{\smash{{\SetFigFont{11}{13.2}{\rmdefault}{\mddefault}{\updefault}{\color[rgb]{0,0,0}$u$}%
}}}}
\put(2401,-2911){\makebox(0,0)[lb]{\smash{{\SetFigFont{11}{13.2}{\rmdefault}{\mddefault}{\updefault}{\color[rgb]{0,0,0}$v$}%
}}}}
\put(2124,-3489){\makebox(0,0)[lb]{\smash{{\SetFigFont{11}{13.2}{\rmdefault}{\mddefault}{\updefault}{\color[rgb]{0,0,0}$w$}%
}}}}
\put(2934,-2169){\makebox(0,0)[lb]{\smash{{\SetFigFont{11}{13.2}{\rmdefault}{\mddefault}{\updefault}{\color[rgb]{0,0,0}$f_2$}%
}}}}
\put(1538,-2214){\makebox(0,0)[lb]{\smash{{\SetFigFont{11}{13.2}{\rmdefault}{\mddefault}{\updefault}{\color[rgb]{0,0,0}$f_1$}%
}}}}
\end{picture}%